\documentclass[preprint, 11pt,a4paper,english]{elsarticle}

\usepackage[scaled=.98]{newtxtext} 
\usepackage{fancyhdr}
\usepackage{natbib}
\setcitestyle{authoryear,round}
\usepackage{dsfont}
\usepackage{mathrsfs}
\usepackage{bm}
\usepackage{hhline}
\usepackage{subcaption}
\usepackage{titlesec}
\usepackage{eurosym}
\usepackage{enumitem}
\usepackage{todonotes}
\usepackage{lipsum}  
\usepackage{graphicx}
\usepackage{comment}
\usepackage[toc, page]{appendix}

\newcommand{\R}{\mathbb{R}} 
\newcommand{\N}{\mathbb{N}} 

\newcommand{\p}{\mathbb{P}} 

\newcommand{\q}{\mathbb{Q}}
\newcommand{\E}{\mathbb{E}}

\newcommand{\bunderline}[1]{\underline{#1\mkern-4mu}\mkern4mu }

\makeatletter

\usepackage{etoolbox}
\DeclareCaptionSubType*{figure} 
\makeatletter
\patchcmd{\chapter}{\if@openright\cleardoublepage\else\clearpage\fi}{}{}{}
\def\ps@pprintTitle{%
	\let\@oddhead\@empty
	\let\@evenhead\@empty
	\def\@oddfoot{}%
	\let\@evenfoot\@oddfoot}
\makeatother

\usepackage{float}
\usepackage{algorithm}
\usepackage{algorithmic}
\usepackage{geometry}
\usepackage{amsmath,amsthm,amssymb, array,nccmath}
\newtheorem{Theorem}{Theorem}[section]

\newtheorem{Assumption}[Theorem]{Assumption}
\newtheorem{Proposition}[Theorem]{Proposition}

\newtheorem{Corollary}[Theorem]{Corollary}
\theoremstyle{definition}

\theoremstyle{remark}
\newtheorem{Remark}[Theorem]{Remark}

\numberwithin{equation}{section}
\numberwithin{figure}{section}
\numberwithin{table}{section}

\titlelabel{\thetitle.\quad}
\makeatother

\usepackage{babel}
\makeatletter
\addto\extrasfrench{%
	\providecommand{\fg}{\ifdim\lastskip>\z@\unskip\fi~\frqq}%
}

\titleformat{\subsection}
{\normalfont\normalsize\itshape\bfseries}{\thesubsection.}{1em}{}
\makeatother
\usepackage{hyperref}
\usepackage[hyperref]{xcolor}

\hypersetup{
	colorlinks=true,
}

\begin{document}
	
	\setcounter{page}{0}
	\begin{frontmatter}
		
		
		
		\title{General Bounds on Functionals of the Lifetime under Life Table Constraints in a Joint Actuarial-Financial Framework}
		

		\author[inst1,inst2]{Jean-Loup Dupret\corref{cor1}}
          \author[inst3]{Edouard Motte}
		\affiliation[inst1]{organization={Department of Mathematics and RiskLab, ETH Zurich},
			addressline={R\"amistrasse 101},
			city={Zurich},
			postcode={8092},
			country={Switzerland}}
            \affiliation[inst2]{organization={Amsterdam School of Economics, University of Amsterdam},
			addressline={Roetersstraat 11},
			city={Amsterdam},
			postcode={1018 WB},
			country={Netherlands}}
            \affiliation[inst3]{organization={LIDAM-ISBA, Université Catholique de Louvain},
			city={Louvain-la-Neuve},
			postcode={1348},
			country={Belgium}}
		\cortext[cor1]{Corresponding author, Email address: \url{j.l.dupret@uva.nl}}
		

	\begin{abstract}
{\small In life insurance, life tables provide survival probabilities at integer ages but leave the distribution of deaths within each year unspecified. This incompleteness is particularly relevant for modern insurance products such as variable annuities, whose payoffs depend jointly on lifetime uncertainty  and financial market performance. The valuation of such contracts must therefore be carried out in a joint actuarial-financial framework, as their values depend not only on the full information about mortality rates but also on the interaction between mortality risk, asset dynamics and embedded guarantees.
Standard approaches impose fractional-age assumptions or mortality rate models,  but the resulting valuations strongly depend on these assumptions, which makes it difficult to generalize them. We hence derive upper and lower bounds of hybrid functionals of the lifetime with respect to mortality rates, which are compatible with the observed life table at integer ages and the given financial market. We derive two sets of results under distinct assumptions. In the first, we assume that each mortality trajectory is almost surely consistent with all the  tabulated  one-year survival probabilities. 
In the second, we consider a relaxed formulation that allows for deviations of the mortality rates while still being consistent in expectation with the given reference life table. These distinct yet complementary formulations characterize the worst- and best-case contract values over all mortality processes that remain compatible with the observed life-table information and the financial market.  The resulting bounds  provide insurers with a practical and robust risk-management tool for assessing the impact of adverse mortality dynamics on variable annuities, quantifying the uncertainty arising from incomplete life-table information, and incorporating mortality model risk into pricing decisions.}
\end{abstract}
		
	\end{frontmatter}
	
	\section{Introduction}
	\hypersetup{
		linkcolor=cyan,
		urlcolor=red,
		citecolor=black,
	}
The two most common approaches to mortality modeling in life insurance rely either on \textbf{fractional age assumptions} or \textbf{mortality rate models}. Each provides a way to derive survival probabilities over continuous time, yet both come with inherent limitations.
\\ \\
\textbf{(i) Fractional age assumptions}. In practice, life insurers and pension funds often rely on life tables, which provide the survival probabilities of individuals at integer ages. However, many actuarial quantities (life annuities, variable annuities, or other functionals of the lifetime) depend on the full distribution of survival times, which requires information about the distribution of deaths between integer ages. To address this, actuaries introduce fractional age assumptions that interpolate survival probabilities between successive integer ages, see Chapter 3.6 of \cite{bowers1997actuarial} and Chapter 3 of \cite{dickson2020actuarial}. The most common assumptions are the \textit{uniform distribution of deaths (UDD)} assuming that deaths occur uniformly over each one-year interval; \textit{the constant force of mortality (CFM)} with an exponential decay of the survival function within each interval;  and \textit{the Balducci assumption} under which the probability of death decreases linearly with time over each year. Each of these assumptions produces a different continuous-time survival distribution consistent with the same life table at integer ages. However, since life tables provide no information about the actual distribution of deaths between integer ages, the resulting actuarial quantities depend heavily on the chosen assumption, which introduces a structural source of model risk that is rarely quantified by insurers.
\\
\\
\textbf{(ii) Mortality rate models}. An alternative approach to fractional age assumptions is to model mortality continuously in time through a mortality rate, or force of mortality, $\mu_x(t)$ at initial age $x$
and to derive the corresponding survival function $S_x(t) = \exp\left(-\int_0^t\mu_{x}(s)\mathrm{d}s\right)$. Deterministic mortality laws, such as the Gompertz–Makeham \citep{gompertz1825nature,makeham1860law} or Heligman–Pollard models \citep{heligman1980age}, model $\mu_x(t)$ as a smooth function of time $t$ and provide analytical tractability for pricing and reserving purposes. More recent developments have introduced stochastic mortality models to capture systematic longevity trends and uncertainty in future mortality evolution. Notable examples include the Lee–Carter models \citep{lee1992modeling, booth2002applying,renshaw2003lee}, the Cairns–Blake–Dowd (CBD) model \citep{cairns2006two, cairns2009quantitative}, and cohort-based extensions such as the Renshaw–Haberman model \citep{renshaw2006cohort, currie2006smoothing} or  multi-population approaches \citep{li2005coherent}, see also \cite{plat2009stochastic}  for a thorough review on stochastic mortality modeling. A related body of literature also integrates stochastic mortality into financial valuation \citep{dahl2004stochastic, biffis2005affine, cairns2006pricing}. All these approaches provide tools for mortality forecasting, valuation and risk management. However, they remain inherently model-based: their application  relies on  parametric and structural assumptions about mortality dynamics. Consequently, computed actuarial functionals of the lifetime depend strongly on these modeling choices and different specifications can lead to substantially different valuations of life-contingent liabilities.
\\
\\
Therefore, to avoid these two sources of model incompleteness, authors have considered fractional-age interpolation and safe-side reserving which include the sum-at-risk approaches \citep{lidstone1905changes, norberg1985lidstone, linnemann1993application, kalashnikov2003sensitivity}; the fractional-age comparisons  via stochastic ordering \citep{hurlimann1990life, willmot1997statistical, jones2000family, frostig2002comparison, barz2012comparison}; and  the worst-case techniques \citep{christiansen2010biometric, christiansen2010first, christiansen2013worst}. However, these approaches remain fundamentally actuarial in both object and method as they seek to compare admissible lifetime distributions, construct conservative first-order mortality bases or derive worst-case actuarial reserves from one-year life-table information. In particular, \cite{barz2012comparison}  derive tight bounds for functionals of future lifetimes under an increasing force of mortality, while \cite{christiansen2010first} construct conservative first-order bases explicitly for safe-side actuarial calculations. \cite{christiansen2013worst} then extend this logic to single- and multiple-decrement settings and mixed survival/death benefits.  What these papers on safe-side reserving and actuarial bounds do not provide is a joint financial-market framework: there are no traded assets, no self-financing investment strategy, no hedging argument and no risk-neutral pricing measure. By contrast, the present paper keeps the same life-table incompleteness problem but shifts it into a market-consistent robust valuation setting, where mortality ambiguity is analyzed jointly with financial dynamics and contract pricing. This in turn leads to the application of stochastic control techniques over mortality processes, which, to the best of our knowledge, is new to the literature and allows for a much broader admissible mortality class than the shape-restricted intensities of \cite{barz2012comparison} or \cite{christiansen2010first, christiansen2013worst}. Related approaches in  mathematical finance derive robust valuation bounds over models consistent with observed option prices \citep{avellaneda1996managing, hobson1998robust,beiglbock2013model, galich}. Our framework follows a similar principle for mortality uncertainty, using life-table survival probabilities as constraints within a given financial-market model.
\\
\\
Specifically, we propose a framework that avoids both sources of model misspecification, namely, the arbitrary choice of fractional age assumption and the structural assumptions of mortality rate models, while being robust with market-consistent valuation. We derive upper and lower bounds for actuarial functionals of the lifetime that are compatible with the observed life table at integer ages. Our approach characterizes all possible mortality rates consistent with these discrete data points, and determines the corresponding optimal values of hybrid actuarial-financial quantities such as variable annuities. We present two complementary formulations: a strict setting in which each mortality path must match almost surely all the one-year survival probabilities given by the reference life table, and a relaxed (more realistic) setting in which this constraint only holds in expectation, allowing for mortality trajectories to deviate  while still being consistent on average with each one-year tabulated survival probabilities. More precisely, for a fixed maturity $T>0$ and a risk-neutral measure $\q$, we aim at solving stochastic control problems of the form
$$\sup_{\mu \in \mathcal{A}}\mathbb{E}^\q\left[F(T,\mu, A, r) \right]
\quad \ \text{and} \quad  \ \inf_{\mu \in \mathcal{A}}\mathbb{E}^\q\left[F(T,\mu, A, r) \right],$$
for some functional $F$ of the lifetime\footnote{Strictly speaking, functional of the mortality rates, but we interchangeably use these terms when it is clear from the context.} $\mu = (\mu_t)_{t \geq 0}$, investment fund $A = (A_t)_{t \geq 0}$, and interest rates $r =(r_t)_{t \geq 0}$. The strict setting of Section \ref{Sec: bounds strict} will define the set of admissible controls $\mathcal{A}$ such as to match almost surely the given  survival probabilities over each integer year from the life table. Since strict pathwise matching can be overly restrictive and unrealistic in practice, we then consider in Section \ref{Sec: bounds_relaxed} a relaxed setting in which the admissible set $\mathcal{A}$ is defined as the class of essentially bounded admissible controls satisfying the survival probability constraints in expectation. To make the corresponding control problems suitable for a dynamic programming approach, we further restrict this set $\mathcal{A}$ to bounded controls from below and above by mortality rates consistent with the life table. We establish asymptotic results showing that, as the set of bounded controls expands,  the solutions to the corresponding restricted problems converge to the solution of the problem over the class of essentially bounded admissible controls. Moreover, this relaxed framework  allows us  to account for plausible deviations from any given stochastic mortality model chosen by the insurer in a way that is consistent with \cite{cairns2006pricing}.
\\
\\
These two distinct (yet complementary) approaches let us derive general bounds for hybrid functionals of the lifetime that are compatible with observed life tables and with arbitrage-free financial valuation, enabling insurers to quantify the worst- and best-case contract prices resulting from deviations of the observed mortality rates from  their mortality assumptions/models.  This framework is therefore not a pricing engine, but instead provides a novel robust and model-free risk management tool to quantify the sensitivity of financial life-contingent valuations to within-year mortality uncertainty. We highlight that \textit{'model-free'} refers here to the mortality component: no explicit mortality model is ever imposed. A financial model is still required for valuation, see Section \ref{sec: fin}, but the bounds derived in Section \ref{Sec: bounds strict} (strict setting) and \ref{Sec: bounds_relaxed} (relaxed setting) directly generalize to more complex financial dynamics. Numerical experiments in Section \ref{sec: numerics} further illustrate and support these findings. 

	\section{General modeling framework} 
 We work on a complete probability space
$(\Omega,\mathcal H,\mathbb P)$, where $\mathbb P$ denotes
the physical probability measure. Let $\mathcal F$ and $\mathcal G$ be sub-$\sigma$-fields of $\mathcal H$ representing financial and mortality uncertainty, respectively, such that $
\mathcal H=\mathcal F\vee\mathcal G
:=\sigma(\mathcal F\cup\mathcal G).$ 
No independence between $\mathcal F$ and $\mathcal G$ is ever assumed. We also consider a financial filtration
$\mathbb F=(\mathcal F_t)_{t\ge0}$ and a mortality-factor
filtration $\mathbb G=(\mathcal G_t)_{t\ge0}$, with
$\mathcal F_t\subseteq\mathcal F$ and
$\mathcal G_t\subseteq\mathcal G$ for every $t\ge0$.
The corresponding joint filtration $\mathbb H=(\mathcal H_t)_{t\ge0}$
is then defined by
\(
\mathcal H_t:=\mathcal F_t\vee\mathcal G_t, \, t\ge0.
\)
The filtration $\mathbb G$ hence describes
the information generated by mortality factors, while the future lifetime variables $T_x$ introduced below are $\mathcal G$-measurable but need not be $\mathbb G$-stopping times. The mortality model introduced in the next section is therefore understood under the restriction $\mathbb P|_{\mathcal G}$ and the financial dynamics are specified in Section~\ref{sec: fin}.
    \subsection{Mortality modeling}\label{sec:mort}
 Let $T_0$ be a positive, $\mathcal{G}$-measurable random variable
representing the future lifetime of a newborn child. For $x>0$, we define $
T_x := \left(T_0-x \, \middle| \,  T_0>x\right)$,
which is then a positive, $\mathcal{G}$-measurable random variable describing the future lifetime of an individual aged $x$. We consider in this paper a general framework where the lifetime is a mixed random variable with both discrete and continuous components. The cumulative distribution function (CDF) and survival function of $T_x$ are given by
\begin{equation}{}_tq_x :=  F_x(t) \,  := \, \mathbb{P}[T_x \leq t ]  \quad \text{and} \quad {}_tp_x := S_x(t) \, := \, \mathbb{P}[T_x > t ]   \,. \label{eq:defproba}
\end{equation}
By definition, \(F_x\) is non-decreasing and right-continuous, with \(F_x(0)=0\) and \(\lim_{t\to\infty}F_x(t)=1\).  Similarly, $S_x(t)$ is non-increasing, right-continuous, and satisfies  $S_x(0) =1$ with $\lim_{t\to \infty}S_x(t) = 0$. The quantity ${}_tp_x$ is then the survival probability of an $x$-year old individual till age $x+t$ , whereas \({}_tq_x\) is  the corresponding probability of dying within the next \(t\) years.  Moreover, for the probability of surviving  to age \(x+s\) conditional on survival to age \(x+t\),  \(s\geq t\), it holds that
\(
{}_{s-t}p_{x+t}
=
\mathbb{P}\!\left(T_x>s\,\middle|\,T_x>t\right).
\) 
 For an absolutely continuous random variable $T_x$, the probability density function (pdf) of $T_x$ is given by 
$f_x(t) = -\mathrm{d} S_x(t)/\mathrm{d}t$ and we define the force of mortality $\mu_{x,c}(t)$ as the conditional density function of $T_x$ at exact age $x+t$ given survival to that age,
$$\mu_{x,c}(t) = f_x(t) /S_x(t) = -\mathrm{d} \ln S_x(t) / \mathrm{d}t \, ,$$ which is a continuous map  $\mu_{x,c} : \R^+ \to \R^+$. Clearly, we have
\begin{equation}
	{}_tp_x = \exp \left(-\int_0^t \mu_{x,c}(s) \mathrm{d}s \right)  \, \vspace{2mm}.\label{proba}
\end{equation} 
On the other hand, if $T_x$ is a discrete random variable taking values $0 <t_0 < t_1 < \ldots \in \R^+$ with associated probability function $f_x(t_j) = \mathbb{P}(T_x=t_j)$ for $j\in \mathbb{N}$,  the probability of  survival is 
$${}_tp_x = \sum_{j \, : \, t_j > t} f_x(t_j) \, .$$ The (discrete) mortality rate $\mu_{x,d}:\R^+ \to [0,1]$ is then defined accordingly as the conditional death probability at time $t_j$ given that the individual has survived to $t_j$,
\begin{equation}
	\mu_{x,d}(t_j) = \mathbb{P}\left(T_x = t_j  \, \mid \, T_x \geq t_j\right) = \frac{f_x(t_j)}{S_x(t_j^{-})} \, ,
\end{equation} where $S_x(t^{-}) = \lim_{s\to t^{-}} S_x(s)$. When it is clear from the context, we simplify the notation of the discrete mortality process $\mu_{x,j} := \mu_{x,d}(t_j)$. Corresponding to \eqref{proba}, the survival probabilities can be written as
\begin{equation}
	{}_tp_x = \prod_{j \, : \, t_j \leq t} (1- \mu_{x,j}) \quad \text{and} \quad f_x(t_j) = \mu_{x,j} \prod_{i \, : \, t_i < t_j} (1- \mu_{x,i}) \, . \label{proba2}
\end{equation}
As in the continuous case, the discrete hazard process $(\mu_{x,j})_{ j \in \mathbb{N}}$ uniquely determines the distribution of the lifetime $T_x$.
\\
\\
 More  generally, the distribution of $T_x$ may have both discrete and continuous components. In this case, the mortality rate can be defined to have continuous component $\mu_{x,c}(t)$ and discrete components $\mu_{x,0}, \mu_{x,1}, \ldots$ at the discrete times $t_0, t_1, \ldots$ Following \cite{kalbfleisch2002statistical}, the survival probability is then
 \begin{equation}
 	{}_tp_x = \exp\left(-\int_0^t \mu_{x,c} (s) \mathrm{d}s \right) \prod_{j  \, : \, t_j \leq t} (1-\mu_{x,j}) \, , \label{surv}
 \end{equation}
and as before, the conditional probability of death over $
[t,t+\mathrm{d}t)$
 is given by the hazard measure:
\begin{equation} \label{hazard_m}
\mathbb{P}[T_x \in[t, t+\mathrm{d}t) \mid T_x \geq t] = \begin{cases}
\mu_{x,j}, & t=t_j, \quad j=0, 1,2, \ldots \\
\mu_{x,c}(t) \, \mathrm{d}t, & \text { otherwise.}
\end{cases} \vspace{2mm}
\end{equation}
Instead of the deterministic setting \eqref{surv}, we allow in this work for stochastic mortality intensity. Specifically, we consider a continuous-time component  $\mu_{x,c} = ( \mu_{x,c}(t))_{t \geq 0}$ and a discrete-time component $\mu_{x,d} = (\mu_{x,d}(\tau_j) )_{j\in \mathbb{N}}$.    Both processes are assumed to be adapted to the mortality-factor filtration  $\mathbb{G}$,  with $(\tau_j)_{j\in\mathbb{N}}$ a strictly increasing sequence of $\mathbb{G}$-stopping times with $0 <\tau_0 < \tau_1 < \cdots $ a.s.\  and the jump size $\mu_{x,j} = \mu_{x,d}(\tau_j) \in [0,1]$  being $\mathcal{G}_{\tau_j}$-measurable for each $j$. We then impose the following conditional survival assumption, which follows from the more general Assumption~\ref{Assumption: cond indep} introduced later in the joint actuarial-financial framework of Section \ref{Sec: bounds strict}. 
\begin{Assumption} \label{ass:1} For every $s\geq 0$, the future lifetime $T_x$ satisfies $\p$-a.s.,
\[ \mathbb{P}(T_x > s \, \big| \, \mathcal{G}_s) = 
\exp\!\left(-\int_0^s\mu_{x,c}(u)\,\mathrm du\right)
\prod_{j:\,\tau_j\le s}(1-\mu_{x,j}).
\]
\end{Assumption} 
By the tower
property, we obtain the following extension of \eqref{surv}:
\begin{align*}
{}_tp_x = \mathbb{P}(T_x>t\mid\mathcal{G}_0)
= \mathbb{E}\!\left[
    \mathbb{P}(T_x>t\mid\mathcal{G}_t)
    \,\middle|\,\mathcal{G}_0
\right]
= \mathbb{E}\!\left[
    \exp\!\left(-\int_0^t\mu_{x,c}(s)\,\mathrm{d}s\right)
    \prod_{j:\,\tau_j\leq t}(1-\mu_{x,j})
    \,\middle|\,\mathcal{G}_0
\right].
\end{align*}
Assuming $\mathcal{G}_0$ trivial, this stochastic setting coincides with the unconditional
survival probability \eqref{eq:defproba}. Similarly, for $0\leq t\leq T$, consider the probability ${}_{T-t}p_{x+t}$ of surviving
to age $x+T$, conditional on survival to age $x+t$ and\hspace{-0.2mm} mortality\hspace{-0.2mm} rate\hspace{-0.2mm} information\hspace{-0.2mm}
$\mathcal{G}_t$.\hspace{-0.2mm} Conditional\hspace{-0.2mm} Bayes' formula
and  tower property yield
{\small \begin{align}
	{}_{T-t}p_{x+t} = \mathbb{P}(T_x > T \, \big| \, T_x > t, \,  \mathcal{G}_t) &= \frac{\mathbb{E}\left[\p(T_x > T \, \big| \, \mathcal{G}_T) \, \big| \, \mathcal{G}_t \right]}{\mathbb{P}(T_x > t \, \big| \, \mathcal{G}_t) } \nonumber
    \\
   & = \mathbb{E}\left[ \frac{\p(T_x > T \, \big| \, \mathcal{G}_T)}{\mathbb{P}(T_x > t \, \big| \, \mathcal{G}_t)} \, \Bigg| \, \mathcal{G}_t \right] \nonumber \\
    &=  \mathbb{E} \hspace{-0.5mm}\left[ \exp\left(-\int_t^T \mu_{x,c} (s) \mathrm{d}s \right) \hspace{-0.5mm} \prod_{j \, : \, t < \tau_j \leq T} (1-\mu_{x,j}) \ \Big| \ \mathcal{G}_t \right] . \label{surv3}
\end{align}}Finally, we assume that we are given some life table which is described as follows. Let us have a group of $l_0$ new born children and $T_{0,i}$ the random future lifetime of child $i=1,\ldots, l_0$.  Then, for $\mathds{1}_{(x,\infty)}(T_{0,i})$ indicating the survival of the new born child number $i$ to age $x$,
$L(x) := \sum_{i=1}^{l_0} \mathds{1}_{(x,\infty)}(T_{0,i})$ is a random variable denoting the number of children alive at age $x$. With ${}_xp_0 := \mathbb{P}[\mathds{1}_{(x,\infty)}(T_{0,i})=1]$ for every $i=1, \ldots, l_0$, the expectation of $L(x)$ is given by
$$l_x := \mathbb{E}[L(x)] = \mathbb{E}\left[\sum_{i=1}^{l_0} \mathds{1}_{(x,\infty)}(T_{0,i}) \right] = l_0 \cdot {}_xp_0 \, ,$$
which is precisely the quantity reported in the life table. More generally, we have for any integer $j \in \mathbb{N}$,
\begin{equation}
	 l_{x+j} = l_0 \cdot {}_{x+j}p_0 = l_0 \cdot {}_xp_0 \cdot  {}_{j}p_x =  l_x \cdot {}_{j}p_x \, . \label{lifetable}
\end{equation}
Therefore,  we can equivalently consider that the given life table provides us with ${}_{j}p_x = l_{x+j}/ l_{x}$ for $j\in \mathbb{N}$, and thus only specifies the survival function of $\lceil T_x \rceil$. However, many actuarial products are functionals of the lifetime, and require full information on the survival function of $T_x$, \textit{i.e.}\  on ${}_tp_x$ for all $t \geq 0$. The insurer must then postulate a mortality rate model (Gompertz-Makeham, Lee-Carter, CBD, etc.) or supplement
the observed life tables with one of the following three standard
fractional-age assumptions. Given a fixed $x=0,1,\ldots$ and\footnote{We restrict $t\in[0,1)$ since w.l.o.g.\ when $t \geq 1$, we can always rewrite ${}_tp_x = {}_{\lfloor t \rfloor}p_{x} \cdot {}_{t'}p_{x+\lfloor t \rfloor}$ with $t' = t - \lfloor t \rfloor \in [0,1)$. } $t\in[0,1)$, we have
\begin{itemize}
	\item Uniform distribution of deaths (UDD): 
	\begin{equation} \label{UDD}
    {}_tp_x = 1-t\cdot (1-{}_1p_x) \,.
    \end{equation}
	\item Constant force of mortality (CFM):
	\begin{equation} \label{CFM}{}_tp_x = e^{-t \cdot \mu_x(0)} = ({}_1p_x)^t \, .\end{equation}
	\item Balducci approximation:
	\begin{equation} \label{Bald}{}_tp_x = \frac{{}_1p_x}{t +(1-t) \cdot {}_1p_x} \, .
	\end{equation}
\end{itemize}

 We also note that our framework and the bounds derived below are not restricted to static period life tables but can also incorporate a projected mortality surface. These projections may in turn be obtained from a mortality model calibrated to multiple historical life tables.
\subsection{The financial market} \label{sec: fin}
We then consider the following financial framework, while emphasizing that the contract bounds derived in Section \ref{Sec: bounds strict} and \ref{Sec: bounds_relaxed} can be extended directly to other and more complex (Markovian) financial dynamics\footnote{This is straightforward by adapting the contract price $C_t$ in Section \ref{Sec: bounds strict} and the coefficients of the HJB equations in Section \ref{Sec: bounds_relaxed}.}. The pricing measure $\mathbb{Q}$ is defined on the  measurable
space $(\Omega,\mathcal{H})$ and is assumed to be equivalent
to $\mathbb{P}$. 
The financial filtration
$\mathbb{F}:=(\mathcal{F}_t)_{t\geq0}$ is then generated by a vector of independent (standard) $\mathbb{Q}$-Brownian motions $\mathbf{W}^\q=\left(W_t^{1,\q}, W_t^{2,\q}\right)_{t \geq 0}^{\top}$. Under $\mathbb{Q}$, the $\mathbb{F}$-adapted stock-price and
interest-rate processes $(S_t)_{t\geq0}$ and $(r_t)_{t\geq0}$
then satisfy 
\begin{equation}
\mathrm{d} \, \binom{\, S_t \, }{r_t}=\binom{r_t S_t}{\ \kappa\left(\theta(t)-r_t\right) \ } \, \mathrm{d}t+\binom{\ S_t \, \sigma_S \, \mathbf{e}_1^{\top} \ }{\sigma_r \, \mathbf{e}_2^{\top}} \, \Sigma \, \mathrm{d} \mathbf{W}^\q_t \, , \label{dynamics}
\end{equation}
where $\mathbf{e}_1^{\top}=(1,0), \, \mathbf{e}_2^{\top}=(0,1)$ and $ \Sigma$ is the  Cholesky decomposition of the correlation matrix between the stock and the interest rate,
$$
\Sigma=\left(\begin{array}{cc}
	1 & 0 \\
	\rho & \sqrt{1-\rho^2}
\end{array}\right) ,  \ \quad \Sigma\Sigma^\top = \left(\begin{array}{cc}
	1 & \rho \\
	\rho & 1
\end{array}\right) ,
$$
where $\rho \in(-1,1)$ is the stock-rate correlation. The cash account is $B_t=\exp \left(\int_0^t r_s \mathrm{d}s\right)$ and the zero-coupon (ZC) bond $P(t,T) = \mathbb{E}^{\mathbb{Q}}\left[e^{-\int_t^T r_s \mathrm{d}s} \, \big| \, \mathcal{F}_t \right]$.  Then this financial market is complete with respect
to $\mathbb{F}$ and we know that there exists a unique density process $(Z_t)_{t \geq 0}$ defined as 
\begin{equation}\label{eq:unique_financial_density_process}
    Z_t:=\frac{\mathrm{d}\mathbb{Q}}{\mathrm{d}\p}\bigg{|}_{\mathcal{F}_t}, \quad  t \geq 0.
\end{equation}
We again highlight that the mortality rate processes $(\mu_{x,c}, \mu_{x,d})$ (and the random future lifetime $T_x$) are not assumed independent from the financial assets $ \left(r_t\right)_{t \geq 0}$ and $\left(S_t\right)_{t \geq 0}$. 
The parameters are $\kappa, \sigma_S, \sigma_r>0$ and  $\theta(t)$ is a function of time, chosen to match the yield curve
\begin{equation}
	\theta(t)=\frac{1}{\kappa} \frac{\partial}{\partial t} f(0, t)+f(0, t)+\frac{\sigma_r^2}{2 \kappa^2}\left(1-e^{-2 \kappa t}\right) , \label{theta}
\end{equation}
where $f(0, t) = -\partial_t \ln P(0,t)$ is the instantaneous forward rate.
With this specification, the interest rate model can be fitted on a standalone basis. The next standard proposition, see \textit{e.g.}\ \cite{hainaut2024valuation}, reminds the price of a ZC bond in this model.
\begin{Proposition}
	The price at time $t$ of a zero-coupon bond of maturity $T \geq t$, is given by
	$$
	P(t, T)=\exp \left(-B(t, T) r_t + \ln \frac{P(0, T)}{P(0, t)}+B(t, T) f(0, t)-\frac{\sigma_r^2}{4 \kappa}\left(1-e^{-2 \kappa t}\right) B(t, T)^2\right) ,
	$$
	where
	$B(t, T)=\frac{1}{\kappa}\left(1-e^{-\kappa(T-t)}\right)$. Consequently,  we obtain the bond price dynamics under $\q$ as
	\begin{equation}
	\frac{	\mathrm{d}P(t,T)}{P(t,T)} = r_t \mathrm{d}t - B(t,T) \, \sigma_r \, \mathbf{e}_2^\top \Sigma \, \mathrm{d} \mathbf{W}^\q_t \, . \label{dyn2}
	\end{equation}
\end{Proposition}
The insurance benefits considered below are written on a reference investment fund $(A_t)_{t \geq 0}$, composed of cash, stocks, and zero-coupon bonds, with given initial value $A_0>0$.  The quantity $A_0$ determines the scale of the reference fund and the contractual guarantees and is therefore not a fair premium obtained by equating $A_0$ to the inception value of these benefits.
 More precisely, the insurer invests constant percentages $ \pi_S$ in stocks, $\pi_P$ in bonds and $1-\pi_S-\pi_P$ in cash, with $\bm{\pi} = (\pi_S, \pi_P) \in[0,1]^2$ and $\pi_S+\pi_P \leq 1 $.
We assume that the reference fund is self-financing, with no external injections or withdrawals. In particular, the insurance benefit payments specified below are financed separately by the insurer and are not deducted from $A_t$.  Under this modeling assumption, the instantaneous return of the reference fund is given by
\begin{equation}
	\frac{\mathrm{d}A_t}{A_t} = \pi_S \frac{\mathrm{d}S_t}{S_t} + \pi_P \frac{\mathrm{d}P(t,T)}{P(t,T)} +(1-\pi_S - \pi_P) \, r_t \mathrm{d}t \, .
\end{equation}
If we remember Equation \eqref{dynamics} and \eqref{dyn2}, we infer that
\begin{equation}\label{eq:A_t}
\begin{aligned}
	\frac{\mathrm{d} A_t}{A_t} 
	& =r_t \mathrm{d}t+\left(\pi_S \sigma_S \mathbf{e}_1^{\top}-\pi_P \sigma_r B(t,T) \, \mathbf{e}_2^{\top}\right) \Sigma \,  \mathrm{d} \mathbf{W}^\q_t \, .
\end{aligned}
\end{equation}
Let us consider for now the pricing of a financial derivative  with payoff
\begin{equation}
	H(A_T, G_T) = \max\{A_T, G_T\} = G_T + (A_T - G_T)_+ \, .\label{payoff1}
\end{equation}
A typical choice in the following consists in  $G_T = A_0  e^{r_g T}$ for $r_g \geq 0$ a minimum guaranteed return. The price of this contract with maturity $T$ is equal to the expected discounted cash-flow under the risk-neutral measure $\q$,
\begin{equation}
	C_t (T,G_T)= \mathbb{E}^\mathbb{Q}\left[ e^{-\int_t^T r_s \mathrm{d}s} \, H(A_T, G_T) \, \big| \, \mathcal{F}_t \right] .
\end{equation}
The next Proposition follows from \cite{hainaut2024valuation} and gives us the price of this contract.
\begin{Proposition}
	\label{Prop price}
The contract's price $C_t$ for $t\leq T$ is given by 
\begin{equation}
	C_t (T,G_T)= G_T P(t,T) +A_t \Phi(-d_1(t)) - G_T P(t,T)  \Phi(-d_2(t)) \, ,
\end{equation}
where
\begin{align}
	d_1(t) &= d_2(t) - \sigma_Y(t,T) \ , \quad  \ \	d_2(t) = \frac{\ln(\frac{G_TP(t,T)}{A_t})+\frac{1}{2}\sigma^2_Y(t,T)}{\sigma_Y(t,T)} \, ,
	\\
	\sigma^2_Y(t,T) &= \int_t^T \left(B(s,T)\sigma_r \mathbf{e}_2^\top \Sigma + \bm{\pi}^\top \bm{\Sigma}_F(s) \right) \left(B(s,T)\sigma_r \Sigma^\top \mathbf{e}_2 + \bm{\Sigma}_F^\top(s)  \bm{\pi}\right) \mathrm{d}s \, ,  \label{sigg}
	\\
	\bm{\Sigma}_F(s)&=\left(\begin{array}{cc}
		\sigma_S  & 0 \\
		-\rho \sigma_r B(s,T) & -\sqrt{1-\rho^2} \, \sigma_r B\left(s, T\right) 
	\end{array}\right) .
\end{align}
\end{Proposition}
Note that \eqref{sigg} can be obtained in closed-form, see Appendix B and C of \cite{hainaut2024valuation}.

\section{Bounds on variable annuities: almost sure matching of the life tables} \label{Sec: bounds strict}
In this section, we first derive bounds for three variable annuity products with respect to mortality rates, while enforcing strict consistency with the observed life tables. Accordingly, we assume that the total force of mortality available for each integer year $x$ (or equivalently, each one-year survival probability) is given by the life table $(l_x)_{x \in \N}$ in \eqref{lifetable} and we then restrict attention to admissible mortality trajectories that reproduce almost surely  each of these prescribed annual quantities, see constraint \eqref{const} below.
\\
\\
Let us assume that we are given a cohort of individuals aged $x$ at time $t=0$ and work with the joint filtration $\mathbb{H}$ introduced above. We require the two following standard assumptions throughout the rest of the manuscript.

\begin{Assumption} \label{Assumption: cond indep}
     The random variable $T_x$ satisfies $\p$-a.s.\ for $s\geq t\geq 0$,
     \begin{equation}
         \p(T_x > s \mid  \mathcal H_{s,t} ) = \mathds{1}_{\{T_x > t \}} \, e^{-\int_t^{s} \mu_{x,c}(u) \mathrm{d}u } \hspace{-2mm}\prod_{j \, : \,t < \tau_j \leq s} \hspace{-2mm}(1-\mu_{x,j}) \, ,
     \end{equation}
     where we define
\( 
\mathcal H_{s,t}
:=
\mathcal H_s\vee\sigma(\{T_x>t\}) \, .
\)
\end{Assumption}
This property extends (and implies) Assumption \ref{ass:1}, see also \cite{cox1980point}, and has a conditional-independence interpretation
\(
\mathbb P(T_x>s\mid \mathcal H_s \vee\sigma(\{T_x>t\}))
=
\mathbb P(T_x>s\mid \mathcal G_s \vee\sigma(\{T_x>t\}))
\).
Thus, given the mortality information $\mathcal{G}_s$ to time $s$ and survival to \(t\), additional financial information  does not change the survival probability. This still allows the mortality factors to depend on financial conditions.
\begin{Assumption} \label{Assumption: change of measure} Let us introduce the filtration containing observed deaths,
\[
\mathcal D_s:=\sigma(\{T_x\le u\}:0\le u\le s),
\qquad
\widehat{\mathcal H}_s:=\mathcal H_s\vee\mathcal D_s,
\]
We then choose $\mathbb{Q}$ as the minimal martingale measure (MMM) \citep{schweizer1991option,schweizer1995minimal} for the
joint actuarial-financial market and assume that its density
process with respect to $\widehat{\mathbb{H}}$ coincides with the financial density process $(Z_t)_{t \geq 0}$ defined in \eqref{eq:unique_financial_density_process}, namely, it holds $\p$-a.s.,
    \begin{equation}\label{eq:combined_density_process}
       \frac{\mathrm{d}\q}{\mathrm{d}\p}\bigg{|}_{\widehat{\mathcal{H}}_t}=Z_t = \frac{\mathrm{d}\mathbb{Q}}{\mathrm{d}\p}\bigg{|}_{\mathcal{F}_t} , \quad  t \geq 0 .
    \end{equation}
\end{Assumption}
Several arguments support the use of the MMM. In particular, in the valuation and hedging of variable annuities,  \cite{moller1998risk,moller2001risk} and \cite{ceci2017unit} show that, within the framework of local risk-minimization, the fair value of variable annuities can be expressed as the discounted expectation of its payoff under the MMM. Hence, this measure arises naturally in incomplete markets, where perfect replication is not feasible, and it provides the pricing rule consistent with the locally risk-minimizing hedging strategy. The MMM is also mathematically tractable: its construction relies solely on the orthogonal decomposition of the discounted price process. Moreover, under the MMM, no additional mortality risk premium is introduced. More precisely, since the mortality intensity may be correlated with traded assets, its drift is adjusted under the MMM, which induces a change in the dynamic but not in the functional form of the mortality intensity. Although the choice of the martingale measure in Assumption \ref{Assumption: change of measure} appears arbitrary, it is in fact the most natural one in our setting as  we have no prior knowledge regarding the mortality rates.\ Other martingale measures could be considered but would require additional modeling choices, such as the mortality risk process $(\lambda_t)_{t \geq 0}$, with $\lambda_t>-1$ a.s., which determines the pricing mortality under $\q$ as $\mu^\q_t=(1+\lambda_t) \mu_t$ for all $t \geq 0$, as well as assumptions on any other stochastic factors that could influence the dynamics of mortality process.
\\[3mm]
Assumption \ref{Assumption: change of measure} justifies to introduce our general objective functional as the expectation under $\q$ of a given functional $F$ of the lifetime (equivalently, of the mortality rates), account value $(A_t)_{t \geq 0}$ and interest rate $(r_t)_{t \geq 0}$, as follows
\begin{equation}
    V^\mu_t = \mathbb{E}^{\mathbb{Q}}\left[F(T,\mu_{x,c}, \mu_{x,d}, A, r) \mid \mathcal{H}_{t,t} \right], \qquad t \geq 0.
\end{equation}
Our objective is to derive lower and upper bounds for $V^\mu_t$ over a prescribed class of mortality models consistent with a given life table $(l_x)_{x\in \mathbb{N}}$ described by \eqref{lifetable}. To this end, we consider the continuous-time force of mortality $\mu_{x,c} = (\mu_{x,c}(t))_{t \geq 0}$  as a  control process taking values in a given compact set $\mathbb{M} \subset [0, +\infty)$ containing $0$, and we let $\mathcal{M}$ be a family of such continuous $\mathbb{H}$-adapted, $\mathbb{M}$-valued control processes $\mu_{x,c}$. 
In addition, we consider the discrete process $\mu_{x,d}=(\tau_0, \tau_1,  \ldots ; \mu_{x,0}, \mu_{x,1}, \ldots) \in \mathcal{V}$ as a given impulse control, where $\mathcal{V}$ is the set of  admissible impulse controls $\mu_{x,d}$ such that $0 < \tau_0 < \tau_1 < \ldots $ (a.s.) are $\mathbb{H}$-stopping times and $\mu_{x,j} \in [0,1]$ are $\mathcal{H}_{\tau_j}$-measurable for all $j\geq 0$. We then call $\mu := (\mu_{x,c}, \mu_{x,d})$ a combined control.
Finally, we assume that we are given a set $\mathcal{W} \subset \mathcal{M} \times \mathcal{V}$ of admissible combined controls $\mu$, such that $V^\mu_t$ is well-defined for all $t \geq 0$ and such that the life-table constraint \eqref{const} below is satisfied. 
\\
Note that we do not only want the force of mortality to be consistent with  the final survival probability ${}_{T-t}p_{x+t}$, 
but with all intermediate one-year survival probabilities at integer ages between $x+t$ and $x+T$. This can be seen writing for  $0\leq t<T<\infty$,
$$ 
	{}_{T-t}p_{x+t} = {}_{\lceil t\rceil -t}p_{x+t} \times {}_1p_{x+\lceil t\rceil} \times {}_1p_{x+1 +\lceil t\rceil} \times \cdots \times {}_1p_{x+j +\lceil t\rceil} \times \cdots \times  {}_{T - \lfloor T \rfloor}p_{x+\lfloor T \rfloor} \, . $$
 These one-year survival probabilities $ {}_1p_{x+j +\lceil t\rceil}$, $j=0, \ldots, \lfloor T \rfloor-\lceil t\rceil -1$, can then be observed from the given life table \eqref{lifetable}, \textit{i.e.}
 \begin{equation} \label{eq: table}\hat{p}_j :=  {}_1p_{x+j +\lceil t\rceil} = \frac{l_{x+j+\lceil t\rceil +1}}{l_{x+j+\lceil t\rceil}} .\end{equation} We hence obtain
\begin{align} \label{termproba}
		{}_{T-t}p_{x+t} &= 
{}_{\lceil t\rceil-t}p_{x+t}
\left(\prod_{j=0}^{\lfloor T\rfloor-\lceil t\rceil-1}
\hat p_j\right)
{}_{T-\lfloor T\rfloor}p_{x+\lfloor T\rfloor}.
\end{align} \begin{Assumption} \label{ass: extremal}
 We assume that the extremal survival probabilities ${}_{\lceil t\rceil -t}p_{x+t}$ and $ {}_{T - \lfloor T \rfloor}p_{x+\lfloor T \rfloor}$ cannot be controlled through the mortality control $\mu$ and satisfy $
    {}_{\lceil t\rceil -t}p_{x+t}=  {}_{T - \lfloor T \rfloor}p_{x+\lfloor T \rfloor} = 1$.
\end{Assumption}
 This simplifying assumption is not restrictive as alternative mortality specifications (from life tables, fractional age assumptions or from expert knowledge) can always be imposed on ${}_{\lceil t\rceil -t}p_{x+t}$ and $ {}_{T - \lfloor T \rfloor}p_{x+\lfloor T \rfloor}$. Throughout this section, empty products  are
understood to equal one. In particular,
if $\lfloor T\rfloor-\lceil t\rceil \in \{-1, 0\}$,
Assumption~\ref{ass: extremal} implies
${}_{s-t}p_{x+t}=1$ for every $s\in[t,T]$.
In the strict setting of this section, admissible mortality controls $\mu \in \mathcal{W}$ are required to strictly match all the observed $\hat{p}_j$. Specifically, for each $j=0, \ldots, \lfloor T \rfloor-\lceil t\rceil -1$,
\begin{equation}
\exp\left(-\int_{\lceil t\rceil +j }^{\lceil t\rceil + j+1} \mu_{x,c} (s) \mathrm{d}s \right) \hspace{-1.3mm} \prod_{i \, :  \, \lceil t\rceil + j < \tau_i \leq \lceil t\rceil + j+1} \hspace{-4mm} (1-\mu_{x,i}) =	\hat{p}_j \, , \quad \mathbb{P}\text{-a.s.} \label{const}
\end{equation}
The aim of this Section \ref{Sec: bounds strict} is then to characterize the upper bound
\begin{equation*}
	\bar{V}_t = \operatorname*{ess \, sup}_{\mu \in \mathcal{W}}\mathbb{E}^{\mathbb{Q}} \left[ F( T,\mu_{x,c}, \mu_{x,d}, A, r) \mid \mathcal{H}_{t,t} \right] , 
\end{equation*}
and the lower bound
\begin{equation*}
	\bunderline{V}_t = \operatorname*{ess \, inf}_{\mu \in \mathcal{W}}\mathbb{E}^{\mathbb{Q}}  \left[ F(T,\mu_{x,c}, \mu_{x,d}, A, r) \mid \mathcal{H}_{t,t}\right] , 
\end{equation*}
of functionals of the lifetime over all admissible mortality controls. Essential suprema and infima are used in this section because the conditional objectives are random variables defined only up to almost-sure equality. Likewise, equalities and inequalities between random variables are understood almost surely. We then denote, if they exist, the corresponding optimal mortality controls $\bar{\mu}^*,\bunderline{\mu}^* \in \mathcal{W}$  such that these bounds are attained, \textit{i.e.}\ $\bar{V}_t = V^{\bar{\mu}^*}_t$ and $\bunderline{V}_t = V^{\bunderline{\mu}^*}_t$ ($\mathbb{Q}$-a.s.). 
\\
\\
Moreover,  we emphasize that this pathwise constraint is quite strong in practice: it requires almost surely that every realized mortality path matches the annual survival probabilities $\hat{p}_j$ from the life table. While a constraint in expectation would be more natural, a naive implementation of the resulting control problem appears to be ill-posed. We address this issue in Section \ref{Sec: bounds_relaxed}, by appropriately relaxing this constraint \eqref{const} and specifying the corresponding set of admissible (bounded) controls. In the following subsections, we study the (strict) upper and lower bounds of three different types of variable annuities: GMAB, GMIB and GMDB. Proofs are deferred in \ref{sec: appendix}.
\subsection{Guaranteed Minimum Accumulation Benefits}
Guaranteed Minimum Accumulation Benefits (GMAB) are retirement savings products that promise in case of survival, the maximum between investments and a guaranteed capital. The contract is subscribed by an individual aged $x$ and guarantees at  expiry date  $T$ a payout equal to the maximum between the guaranteed capital $ G^A_T=A_0 e^{r_g T}$ and the fund $A_T$, in the event of survival.  The payout in case of survival is then given by \eqref{payoff1}.   Recalling $\mathcal H_{s,t} =
\mathcal H_s\vee\sigma(\{T_x>t\})$, $s\geq t \geq 0$, the fair value of such a policy at time $t$  is denoted by $V^\mu_t$ and given by the expected discounted cash-flows under $\mathbb{Q}$, 
\begin{align}
	V^\mu_t &= \mathbb{E}^{\mathbb{Q}} \left[ e^{-\int_t^T r_s \mathrm{d}s} \, \mathds{1}_{\{T_x > T\}} \, H(A_T,G^A_T) \ \big| \ \mathcal{H}_{t,t} \right] \nonumber
	\\
	&= \mathbb{E}^{\mathbb{Q}}  \left[  e^{-\int_t^T r_s \mathrm{d}s} \, H(A_T,G^A_T) \,  \mathbb{E}^{\mathbb{Q}}  \left[  \mathds{1}_{\{T_x > T\}} \ \big| \ \mathcal{H}_{T,t} \right] \, \Big| \ \mathcal{H}_{t,t} \right] \nonumber
	\\
	&=  \mathds{1}_{\{T_x > t \}} \, \mathbb{E}^{\mathbb{Q}}  \bigg[ \, e^{-\int_t^T (r_s+\mu_{x,c}(s) )\mathrm{d}s} \hspace{-2mm}\prod_{j \, : \,t < \tau_j \leq T} \hspace{-2mm}(1-\mu_{x,j}) \, H(A_T,G^A_T) \ \big| \  \mathcal{H}_{t,t} \, \bigg] . \label{GMAB}
\end{align}
 The second equality follows from the tower property, since
$\mathcal H_{t,t}\subseteq\mathcal H_{T,t}$ and the discounted
financial payoff is $\mathcal H_T$-measurable. The final equality uses the conditional survival formula under $\mathbb Q$, as specified by Assumptions~\ref{Assumption: cond indep}
and~\ref{Assumption: change of measure}. 
For the valuation of an in-force contract, we henceforth work on the event $\{T_x>t\}$, where the survival indicator equals one. The upper and lower GMAB bounds over the admissible mortality controls are therefore
\begin{equation}
	\bar{V}_t = \operatorname*{ess \, sup}_{\mu \in \mathcal{W}}\mathbb{E}^{\mathbb{Q}} \left[ e^{-\int_t^T (r_s+\mu_{x,c}(s) )\mathrm{d}s} \hspace{-2mm}\prod_{j \, : \, t <  \tau_j \leq T} \hspace{-2mm}(1-\mu_{x,j}) \, H(A_T,G^A_T) \ \big| \ \mathcal{H}_{t,t}  \right] , \label{ubound}
\end{equation}
and 
\begin{equation}
	\bunderline{V}_t = \operatorname*{ess \, inf}_{\mu \in \mathcal{W}}\mathbb{E}^{\mathbb{Q}} \left[ e^{-\int_t^T (r_s+\mu_{x,c}(s) )\mathrm{d}s} \hspace{-2mm}\prod_{j \, : \, t < \tau_j \leq T} \hspace{-2mm}(1-\mu_{x,j}) \, H(A_T,G^A_T) \ \big|  \  \mathcal{H}_{t,t} \right] . \label{lbound}
\end{equation}
These GMAB bounds are however trivial. Any force of mortality $\mu\in \mathcal{W}$, thereby satisfying \eqref{const}, indeed leads to the same price $V^\mu_t = V_t$ because the contract's value $V^\mu
_t$ only depends on the total survival probability $	{}_{T-t}p_{x+t} $ given by \eqref{termproba}. In fact, the exact death time of an individual in each one-year period does not change the contract's value as in any case, a deceased policyholder will not receive the terminal payoff $H(A_T,G^A_T)$. In other words, the GMAB price does not depend on the exact path of the mortality process, as summarized in the following proposition.
\begin{Proposition}[GMAB bounds] \label{propgmab}
For the GMAB contract \eqref{GMAB}, each admissible control $\mu \in \mathcal{W}$	is an optimal control $\bar{\mu}^*$ of \eqref{ubound} and $\bunderline{\mu}^*$ of \eqref{lbound}. The GMAB price then satisfies $\q$-a.s.\ on $\{T_x > t\}$, 
\begin{equation}
	\bar{V}_t = V^\mu_t  = {}_{T-t}p_{x+t} \, C_t(T,G^A_T) \, , \label{gmabup}
\end{equation} 
 and, respectively,
\begin{equation}
	\bunderline{V}_t  = V^\mu_t = {}_{T-t}p_{x+t} \, C_t(T,G^A_T)\, .\label{gmabdow}
\end{equation} 
for each $\mu \in \mathcal{W}$ and where ${}_{T-t}p_{x+t}$ is given by \eqref{termproba}.
\end{Proposition}
This Proposition \ref{propgmab} is still valid when considering any form of payoff $H(A_T, G^A_T)$ for the GMAB \eqref{GMAB}, by adjusting the price $C_t(T,G^A_T)$ in \eqref{gmabup}-\eqref{gmabdow}. Even though the bounds on the GMAB are trivial, they lay the foundations for studying the following more complex variable annuities.

\subsection{Guaranteed Minimum Income Benefit}
We now consider the following variant of Guaranteed Minimum Income Benefits (GMIB) where an individual aged $x$ continuously receives an income rate $r_I > 0$ on the initial reference amount $A_0$ as long as the policyholder is alive, plus a bonus rate proportional to the excess value of the fund $A_t$ over the pre-specified guarantee $G^I_t = A_0 e^{r_g t}$.  The instantaneous payment rate  is given by  %
\begin{equation}
	\widetilde{H}(A_t, G^I_t)  =  A_0 r_I + \beta(A_t-G^I_t)_+ \, . \label{payoff2}
\end{equation}
with the bonus-rate parameter $\beta >0$.  Proposition \ref{Prop price} can be straightforwardly adapted to obtain as contract price for $s \geq t$, the sum of a ZC bond and a European call,
\begin{align}
	C_t (s,G^I_s) &= \mathbb{E}^\mathbb{Q}\left[ e^{-\int_t^s r_u \mathrm{d}u} \, \widetilde{H}(A_s, G^I_s) \, \big| \, \mathcal{F}_t \right] \nonumber
	\\
	& =  A_0 r_I P(t,s) + \beta \left( A_t \Phi(-d_1(t)) - G^I_s P(t,s)  \Phi(-d_2(t)) \right) . \label{call}
\end{align}
We first develop the analysis for a finite horizon $T<\infty$, before showing that the infinite-horizon case $T=\infty$ follows as a direct extension. For each $s<T$, the insurer pays $\widetilde H(A_s,G_s^I) \mathrm{d}s$ over an infinitesimal interval $\mathrm{d}s$, provided that the policyholder is still alive. As specified in Section~\ref{sec: fin}, these payments are financed separately and\hspace{-0.2mm} do\hspace{-0.2mm} not constitute withdrawals\hspace{-0.2mm} from\hspace{-0.2mm} the\hspace{-0.2mm} reference\hspace{-0.2mm} fund.\hspace{-0.2mm} Accordingly, $V_t^\mu$ denotes the time-$t$ value of this income stream  under the risk neutral measure $\mathbb{Q}$ for a\hspace{-0.2mm} given\hspace{-0.2mm} initial\hspace{-0.2mm} reference amount $A_0$,
\begin{align}
	V^\mu_t& =  \mathbb{E}^{\mathbb{Q}} \left[ \int_t^T \hspace{-1mm}e^{-\int_t^s r_u \mathrm{d}u} \, \mathds{1}_{\{T_x > s\}} \, \widetilde{H}(A_s,G^I_s)  \, \mathrm{d}s \ \Big| \ \mathcal{H}_{t,t} \right] \hspace{-0.7mm} \nonumber
    \\
    &=   \mathbb{E}^{\mathbb{Q}} \left[ \int_t^T \hspace{-1mm}e^{-\int_t^s r_u \mathrm{d}u} \,  \mathbb{E}^{\mathbb{Q}}  \left[  \mathds{1}_{\{T_x > s\}} \ \big| \ \mathcal{H}_{s,t} \right] \, \widetilde{H}(A_s,G^I_s)  \, \mathrm{d}s  \ \Big| \ \mathcal{H}_{t,t} \right] \nonumber
	\\
	&= \mathds{1}_{\{T_x > t\}} \,  \mathbb{E}^{\mathbb{Q}} \left[ \int_t^T e^{-\int_t^s (r_u+\mu_{x,c}(u)) \mathrm{d}u}  \hspace{-2mm}\prod_{j \, : \, t < \tau_j \leq s} \hspace{-2mm}(1-\mu_{x,j}) \, \widetilde{H}(A_s,G^I_s) \, \mathrm{d}s  \ \Big| \  \mathcal{H}_{t,t}  \right] , \label{GMIB}
\end{align}
under Assumptions \ref{Assumption: cond indep} and \ref{Assumption: change of measure}. For an in-force contract, the upper bound
\begin{equation}
	\bar{V}_t = \operatorname*{ess \, sup}_{\mu \in \mathcal{W}} \mathbb{E}^{\mathbb{Q}} \left[ \int_t^T e^{-\int_t^s (r_u+\mu_{x,c}(u)) \mathrm{d}u}  \hspace{-2mm}\prod_{j \, : \, t < \tau_j \leq s} \hspace{-2mm}(1-\mu_{x,j}) \, \widetilde{H}(A_s,G^I_s) \, \mathrm{d}s \ \Big| \ \mathcal{H}_{t,t} \right] , \label{uboundi}
\end{equation}
and lower bound 
\begin{equation}
	\bunderline{V}_t = \operatorname*{ess \, inf}_{\mu \in \mathcal{W}} \mathbb{E}^{\mathbb{Q}} \left[ \int_t^T e^{-\int_t^s (r_u+\mu_{x,c}(u)) \mathrm{d}u}  \hspace{-2mm}\prod_{j \, : \, t < \tau_j \leq s} \hspace{-2mm}(1-\mu_{x,j}) \, \widetilde{H}(A_s,G^I_s) \, \mathrm{d}s \ \Big| \ \mathcal{H}_{t,t} \right] , \label{lboundi}
\end{equation}
of the  GMIB contract \eqref{GMIB} are again such that the  constraint \eqref{const} from the given life table holds and are derived in the following proposition.
\begin{Proposition}[GMIB bounds] \label{propgmib} 
\noindent\textup{(i) Upper bound.}
An optimal control $\bar{\mu}^* = (\bar{\mu}^*_{x,c}, \bar{\mu}^*_{x,d})$ for the upper bound $\bar V_t$ is given by
\[
\bar\mu^{*}_{x,c}(t)\equiv 0,\qquad
\bar\mu^{*}_{x,d}=\{(\bar\tau_j^*,\bar\mu_{x,j}^*)\}_{j=0,\dots,\lfloor T\rfloor-\lceil t\rceil-1},
\]
with intervention times and jump sizes
\[
\bar\tau_j^*=\lceil t\rceil+j+1,\qquad \bar\mu_{x,j}^*=1-\hat p_j.
\]
 Moreover, the upper bound satisfies $\q$-a.s.\ on $\{T_x > t\}$,
\begin{equation} \label{upper_GMIB}
    \bar V_t=\int_t^T {}_{s-t}\bar p^{*}_{x+t}\, C_t(s,G_s^I)\,\mathrm{d}s,
\end{equation}
where, for $s\in[t,T]$, the induced survival probability is
deterministic and given by
\begin{equation} \label{sup_proba}
    {}_{s-t}\bar p^{*}_{x+t}
= \prod_{\substack{j=0,\ldots,\lfloor T\rfloor-\lceil t\rceil-1:\\
                 \lceil t\rceil+j+1\le s}}
\hat p_j.
\end{equation}
\smallskip
\noindent\textup{(ii) Lower bound.}
The lower bound satisfies $\q$-a.s.\ on $\{T_x > t\}$,
\begin{equation}\label{lower_GMIB}
\bunderline{V}_t
=\int_t^T {}_{s-t}\bunderline{p}^{*}_{x+t}\,
C_t(s,G_s^I)\,\mathrm{d}s,
\end{equation}
where, for $s\in[t,T]$, the deterministic limiting survival
curve is
\begin{equation}\label{inf_proba}
{}_{s-t}\bunderline{p}^{*}_{x+t}
=
\prod_{\substack{j=0,\ldots,\lfloor T\rfloor-\lceil t\rceil-1:\\
                 \lceil t\rceil+j<s}}
\hat p_j.
\end{equation}
This bound is approached by a sequence of admissible controls
$\{\mu^{(n)}\}_{n\ge1}\subset\mathcal W$ defined by
\[
\mu_{x,c}^{(n)}\equiv0,
\qquad
\mu_{x,d}^{(n)}
=\{(\tau_j^{(n)},\mu_{x,j}^{(n)})\}_{j=0,\ldots,\lfloor T\rfloor-\lceil t\rceil-1},
\]
with
\[
\tau_j^{(n)}=\lceil t\rceil+j+\frac1n,
\qquad
\mu_{x,j}^{(n)}=1-\hat p_j.
\]
Their intervention times and values satisfy $\q$-a.s.,
\[
\tau_j^{(n)}\downarrow\lceil t\rceil+j,
\qquad
V_t^{\mu^{(n)}}\downarrow\bunderline{V}_t.
\]

If 
\(
\prod_{j=0}^{\lfloor T\rfloor-\lceil t\rceil-1}\hat p_j<1,
\)
the lower bound is not attained in $\mathcal W$.
Otherwise, when $\lfloor T\rfloor-\lceil t\rceil \in \{-1, 0\}$ or if all
$\hat p_j=1$, the upper and lower bounds coincide, 
\(
\bar V_t=\bunderline{V}_t
=\int_t^T C_t(s,G_s^I)\,\mathrm{d}s \), 
and every admissible control is optimal.
\end{Proposition} 
 However, GMIBs are often used to ensure pension payments and hence do not have a finite maturity. The following corollary extends the previous bounds to the case $T=\infty$. 
\begin{Corollary}
Suppose that,
$$
\int_t^\infty {}_{s-t}\bar p^*_{x+t}\,
C_t(s,G_s^I)\,\mathrm{d}s<\infty.
$$
Then Proposition~\ref{propgmib} extends to the infinite-horizon case \(T=\infty\). In particular, $\q$-a.s.\ on $\{T_x > t\}$,
$$
\bar V_t
=
\int_t^\infty
{}_{s-t}\bar p^*_{x+t}\,
C_t(s,G_s^I)\,\mathrm{d}s, \qquad \text{and} \qquad 
\underline V_t
=
\int_t^\infty
{}_{s-t}\underline p^*_{x+t}\,
C_t(s,G_s^I)\,\mathrm{d}s \, .
$$
\end{Corollary}
\begin{proof}
The result follows by a direct application of Proposition~\ref{propgmib} on the finite horizons
\(T_N=\lceil t\rceil+N\) and  by letting \(N\to\infty\). Nonnegativity of the integrands, monotone convergence and interchangeability of suprema yield the stated upper bound. For the lower bound, the integrability assumption ensures that the tail contribution vanishes uniformly over admissible controls, allowing passage to the limit in the infimum.
\end{proof}
\begin{Remark}
    Under the financial model of Section \(\ref{sec: fin}\), the integrability condition holds automatically whenever the life table has zero survival beyond a finite limiting age, as is common in actuarial practice.
\end{Remark}
Again, the exact  form of the payoff $\widetilde{H}$ is chosen in this section solely for analytical convenience and does not restrict the generality of our analysis as Proposition \ref{propgmib} remains valid for any positive payoff $\widetilde{H}$, provided that the corresponding price $C_t(s, G^I_s)$ is adjusted accordingly. For more general payoffs $\widetilde{H} \in \R$, the existence and uniqueness of each $\tau^*_j$ then depend on the growth and convexity properties of  the $\mathcal{H}_{t,t}$-conditional expectation in \eqref{GMIB}. Naturally, when $\tau^*_j$ is unique, we obtain directly from constraint \eqref{const} with $\mu^*_{x,c} \equiv 0$ that $\mu^*_{x,j} =1- \hat{p}_j$. Even when $\tau^*_j$ is not unique, choosing arbitrarily any such $\tau^*_j$ with $\mu^*_{x,j} =1- \hat{p}_j$ remains sufficient in our setting, since we are anyway only interested in the best- and worst-case contract values $\bar{V}_t$ and $\bunderline{V}_t$, the observed mortality being indeed not controllable in practice.

\subsection{Guaranteed Minimum Death Benefits}
Guaranteed Minimum Death Benefits (GMDB) ensure that if the policyholder dies before the contract’s maturity, their beneficiaries will receive at least a guaranteed minimum amount, even if the investment’s performance is lower. The payout $\widehat{H}$ in case of death at time $s \geq t$  is then defined similarly as in \eqref{payoff1} using $G^D_s = A_0   e^{r_g s}$, 
\begin{equation}
	\widehat{H}(A_s, G^D_s) = \max\{A_s, G^D_s\} = G^D_s + (A_s - G^D_s)_+ \, , \label{payoff3}
\end{equation}
with $G^D_s > 0$, and whose price $C_t(s,G^D_s)$ at time $t\leq s$ is given in Proposition \ref{Prop price}. The fair value of such a policy can be written using \eqref{hazard_m}--\eqref{surv3}, 
{\small \begin{align}
	V^\mu_t &= \mathbb{E}^{\mathbb{Q}} \left[ \int_{(t,T]} \hspace{-1mm} e^{-\int_t^s r_u \mathrm{d}u} \,  \mathbb{E}^{\mathbb{Q}}  \left[  \mathds{1}_{\{T_x \geq s\}} \ \big| \ \mathcal{H}_{s,t} \right] \left(\mu_{x,c}(s) + \sum_{j \in \mathbb{N}} \mu_{x,j} \, \delta(s-\tau_j)\right) \widehat{H}(A_s,G^D_s) \, \mathrm{d}s \ \Big| \ \mathcal{H}_{t,t} \right] \nonumber
	\\
	&\hspace{-4.8mm}= \mathds{1}_{\{T_x > t\}} \,  \mathbb{E}^{\mathbb{Q}} \hspace{-0.7mm} \left[ \hspace{-0.4mm} \int_{(t,T]} \hspace{-2.2mm} e^{-\int_t^s (r_u+\mu_{x,c}(u) )\mathrm{d}u} \hspace{-0.6mm}\left( \prod_{j \, : \, t < \tau_j < s}  \hspace{-2.5mm}(1-\mu_{x,j}) \right) \hspace{-0.85mm} \hspace{-0.5mm}\left(\hspace{-0.5mm}\mu_{x,c}(s) + \sum_{j \in \mathbb{N}}  \mu_{x,j} \, \delta(s-\tau_j) \hspace{-0.75mm}\right) \hspace{-0.85mm} \widehat{H}(A_s,G^D_s)  \, \mathrm{d}s \,  \Big| \, \mathcal{H}_{t,t} \right] \label{GMDB}
\end{align}}where $\delta(\cdot)$ is the Dirac measure at zero; see \eqref{eq:gmdb_allocation_identity} for an alternative formulation. We again aim at deriving the upper bound $\bar{V}_t$ and lower bound $\bunderline{V}_t$ of an in-force contract subject to the constraint \eqref{const}.
\begin{Proposition} \label{propgmdb}
   Let $\mathcal T^{\mathbb F}_{[a,b]}$ be the set of
$\mathbb F$-stopping times taking values in $[a,b]$ and $w_j :=  (1-\hat p_j)\prod_{i<j}\hat p_i$. \\
\smallskip
\noindent\textup{(i) Upper bound.}
For each $j=0,\ldots,\lfloor T\rfloor-\lceil t\rceil-1$, define 
\begin{equation} \label{eq:stopp_sup}
\overline C_{t,j}
:=\operatorname*{ess\,sup}_{\tau\in
     \mathcal T^{\mathbb F}_{\large[\lceil t\rceil+j, \lceil t\rceil+j+1\large]}} \hspace{-2mm}
  \mathbb E^{\mathbb Q}[e^{-\int_t^\tau  r_u\,\mathrm du} \, 
       \widehat H(A_\tau ,G_\tau^D)\mid\mathcal F_t],
\end{equation}
This closed-interval financial stopping problem admits an optimal stopping time, denoted $\bar\tau_j^{*, \mathbb{F}}$, and 
\begin{equation}\label{lbound2}
\bar V_t 
=\sum_{j=0}^{\lfloor T\rfloor-\lceil t\rceil-1}\hspace{-1mm}\overline C_{t,j}\,
  (1-\hat p_j)\prod_{i<j}\hat p_i \, , \qquad \q\text{-a.s.} 
\end{equation}
If the optimal stopping times can be chosen such that
$\bar\tau_j^{*,\mathbb F}>\lceil t\rceil+j$ a.s.\ for every $j$
with $w_j>0$, the upper bound is attained by the control
\[
\bar\mu_{x,c}^*\equiv0,
\qquad
\bar\mu_{x,d}^*
=\big\{(\bar\tau_j^*,1-\hat p_j)\big\}_{j=0}^{\lfloor T\rfloor-\lceil t\rceil-1},
\]
where $\bar\tau_j^*=\bar\tau_j^{*,\mathbb F}$.
Otherwise, the upper bound is not attained in $\mathcal{W}, $  but can be approximated by  a sequence of admissible  controls $\{\bar\mu^{(n)}\}_{n\ge 1}\subset\mathcal{W}$ given  by  $\bar\mu_{x,c}^{(n)}\equiv0$ and $\bar\mu_{x,d}^{(n)}$ such that
\[
\bar\tau_j^{(n)}:=\bar\tau_j^{*,\mathbb{F}}\vee(\lceil t\rceil+j+1/n),
\qquad \bar\mu_{x,j}^{(n)}:=1-\hat p_j,
\qquad n\ge1 .
\]
Their intervention times satisfy
$\bar\tau_j^{(n)}\downarrow\bar\tau_j^{*,\mathbb F}$,
and their values satisfy $V_t^{\bar\mu^{(n)}}\to \bar V_t$,  $\q$-a.s.
\\
\noindent\textup{(ii) Lower bound.} Similarly, define
\begin{equation} \label{eq:stopp_inf}
\underline C_{t,j}
:=\operatorname*{ess\,inf}_{\tau\in
     \mathcal T^{\mathbb F}_{[\lceil t\rceil+j,\lceil t\rceil+j+1]}} \hspace{-2mm}
\mathbb E^{\mathbb Q}\!\left[
e^{-\int_t^\tau r_u\,\mathrm du}
\widehat H(A_\tau,G_\tau^D)
\,\middle|\,\mathcal F_t\right].
\end{equation}
This closed-interval financial stopping problem admits an optimal stopping
time, denoted by $\underline\tau_j^{*,\mathbb F}$, and
\begin{equation}\label{ubound2}
\underline V_t
=\sum_{j=0}^{\lfloor T\rfloor-\lceil t\rceil-1}\underline C_{t,j}\,
(1-\hat p_j)\prod_{i<j}\hat p_i \, , \qquad \q\text{-a.s.}
\end{equation}

If the optimal stopping times can be chosen such that
$\underline\tau_j^{*,\mathbb F}>\lceil t\rceil+j$ a.s.\ for every
$j$ with $w_j>0$, the lower bound is attained by the control
\[
\underline\mu_{x,c}^*\equiv0,
\qquad
\underline\mu_{x,d}^*
=\big\{(\underline\tau_j^*,1-\hat p_j)\big\}_{j=0}^{\lfloor T\rfloor-\lceil t\rceil-1},
\]
where $\underline\tau_j^*=\underline\tau_j^{*,\mathbb F}$.
Otherwise, the lower bound is not attained in $\mathcal{W}, $  but can be approximated by  a sequence of admissible controls $\{\bunderline\mu^{(n)}\}_{n\ge 1}\subset\mathcal{W}$ given  by $\underline\mu_{x,c}^{(n)}\equiv0$ and $\bunderline\mu_{x,d}^{(n)}$ such that
\[
\underline\tau_j^{(n)}
:=\underline\tau_j^{*,\mathbb F}\vee(\lceil t\rceil+j+1/n),
\qquad
\underline\mu_{x,j}^{(n)}:=1-\hat p_j,
\quad n\ge1.
\]
They satisfy
$\underline\tau_j^{(n)}\downarrow
\underline\tau_j^{*,\mathbb F}$ and
\(
V_t^{\underline\mu^{(n)}}\rightarrow\underline V_t
\),  $\q$-a.s.
\end{Proposition}
This case differs from the  GMIB \eqref{GMIB}, where the policyholder continuously receives benefits while alive. For the GMDB, we optimize in \eqref{eq:stopp_sup}--\eqref{eq:stopp_inf} the expected discounted lump-sum benefit paid at the time of death 
over admissible financial stopping times $\tau$. We can then use standard least-squares Monte Carlo techniques to numerically solve the (purely) financial optimal stopping problems  \eqref{eq:stopp_sup} and \eqref{eq:stopp_inf} as it is equivalent to the numerical pricing of $\lfloor T \rfloor - \lceil t \rceil$ American options; see  \cite{longstaff2001valuing} and Section \ref{Sec: numerics_strict}. In the case this optimal intervention time $\tau^*_j$ solving \eqref{eq:stopp_sup}--\eqref{eq:stopp_inf} is not unique, we can then again simply  choose arbitrarily  one such $\tau^*_j$  as we are only interested in practice in the bounds $\bar{V}_t$ (or $\bunderline{V}_t$), and not the mortality rate $\mu^*$ itself.
\\[2mm]
As mentioned, the strict matching condition \eqref{const} with the given life table $\eqref{lifetable}$ is a  strong setting, since in practice the observed mortality over each one-year period $j=0, \ldots, \lfloor T \rfloor-\lceil t\rceil -1$ will deviate from the tabulated $\hat{p}_j$.  Wide or extreme bounds in this setting should then be interpreted primarily as an indication of information missing from the life table, but not necessarily as plausible mortality scenarios.
Section \ref{Sec: bounds_relaxed} now relaxes the set of admissible controls and considers a more general framework, where this a.s.\ matching of the life table is replaced by a constraint in expectation. This allows us to consider a more realistic framework where mortality rate trajectories can deviate from the life table while still being consistent on average.  

\section{Bounds on variable annuities:  matching life tables in expectation} \label{Sec: bounds_relaxed}

In this section, we aim to derive worst- and best-case prices of variable annuities for mortality rate models that remain consistent with the initial life table but under  a less restrictive constraint that requires only a matching of the initial life table in expectation. Moreover, the mortality intensity is assumed to be a general progressively measurable process with no explicit discrete component. Therefore,  consistency is ensured by assuming that, for a given age $x$, the mortality intensity $\mu_x$ is now a general positive progressively measurable process that satisfies the following constraints from \eqref{surv3}:
\begin{equation}\label{eq: expectation_constraints}
     \E^{\mathbb{P}}\left(e^{-\int_0^{j}\mu_x(s)\mathrm{d}s}\right)={}_j\hat{p}_{x},\quad j=1,\cdots,n.
\end{equation}
where ${}_j\hat{p}_x$ is again a tabulated value defined by,
$${}_j\hat{p}_{x} := \prod_{i=0}^{j-1} \hat{p}_i \, , $$
with $\hat{p}_i :=  {}_1p_{x+i} = l_{x+i +1}/  l_{x+i}$, see \eqref{eq: table}--\eqref{termproba}. In that context, under Assumption \ref{Assumption: change of measure}, we aim to solve the following stochastic control problems associated to the worst- and best-case prices defined as\footnote{In \vspace{-1mm}contrast to the previous section, we adopt a general formulation based on the full combined product, since the dynamic programming approach used to solve the control problem would\vspace{-1mm} be essentially the same if each product (GMAB, GMIB, GMDB) were formulated separately. See Section \ref{Sec: bounds strict} for the definition of each payoff $H$, $\widetilde{H}$, $\widehat{H}$.}
 {\small \begin{align}\label{eq:worst_case_price_match_expectation} 
\sup_{\mu_x \in \mathcal{A}} \mathbb{E}^{\mathbb{Q}} \Bigg[ \,  e^{-\int_0^T (r_u+\mu_x(u))\mathrm{d}u}H(A_T,G^A_T)+\int_0^T e^{-\int_0^s (r_u+\mu_x(u)) \mathrm{d}u}\left(  \widetilde{H}(A_s,G^I_s) +\mu_x(s) \widehat{H}(A_s,G^D_s) \right)\mathrm{d}s \Bigg],
\end{align}}
and 
 {\small \begin{align}\label{eq:best_case_price_match_expectation}
\inf_{\mu_x \in \mathcal{A}} \mathbb{E}^{\mathbb{Q}} \Bigg[ \,  e^{-\int_0^T (r_u+\mu_x(u))\mathrm{d}u}H(A_T,G^A_T)+\int_0^T e^{-\int_0^s (r_u+\mu_x(u)) \mathrm{d}u}\left(  \widetilde{H}(A_s,G^I_s) +\mu_x(s) \widehat{H}(A_s,G^D_s) \right)\mathrm{d}s \Bigg],
\end{align}}with $T=n$, $H,\widetilde{H},\hat{H}$ positive integrable payoffs, and 
\begin{equation} \label{eq:admissible_set_match_survival_proba}
\small
\mathcal{A}
:=
\left\{
\begin{array}{ll}
&\mu_x:[0,T]\times\Omega\rightarrow\mathbb{R}_+
 \text{ prog. meas. process such that:} \\
&\displaystyle
\mathbb{E}^{\mathbb{P}}
\left(
e^{-\int_0^{j}\mu_x(s)\mathrm{d}s}
\right)
={}_j\hat{p}_{x},
\quad j=1,\ldots,n, \text{ and } \displaystyle
\operatorname*{ess\,sup}_{\omega\in\Omega}
\sup_{t\in[0,T]}
|\mu_x(t,\omega)|
<\infty
\end{array}
\right\}.
\end{equation}
To make the stochastic control problems suitable for a dynamic programming approach, we first consider them over a restricted set of admissible controls. For $\delta>0$, we introduce the admissible set of bounded controls $\mathcal{A}^\delta_{bounded}$ defined by
\begin{equation} \label{eq:admissible_set_pathwise}
\small
    \mathcal{A}^\delta_{bounded}:=\mathcal{A}\cap\left\{ \mu_x:[0,T]\times\Omega\rightarrow\mathbb{R}^+~\text{such that for all } t\leq T,~ ({\mu_{x,a}}(t)-\delta)_+\leq \mu_x(t)\leq {\mu_{x,a}}(t)+\delta~a.s.\right\}, 
\end{equation}
where $\mu_{x,a}\in\mathcal{A}$ is an a priori mortality rate consistent with the initial observed life table\footnote{Typically, $\mu_{x,a}$ can be obtained by any fractional age assumption or  mortality model postulated by the insurer. Additional demographic or biological restrictions can also be imposed on the admissible set $\mathcal{A}$  to satisfy the requirements of \cite{cairns2006pricing} (Section 1.5) of a \textit{good} stochastic mortality model. An alternative way to constrain the level of mortality within our existing formulation is to choose the parameter $\delta$ in \eqref{eq:admissible_set_pathwise} so that admissible mortality rates remain within a sufficiently narrow range around a demographically plausible reference intensity $\mu_{x,a}$.}.
In this setting, the stochastic control problems become
{\small \begin{align}\label{eq:wors_case_price_pathwise_reg}
\hspace{-1mm}\sup_{\mu_x \in \mathcal{A}^\delta_{bounded}} \mathbb{E}^{\mathbb{Q}} \Bigg[ \, e^{-\int_0^T (r_u+\mu_x(u))\mathrm{d}u}H(A_T,G_T^A)+\int_0^T \hspace{-1mm}e^{-\int_0^s (r_u+\mu_x(u)) \mathrm{d}u}\left(  \widetilde{H}(A_s,G^I_s) +\mu_x(s) \widehat{H}(A_s,G_s^D)     \right)\mathrm{d}s \Bigg],
\end{align}}and,
{\small \begin{align}\label{eq:best_case_price_pathwise_reg}
\hspace{-1mm}\inf_{\mu_x \in \mathcal{A}^\delta_{bounded}} \mathbb{E}^{\mathbb{Q}} \Bigg[ \, e^{-\int_0^T (r_u+\mu_x(u))\mathrm{d}u}H(A_T,G_T^A)+\int_0^T \hspace{-1mm} e^{-\int_0^s (r_u+\mu_x(u)) \mathrm{d}u}\left(  \widetilde{H}(A_s,G_s^I) +\mu_x(s) \widehat{H}(A_s,G_s^D)     \right)\mathrm{d}s \Bigg].
\end{align}}Since we assume that the controls in the set $\mathcal{A}^\delta_{bounded}$ are bounded below and above by time-dependent functions consistent with the initial life table, we can then directly deduce that the problems \eqref{eq:wors_case_price_pathwise_reg} and \eqref{eq:best_case_price_pathwise_reg} are well-posed. Moreover, Proposition \ref{prop: convergence_delta} below shows that as $\delta \to \infty$, the solutions to the restricted problems converge to the corresponding solutions of the original worst-case and best-case pricing problems. \\

The restricted control problems in \eqref{eq:wors_case_price_pathwise_reg} and \eqref{eq:best_case_price_pathwise_reg} share similarities with the framework of \cite{li2011uncertain}. However, rather than only assuming boundedness of the controlled process, we introduce the expectation constraints \eqref{eq: expectation_constraints} to ensure consistency with observed life tables. In this sense, the problems are also closely connected to \cite{avellaneda1996managing}, who derived bounds on financial derivative prices under volatility uncertainty while maintaining consistency with observed vanilla option prices. For the sake of clarity, we will mainly consider the problem related to the worst-case price, but similar developments can be deduced in the same manner for the problem related to the best case price. Using a dual approach, as done for instance in \cite{motte2025efficient}, we can show that the stochastic control problem \eqref{eq:wors_case_price_pathwise_reg} is equivalent to the following problem (see proof of Theorem \ref{prop: worst_case_price_regu} below):
{\small \begin{equation}\label{eq:wors_case_price_pathwise_reg_lagrangian_N}
\begin{aligned}
\inf_{\boldsymbol{\lambda}\in \mathbb{R}^n}\hspace{-0.5mm}\sup_{\mu_x \in \tilde{\mathcal{A}}^\delta_{bounded}} \hspace{-2mm}\mathbb{E}^{\mathbb{Q}} \Bigg[  &e^{-\int_0^T (r_u+\mu_x(u))\mathrm{d}u}H(A_T,G_T^A) \hspace{-0.5mm}+ \hspace{-0.8mm}\int_0^T \hspace{-1.4mm}e^{-\int_0^s (r_u+\mu_x(u)) \mathrm{d}u} \hspace{-0.5mm}\left(  \widetilde{H}(A_s,G_s^I) \hspace{-0.5mm}+ \hspace{-0.5mm}\mu_x(s) \widehat{H}(A_s,G_s^D) \hspace{-0.5mm}\right)\mathrm{d}s \Bigg]\\
&-\sum_{j=1}^n\lambda_j \left[\E^{\mathbb{P}}\left(e^{-\int_0^{j}\mu_x(s)\mathrm{d}s}\right)-{}_j\hat{p}_{x}\right],
\end{aligned}
\end{equation}}

with $\boldsymbol{\lambda}:=(\lambda_1,\cdots,\lambda_n)\in \mathbb{R}^n$ and, for $\delta>0$,
\begin{equation*} 
\tilde{\mathcal{A}}^\delta_{bounded}:=\left\{ \mu_x:[0,T]\times\Omega\rightarrow\mathbb{R}^+~\text{s.t.\ for all } t\leq T,~ ({\mu_{x,a}}(t)-\delta)_+\leq \mu_x(t)\leq {\mu_{x,a}}(t)+\delta~a.s.\right\}. 
\end{equation*}

Let us show how to solve the inner problem using a dynamic programming approach. To this end, we assume that the unique change of measure process $(Z_t)_{t \geq 0}$ defined by  \eqref{eq:unique_financial_density_process} is given by
\begin{equation}\label{eq:dynamic_Z_MMM} 
    Z_t:=\exp\left(-\int_0^t\left(\frac{\mu_S-r_s}{\sigma_S}\right)\mathrm{d}{W}_s^{1,\p}-\frac{1}{2}\int_0^t\left(\frac{\mu_S-r_s}{\sigma_S}\right)^2\mathrm{d}s\right), \quad  t \geq 0, 
\end{equation}
with $\mu_S$ the drift of the stock under the real measure $\p$ and ${W}^{1,\p}$ a standard Brownian motion defined on $\p$ such that
\begin{equation*}
    {W}^{1,\p}_t:=W^{1,\q}_t - \int_0^t \left(\frac{\mu_S-r_s}{\sigma_S}\right) \mathrm{d}s, \quad t \geq 0, 
\end{equation*}
with ${W}^{1,\q}$ the standard Brownian motion defined on $\q$ from Section \ref{sec: fin}. Moreover, we also assume that the a priori mortality rate $\mu_{x,a}$ is time-dependent and deterministic. 
\begin{Remark}
In our framework, we consider a complete financial market, which guarantees the uniqueness of the risk-neutral measure $\q$. 
Under $\q$, the stock price $(S_t)_{t \geq 0}$ has a drift equal to the short rate $(r_t)_{t \geq 0}$, as required for arbitrage-free pricing.
For the short rate $(r_t)_{t \geq 0}$, the change of measure from the physical measure $\p$ to $\mathbb{Q}$ generally involves a market price of risk $(\lambda^r_t)_{t \geq 0}$, which is calibrated in practice to match observed prices of traded instruments, such as zero-coupon bonds. For simplicity and consistency with Section \ref{Sec: bounds strict}, we assume that $\lambda^r_t = 0$ for all $t \geq 0$ but other choices of $(\lambda^r_t)_{t \geq 0}$ could be made. 
\end{Remark} 

\begin{Remark}
One could alternatively consider a stochastic Markovian reference mortality rate $\mu_{x,a}\in \mathcal{A}$. This would require the corresponding mortality factor to be included as an additional state variable in the dynamic programming formulation. We restrict attention to a time-dependent reference mortality rate in order to keep the state space and notation as simple as possible.
\end{Remark}

Let us define the process $\tilde{Z}=(\tilde{Z}_t)_{t \geq 0}$ as
\begin{equation}\label{eq:def_Z_tilde}
    \tilde{Z}_t:=e^{-\int_0^t r_s \mathrm{d}s} Z_t, \quad t \geq 0, 
\end{equation}
with $Z$ defined by \eqref{eq:dynamic_Z_MMM}. For $\boldsymbol{\lambda}\in \mathbb{R}^n$, the inner problem in \eqref{eq:wors_case_price_pathwise_reg_lagrangian_N} can be rewritten as 

\begin{equation}\label{eq:inner_prob_wors_case_price_pathwise_reg_lagrangian_N}
\begin{aligned}
\sup_{\mu_x \in \tilde{\mathcal{A}}^\delta_{bounded}} \mathbb{E}^{\mathbb{Q}} \Bigg[ \, &e^{-\int_0^T (r_u+\mu_x(u))\mathrm{d}u}H(A_T,G_T^A)-\sum_{j=1}^n\frac{\lambda_j}{\tilde{Z}_{j}}e^{-\int_0^{j}(r_s+\mu_x(s))\mathrm{d}s}\\
& +\int_0^T e^{-\int_0^s (r_u+\mu_x(u)) \mathrm{d}u}\left(  \widetilde{H}(A_s,G_s^I) +\mu_x(s) \widehat{H}(A_s,G_s^D) \right)\mathrm{d}s\Bigg] .
\end{aligned}
\end{equation}

To solve this control problem, for $\boldsymbol{\lambda}\in \mathbb{R}^n$, we introduce the following value function
\begin{align} \label{eq:sup_stoch_problem_path_regu_n_constraint}
\bar{v}(t,a,r,z; \boldsymbol{\lambda})= \hspace{-1mm}\sup_{\mu_x \in \tilde{\mathcal{A}}^\delta_{bounded}(t,T)} \hspace{-1mm}&\mathbb{E}_{t,a,r,z}^{\mathbb{Q}} \Bigg[ \, e^{-\int_t^T (r_u+\mu_x(u))\mathrm{d}u}H(A_T,G_T^A) - \vspace{-4mm}\sum_{\substack{j\in\{1,\ldots,n\}:\\ j\ge t}}\frac{\lambda_j}{\tilde{Z}_{j}}e^{-\int_t^{j}(r_s+\mu_x(s))\mathrm{d}s} \nonumber \\
&\hspace{-0.2mm}+\hspace{-0.5mm}\int_t^T \hspace{-1.2mm}e^{-\int_t^s (r_u+\mu_x(u)) \mathrm{d}u}\hspace{-0.5mm}\left(  \widetilde{H}(A_s,G_s^I) +\mu_x(s) \widehat{H}(A_s,G_s^D) \right)\hspace{-0.3mm}\mathrm{d}s\Bigg], \hspace{-1mm}
\end{align}
where $\E^{\mathbb{Q}}_{t,a,r,z}(.)$ denotes the conditional expectation given $A_t=a$, $r_t=r$, $\tilde{Z}_t=z$, and
\begin{equation}
\small
\tilde{\mathcal{A}}^\delta_{bounded}(t,T):=\left\{ \mu_x:[t,T]\times\Omega\rightarrow\mathbb{R}^+~\text{s.t.\ for all } s\in[t,T],~ ({\mu_{x,a}}(s)-\delta)_+\leq \mu_x(s)\leq {\mu_{x,a}}(s)+\delta~a.s.\right\}. 
\end{equation}
Let us also define $A^\delta(t)$ as
\begin{equation}
    A^\delta(t):=[(\mu_{x,a}(t)-\delta)_+,\mu_{x,a}(t)+\delta], \quad t\leq T.  
\end{equation}Classical results from dynamic programming (see for instance \cite{touzi2012optimal}) suggest that the value function associated to the inner problem solves the following Hamilton-Jacobi-Bellman (HJB) equation, for $t\in ({j-1},j)$, $j=1,...,n$:
\begin{equation}\label{eq:HJB_sup_n}
    \begin{aligned}
    &\sup_{\mu_x\in A^\delta(t)}\left[\mathcal{L}_{t,a,r,z}\bar{v}(t,a,r,z;\boldsymbol{\lambda})+\widetilde{H}(a,G_t^I)+ \mu_x \left(\widehat{H}(a,G_t^D) -\bar{v}(t,a,r,z; \boldsymbol{\lambda})\right) \right]=0,
\end{aligned}
\end{equation}
with boundary conditions
\begin{equation}\label{eq:boundaries_HJB}
    \begin{aligned}
    &\bar{v}(n,a,r,z; \boldsymbol{\lambda})=H(a,G_{T}^A)-\frac{\lambda_n}{z},\\
    &\bar{v}({j}^-,a,r,z; \boldsymbol{\lambda})=\bar{v}({j}^+,a,r,z; \boldsymbol{\lambda})-\frac{\lambda_j}{z},\quad j=1,...,n-1,
    \end{aligned}
\end{equation}

where \(\mathcal L_{t,a,r,z}\) denotes the classical generator associated with the joint process \((A_t,r_t,\widetilde Z_t)\), defined by \eqref{eq:A_t}, \eqref{dynamics}, and \eqref{eq:def_Z_tilde}, respectively. More explicitly,
\[
\mathcal L_{t,a,r,z} v:=\partial_t v+\mathscr A_{a,r,z} v-rv,
\]
where \(\mathscr A_{a,r,z}\) denotes the spatial generator of \((A_t,r_t,\widetilde Z_t)\). Since \eqref{eq:HJB_sup_n} depends linearly on $\mu$, we deduce that, for $j=1,\cdots,n$ and $t\in ({j-1},{j})$,
\begin{align*}
    \sup_{ \mu_x\in A^\delta(t)} &\mu_x \left(\widehat{H}(a,G_t^D)- \bar{v}(t,a,r,z; \boldsymbol{\lambda}) \right) =\\
    &\begin{cases}
\left(\widehat{H}(a,G_t^D)-\bar{v}(t,a,r,z; \boldsymbol{\lambda}) \right)\left(\mu_{x,a}(t)-\delta\right)_+, & \text{if } \widehat{H}(a,G_t^D) < \bar{v}(t,a,r,z; \boldsymbol{\lambda}) ,\\
\left(\widehat{H}(a,G_t^D)-\bar{v}(t,a,r,z; \boldsymbol{\lambda}) \right)\left(\mu_{x,a}(t)+\delta\right),  & \text{if } \widehat{H}(a,G_t^D) \geq \bar{v}(t,a,r,z; \boldsymbol{\lambda}).
\end{cases} 
\end{align*}
Therefore, the maximizer $\bar{\mu}_x^*$ is given by, for $j=1,\cdots,n$ and $t\in ({j-1},{j})$,
\begin{align} 
    &{\bar{\mu}}_x^*(t,a,r,z; \boldsymbol{\lambda})=\begin{cases}
(\mu_{x,a}(t)-\delta)_+, & \text{if } \widehat{H}(a,G_t^D) < \bar{v}(t,a,r,z;\boldsymbol{\lambda}) ,\\
(\mu_{x,a}(t)+\delta),  & \text{if } \widehat{H}(a,G_t^D) \geq \bar{v}(t,a,r,z; \boldsymbol{\lambda}).
\end{cases}\label{eq:max_mu_n}
\end{align}

Plugging this into the HJB equation \eqref{eq:HJB_sup_n}, we deduce that candidate value function satisfies the following nonlinear PDE:
\begin{equation}\label{eq:HJB_sup_bis_n}
    \begin{aligned}
    &\mathcal{L}_{t,a,r,z}\bar{v}(t,a,r,z; \boldsymbol{\lambda}) + \widetilde{H}(a,G_t^I)+ {\bar{\mu}}_x^*(t,a,r,z; \boldsymbol{\lambda})\left(\widehat{H}(a,G_t^D) -\bar{v}(t,a,r,z; \boldsymbol{\lambda}) \right)=0,
\end{aligned}
\end{equation}
with boundary conditions given by \eqref{eq:boundaries_HJB}. \\

Theorem \ref{prop: worst_case_price_regu} shows that under mild regularity assumptions, the candidate value function solution of the HJB equation is effectively the solution to the restricted inner control problem associated to the worst case price. Proposition \ref{prop: best_case_price} derives an analogous result for  the best case problem. Finally, Proposition \ref{prop: convergence_delta} shows that the solutions of the restricted problems converge toward the solutions of the original worst- and best-case problems given by \eqref{eq:worst_case_price_match_expectation}-\eqref{eq:best_case_price_match_expectation}. Proofs are once again deferred in \ref{sec: appendix}.

\begin{Theorem}\label{prop: worst_case_price_regu}
    For any $\boldsymbol{\lambda}\in\mathbb{R}^n$, assume that there exists a classical solution $\bar{v}(t,a,r,z;\boldsymbol{\lambda})$ to \eqref{eq:HJB_sup_bis_n}, which belongs to $C^{1,2,2,2}$ on each open time interval between consecutive constraint dates, satisfies the prescribed terminal and intermediate boundary conditions \eqref{eq:boundaries_HJB}, and such that all stochastic integrals involved are true martingales. Then, $\bar{v}(t,a,r,z; \boldsymbol{\lambda})$ solves the stochastic control problem \eqref{eq:sup_stoch_problem_path_regu_n_constraint}, and the associated optimal mortality rate is given by \eqref{eq:max_mu_n}. In particular, for $\boldsymbol{\lambda}\in \mathbb{R}^n$, $\bar{v}(0,A_0,r_0,Z_0; \boldsymbol{\lambda})$ solves the initial inner problem such that
   {\small \begin{align*}
        &\bar{v}(0,A_0,r_0,Z_0;\boldsymbol{\lambda})\\
        &=\hspace{-0.5mm} \sup_{\mu_x \in \tilde{\mathcal{A}}^\delta_{bounded}} \hspace{-1mm}\mathbb{E}^{\mathbb{Q}} \Bigg[  e^{-\int_0^T (r_u+\mu_x(u))\mathrm{d}u}H(A_T,G_T^A)+ \hspace{-0.7mm}\int_0^T \hspace{-1mm} e^{-\int_0^s (r_u+\mu_x(u)) \mathrm{d}u} \hspace{-0.4mm}\left(  \widetilde{H}(A_s,G_s^I) +\mu_x(s) \widehat{H}(A_s,G_s^D) \right) \hspace{-0.4mm}\mathrm{d}s \Bigg]\\
        &\hspace{20mm}-\sum_{j=1}^n\lambda_j\E^{\mathbb{P}}\left(e^{-\int_0^{j} \mu_x(s)\mathrm{d}s}\right). 
    \end{align*}}Then, the infimum of $\boldsymbol{\lambda}\mapsto\Phi(\boldsymbol{\lambda}):=\left(\bar{v}(0,A_0,r_0,Z_0; \boldsymbol{\lambda})+\sum_{j=1}^n\lambda_j \,{}_j\hat{p}_{x}\right)$ over $\boldsymbol{\lambda}\in \mathbb{R}^n$ is attained at some $\boldsymbol{\lambda}^*\in \mathbb{R}^n$; hence 
    \begin{equation}
        \bar{v}_*(0,A_0,r_0,Z_0):=\bar{v}(0,A_0,r_0,Z_0; \boldsymbol{\lambda}^*)+\sum_{j=1}^n\lambda^*_j \,{}_j\hat{p}_{x}=\min_{\boldsymbol{\lambda}\in \mathbb{R}^n} \left[\bar{v}(0,A_0,r_0,Z_0; \boldsymbol{\lambda})+\sum_{j=1}^n\lambda_j \, {}_j\hat{p}_{x}\right], 
    \end{equation}
    is well-defined. Moreover, if $\Phi$ is differentiable at $\boldsymbol{\lambda}^*$, then $\bar{v}_*$ is the solution of the restricted worst-case price problem 
    \begin{equation}\label{eq: strong_duality}
        \begin{aligned}
        \bar{v}_*(0,A_0,r_0,Z_0)= \sup_{\mu_x \in \mathcal{A}^\delta_{bounded}} \mathbb{E}^{\mathbb{Q}} \Bigg[ &\, e^{-\int_0^T (r_u+\mu_x(u))\mathrm{d}u}H(A_T,G_T^A)\\
        &+\int_0^T e^{-\int_0^s (r_u+\mu_x(u)) \mathrm{d}u}\left(  \widetilde{H}(A_s,G_s^I) +\mu_x(s) \widehat{H}(A_s,G_s^D) \right)\mathrm{d}s \Bigg].
        \end{aligned}
    \end{equation}
\end{Theorem}

\begin{Proposition}\label{prop: best_case_price}
    For any $\bm{\lambda}\in \mathbb{R}^n$, assume that there exists $\bunderline{v}(t,a,r,z; \boldsymbol{\lambda})\in C^{1,2,2,2}$ solution to the following nonlinear PDE, for $t\in ({j-1},j),$ $j=1,...,n$,
    \begin{equation}\label{eq:HJB_inf_bis_n}
    \begin{aligned}
    &\mathcal{L}_{t,a,r,z}\bunderline{v}(t,a,r,z; \boldsymbol{\lambda}) + \widetilde{H}(a,G_t^I)+ {\bunderline{\mu}}_x^*(t,a,r,z; \boldsymbol{\lambda})\left(\widehat{H}(a,G_t^D) -\bunderline{v}(t,a,r,z; \boldsymbol{\lambda}) \right)=0,
    \end{aligned}
    \end{equation}
    with 
    \begin{align}
    &{\bunderline{\mu}}_x^*(t,a,r,z; \boldsymbol{\lambda})=\begin{cases}
    (\mu_{x,a}(t)-\delta)_+, & \text{if } \widehat{H}(a,G_t^D) > \bunderline{v}(t,a,r,z;\boldsymbol{\lambda}) ,\\
    (\mu_{x,a}(t)+\delta),  & \text{if } \widehat{H}(a,G_t^D) \leq \bunderline{v}(t,a,r,z; \boldsymbol{\lambda}),
    \end{cases}\label{eq:max_mu_n_best_case}
    \end{align}
    that satisfies the following boundary conditions
    \begin{equation}
        \begin{aligned}
        &\bunderline{v}(n,a,r,z; \boldsymbol{\lambda})=H(a,G_{T}^A)+\frac{\lambda_n}{z}, \\
        &\bunderline{v}(j^-,a,r,z; \boldsymbol{\lambda})=\bunderline{v}(j^+,a,r,z; \boldsymbol{\lambda})+\frac{\lambda_j}{z},\quad j=1,...,n-1,
        \end{aligned} 
    \end{equation}
    and such that all stochastic integrals involved are true martingales. Then, for $\boldsymbol{\lambda}\in \mathbb{R}^n$, $\bunderline{v}(0,A_0,r_0,Z_0; \boldsymbol{\lambda})$ solves the initial inner problem such that
   {\small \begin{align*}
        &\bunderline{v}(0,A_0,r_0,Z_0;\boldsymbol{\lambda}) \\
        &= \hspace{-0.7mm} \inf_{\mu_x \in \tilde{\mathcal{A}}^\delta_{bounded}} \hspace{-1.3mm}\mathbb{E}^{\mathbb{Q}} \Bigg[  e^{-\int_0^T (r_u+\mu_x(u))\mathrm{d}u}H(A_T,G_T^A)
    +\hspace{-0.5mm}\int_0^T \hspace{-1.1mm} e^{-\int_0^s (r_u+\mu_x(u)) \mathrm{d}u} \hspace{-0.25mm}\left(  \widetilde{H}(A_s,G_s^I) +\mu_x(s) \widehat{H}(A_s,G_s^D) \right) \hspace{-0.25mm}\mathrm{d}s \Bigg]\\
        &\hspace{20mm}+\sum_{j=1}^n\lambda_j\E^{\mathbb{P}}\left(e^{-\int_0^{j}\mu_x(s)\mathrm{d}s}\right). 
    \end{align*}} Then, the supremum of $\boldsymbol{\lambda}\mapsto \Phi(\boldsymbol{\lambda}):=\left(\bunderline{v}(0,A_0,r_0,Z_0; \boldsymbol{\lambda})-\sum_{j=1}^n\lambda_j \, {}_j\hat{p}_{x}\right)$ over $\boldsymbol{\lambda}\in \mathbb{R}^n$ is attained at some $\boldsymbol{\lambda}^*\in \mathbb{R}^n$; hence 
    \begin{equation}
        \bunderline{v}_*(0,A_0,r_0,Z_0):=\bunderline{v}(0,A_0,r_0,Z_0; \boldsymbol{\lambda}^*)-\sum_{j=1}^n\lambda^*_j \,{}_j\hat{p}_{x}=\max_{\boldsymbol{\lambda}\in \R^n} \left[\bunderline{v}(0,A_0,r_0,Z_0; \boldsymbol{\lambda})-\sum_{j=1}^n\lambda_j {}_j\hat{p}_{x}\right], 
    \end{equation}
    is well-defined. Moreover, if $\Phi$ is differentiable at $\boldsymbol{\lambda}^*$, then $\bunderline{v}_*$ is the solution of the restricted best-case price problem 
    \begin{equation}
        \begin{aligned}
        \bunderline{v}_*(0,A_0,r_0,Z_0)= \inf_{\mu_x \in \mathcal{A}^\delta_{bounded}} \mathbb{E}^{\mathbb{Q}} \Bigg[ &\, e^{-\int_0^T (r_u+\mu_x(u))\mathrm{d}u}H(A_T,G_T^A)\\
        &+\int_0^T e^{-\int_0^s (r_u+\mu_x(u)) \mathrm{d}u}\left(  \widetilde{H}(A_s,G_s^I) +\mu_x(s) \widehat{H}(A_s,G_s^D) \right)\mathrm{d}s \Bigg].
        \end{aligned}
    \end{equation}
\end{Proposition}

\begin{proof}
    The proof follows by arguments analogous to the proof of Theorem \ref{prop: worst_case_price_regu}.
\end{proof}

\begin{Proposition}\label{prop: convergence_delta}
    Let us assume that the assumptions of Theorem \ref{prop: worst_case_price_regu} and Proposition \ref{prop: best_case_price} are satisfied. For $\delta>0$, let $ \bar{v}^\delta_*(0,A_0,r_0,Z_0)$ and $\bunderline{v}^\delta_*(0,A_0,r_0,Z_0)$ denote the solution to the worst- and best-case price problems given by \eqref{eq:wors_case_price_pathwise_reg} and \eqref{eq:best_case_price_pathwise_reg}.
    Then, we have that
        \begin{equation}
        \begin{aligned}
         &\sup_{\mu_x \in {\mathcal{A}}} \mathbb{E}^{\mathbb{Q}} \Bigg[ e^{-\int_0^T (r_u+\mu_x(u))\mathrm{d}u}H(A_T,G_T^A)+\hspace{-0.5mm}\int_0^T  \hspace{-1mm}e^{-\int_0^s (r_u+\mu_x(u)) \mathrm{d}u} \hspace{-0.4mm}\left(  \widetilde{H}(A_s,G_s^I) +\mu_x(s) \widehat{H}(A_s,G_s^D) \right) \hspace{-0.4mm}\mathrm{d}s \Bigg]\\
         &\hspace{5mm}=\lim_{\delta \to \infty} \bar{v}^\delta_*(0,A_0,r_0,Z_0) ,
        \end{aligned}
    \end{equation}
    and
    \begin{equation}
        \begin{aligned}
         & \inf_{\mu_x \in {\mathcal{A}}} \mathbb{E}^{\mathbb{Q}} \Bigg[  e^{-\int_0^T (r_u+\mu_x(u))\mathrm{d}u}H(A_T,G_T^A)+\hspace{-0.5mm}\int_0^T \hspace{-1.2mm}e^{-\int_0^s (r_u+\mu_x(u)) \mathrm{d}u}\hspace{-0.4mm}\left(  \widetilde{H}(A_s,G_s^I) +\mu_x(s) \widehat{H}(A_s,G_s^D) \right) \hspace{-0.4mm}\mathrm{d}s \Bigg]\\
         &\hspace{5mm}=\lim_{\delta \to \infty} \bunderline{v}^\delta_*(0,A_0,r_0,Z_0) .
        \end{aligned}
    \end{equation}
\end{Proposition}

\section{Numerical experiments} \label{sec: numerics}
Since our method in Section \ref{Sec: bounds strict} and \ref{Sec: bounds_relaxed} does not rely on any specific mortality modeling assumptions, the resulting bounds provided below are the most general one can obtain using the
available information from the life table.  These bounds are therefore of high practical interest for insurers, as they  naturally provide a range within which the contract price should lie. Any price outside this range is either not consistent with the life tables or with the specified financial market. Finally, these bounds provide a robust risk-management tool for the insurer as they directly quantify the worst-case loss or best-case gain in case of deviation of the observed mortality from the chosen fractional-age/mortality rate assumption. The results are based on the Belgian female life table.
\subsection{Strict bounds} \label{Sec: numerics_strict}
The parameters used in this section are
$A_0=100$, $r_g=0.03$, $r_0=0.01$, $\pi_S=0.6$, $\pi_P=0.2$, $\sigma_S=0.2$, $\sigma_r=0.02$, $\rho=0.3$, $\kappa=0.1$ and $r_I =0.01$, $\beta=0.05$ for the GMIB.
The mean-reversion interest rate level $\theta(t)$ is obtained from \eqref{theta} by specifying the instantaneous forward rate as
\[
f(0,t)= b + c e^{-at} + d a t e^{-at}, 
\]
with yield-curve parameters $a=0.76$, $b=0.015$, $c=-0.005$, and $d=0.02$. The GMAB contract bounds are omitted from this section since they are trivial in the strict setting, see Proposition \ref{propgmab}. We first depict in Figure \ref{fig:survival} the survival probabilities ${}_{t}p_{x}$ at age $x=40$ and $x=100$ for $t\in [0,5]$. Naturally, at lower age $x$, we see that the specific choice of fractional age assumption has almost no impact on the survival probabilities, while we observe higher discrepancies for $x=100$. By contrast, the GMIB-induced optimal survival curves ${}_t\bar{p}^*_{x}$  and ${}_t\bunderline{p}^*_{x}$, given by \eqref{sup_proba} and \eqref{inf_proba}, deviate significantly from  all the fractional age assumptions. 
\begin{figure}[h!]
  \centering
  \hspace{-4mm}\includegraphics[width=0.5\textwidth]{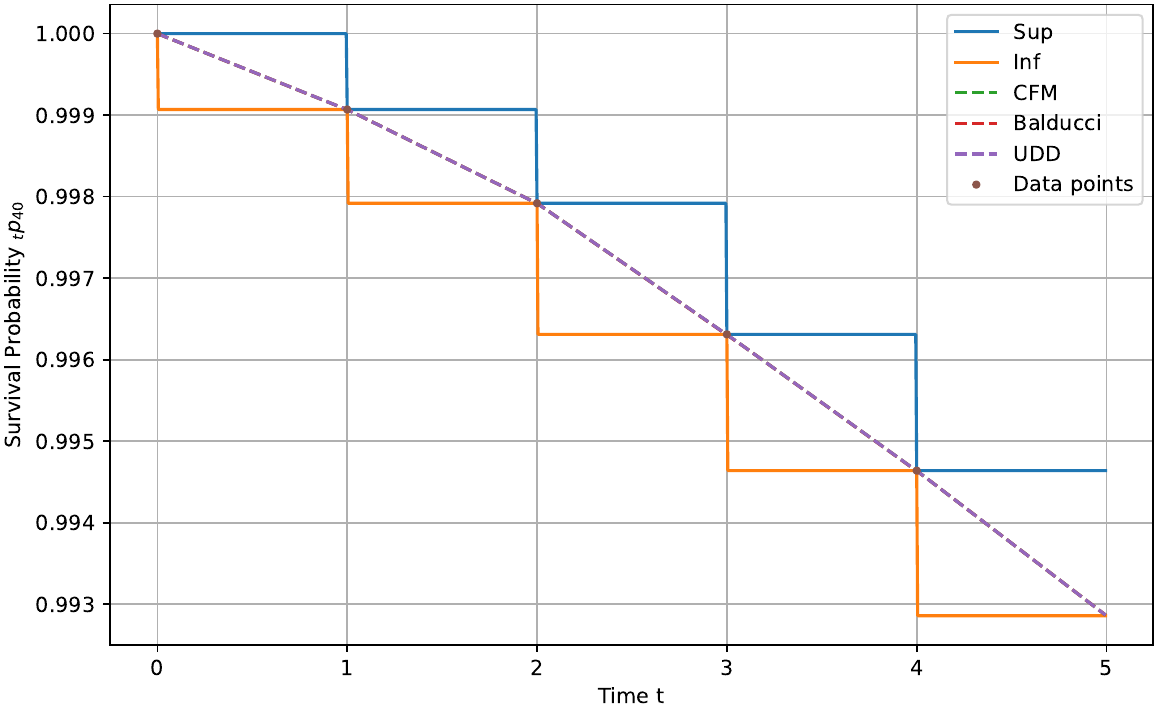} \hspace{-2.4mm}
  \hfill 
  \includegraphics[width=0.5\textwidth]{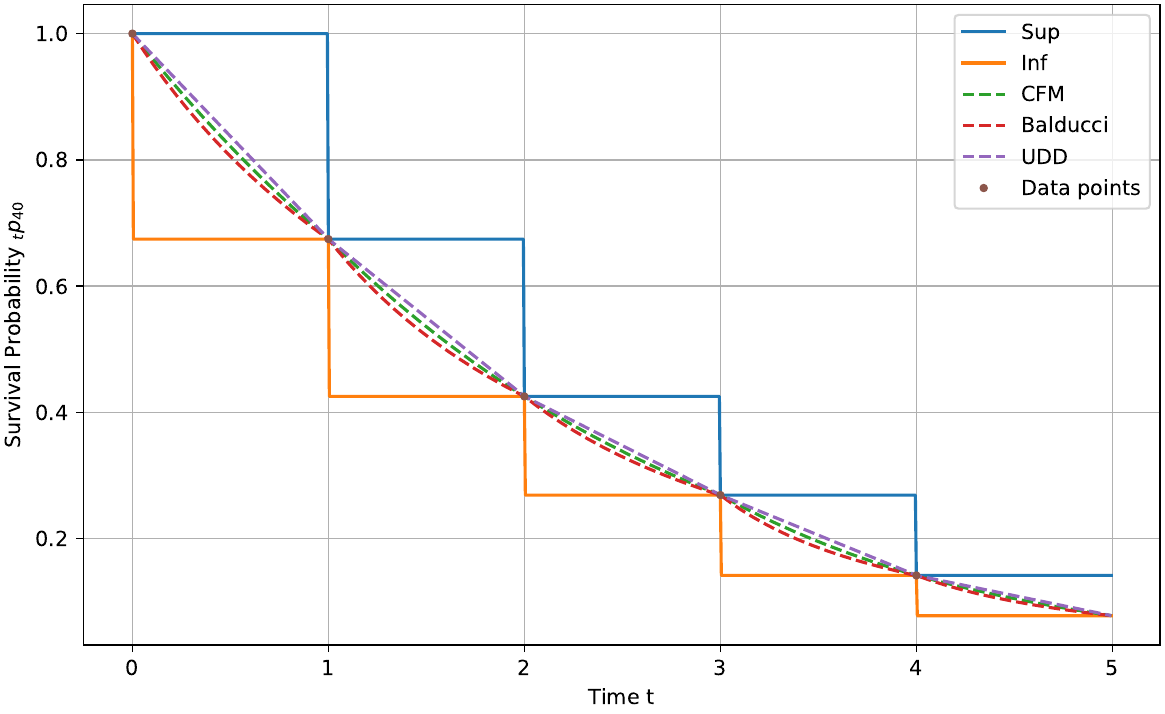} 
    \caption{Comparison of survival probabilities ${}_{t}p_{x}$ at age $x=40$ (left) and $x=100$ (right) over time $t \in [0,5]$, using the fractional age assumptions \eqref{UDD}-\eqref{Bald} and  optimal GMIB  ${}_t\bar{p}^*_{x}$ ('Sup'), ${}_t\bunderline{p}^*_{x}$ ('Inf') from Eqs.\ \eqref{sup_proba}--\eqref{inf_proba}.}
    \label{fig:survival}
\end{figure} \\
We then study the strict bounds of the GMIB contract in Figure \ref{fig:GMIB_strict}.
\begin{figure}[h!]
  \centering
  \hspace{-4mm}\includegraphics[width=0.49\textwidth]{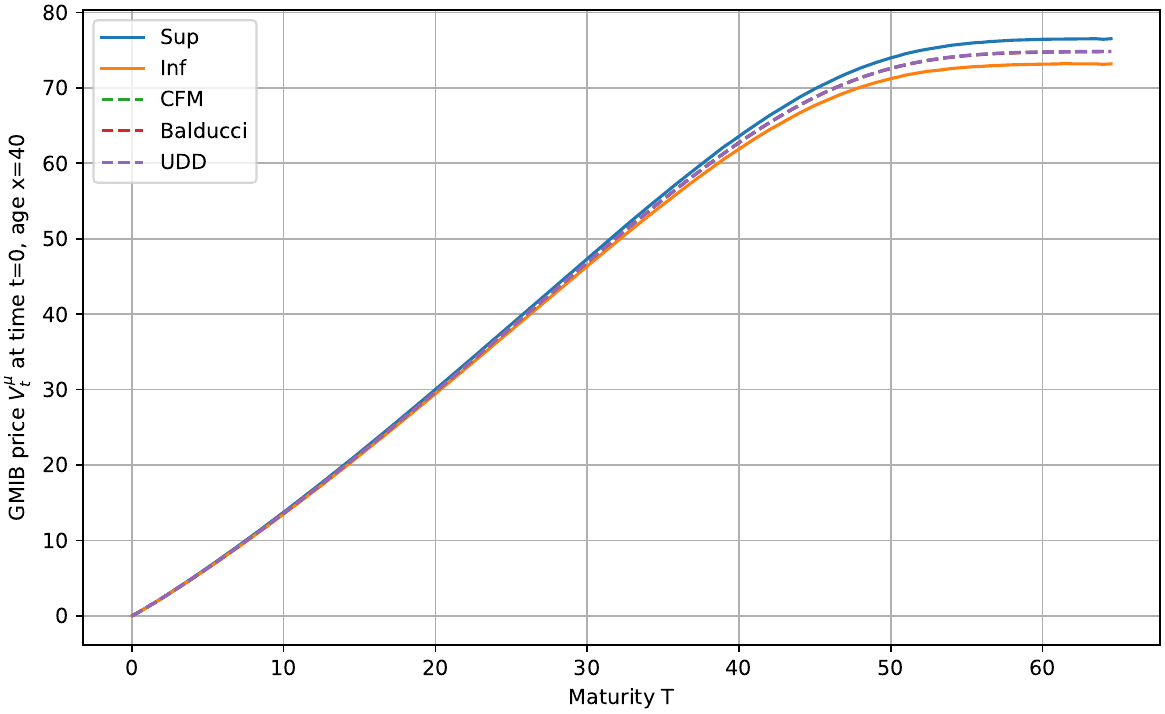} \hspace{-2mm}
  \hfill 
    \includegraphics[width=0.49\textwidth]{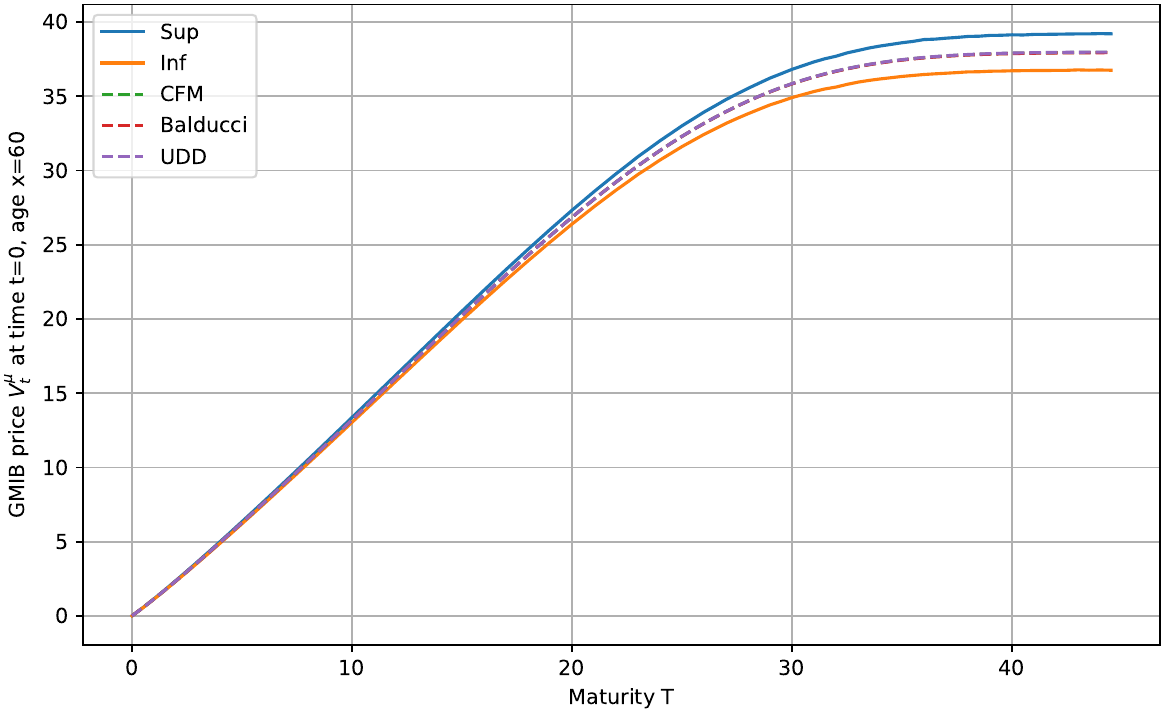}  \\ \includegraphics[width=0.49\textwidth]{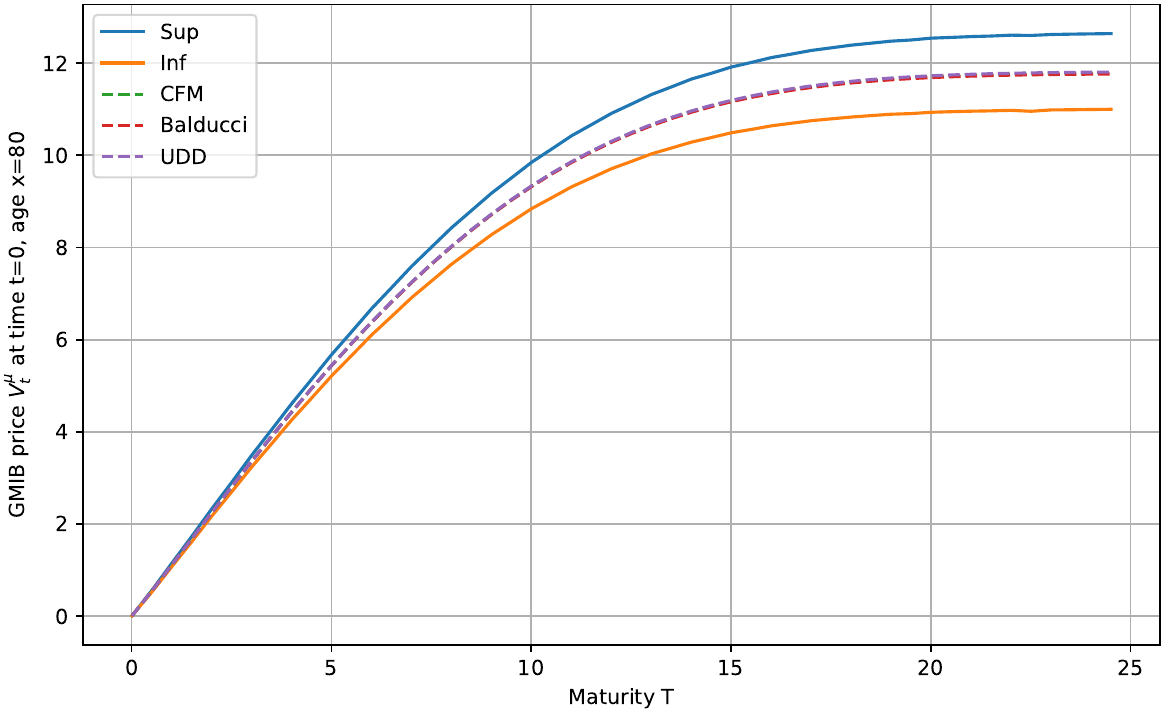}
 
  \caption{GMIB contract price $V^\mu_t$ at $t=0$ for different maturities $T$ with initial age $x=40$ (top left), $x=60$ (top right) and $x=80$ (bottom). Results are shown under the fractional-age assumptions \eqref{UDD}--\eqref{Bald}, together with the GMIB upper bound $\bar{V}_0$ ('Sup') and lower bound $\bunderline{V}_0$ ('Inf') from Eqs.\ \eqref{upper_GMIB}--\eqref{lower_GMIB}.}
  \label{fig:GMIB_strict}
\end{figure} 
  For younger policyholders (smaller initial age $x$), the upper-left plot shows that the insurer’s exposure to deviations of mortality within each one-year period from the chosen fractional-age assumption is moderate, at least for relatively short maturities $T$ (while remaining consistent with the tabulated survival probabilities $\hat{p}_j$'s). In contrast, for older policyholders (bottom plot) or for longer maturities, this mortality risk becomes much more significant as deviations from the fractional-age assumption have a greater impact on GMIB
values. This reflects an increased sensitivity to the timing of deaths within each year, which is not specified by
the annual survival probabilities in the life table.  We also observe that the choice among the three fractional-age assumptions \eqref{UDD}--\eqref{Bald} has virtually no impact on the contract price. \\   Finally, we see that these strict bounds are relatively narrow  due to the pathwise restriction \eqref{const}. Figure~\ref{fig:GMIB_percentage_width} illustrates this by plotting the percentage width of the optimal bounds,
\[
100\,\frac{\bar V_0(T)-\bunderline V_0(T)}{\bunderline V_0(T)},
\]
as a function of maturity $T$. The practical benefit of these robust bounds on valuation can therefore be limited, particularly for younger policyholders ($x$ small) and shorter maturities $T$, for whom standard fractional-age assumptions already yield values close to both bounds. This again motivates the relaxed results in the next section, where consistency with the life table is imposed in expectation.
\begin{figure}[h!]
    \centering
    \includegraphics[width=0.66\textwidth]{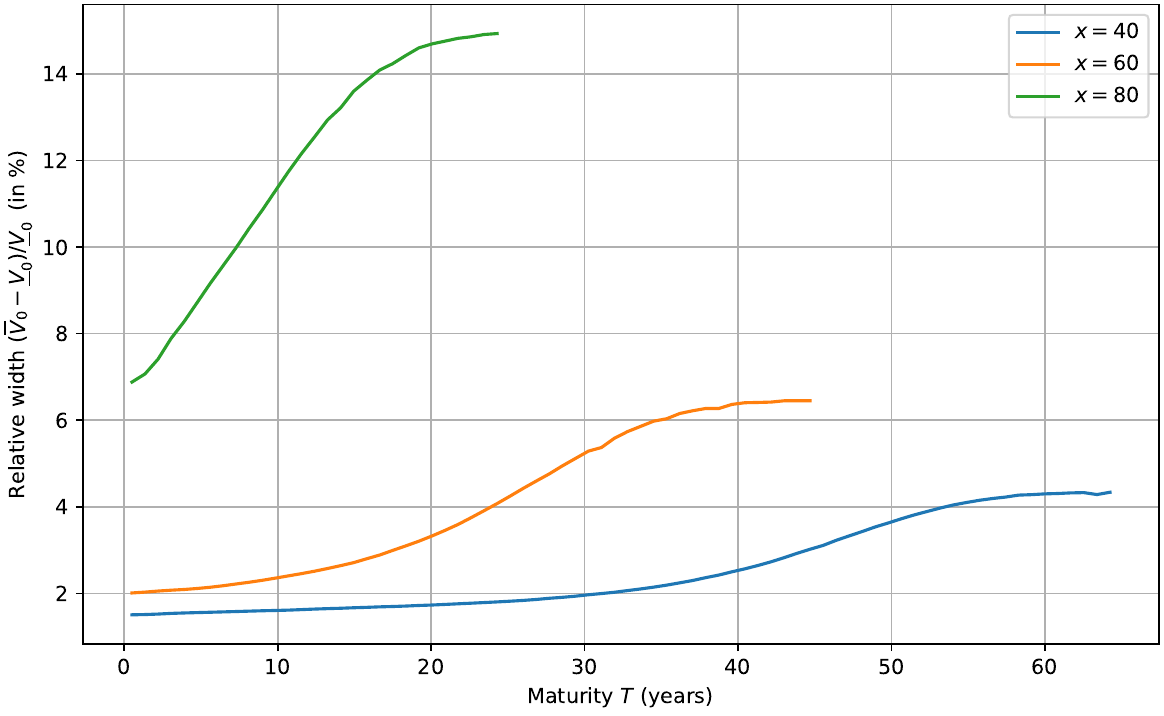}
    \caption{ Percentage width between the upper and lower GMIB  bounds as a function of maturity $T$, for different initial ages \vspace{-2mm}$x$.}
    \label{fig:GMIB_percentage_width}
\end{figure}

We then approximate the GMDB bounds in Proposition \ref{propgmdb} by solving the financial optimal stopping problems \eqref{eq:stopp_sup}--\eqref{eq:stopp_inf} separately on each one-year interval using the least-squares Monte Carlo (LSM) method of \cite{longstaff2001valuing}. Each interval is discretized into 104 equal time steps. Working backwards along this grid, continuation values are estimated by regressing discounted future payoffs on (bivariate) polynomials of total degree 3 in the current state variables $\log(A_s/G_s^D)$ and $r_s$. Stopping occurs when the immediate payoff exceeds the estimated continuation value for the maximization problem, or falls below it for the minimization problem. The LSM policies $\hat\tau^*_j$ are trained on 200,000 simulated paths and evaluated on 400,000 independent paths. The resulting policy values approximate the theoretical GMDB bounds from the inside and, as shown in Figure~\ref{fig:GMDB_strict}, confirm the above conclusions  for the GMIB.
\begin{figure}[h!]
  \centering \hspace{-4mm}\includegraphics[width=0.49\textwidth]{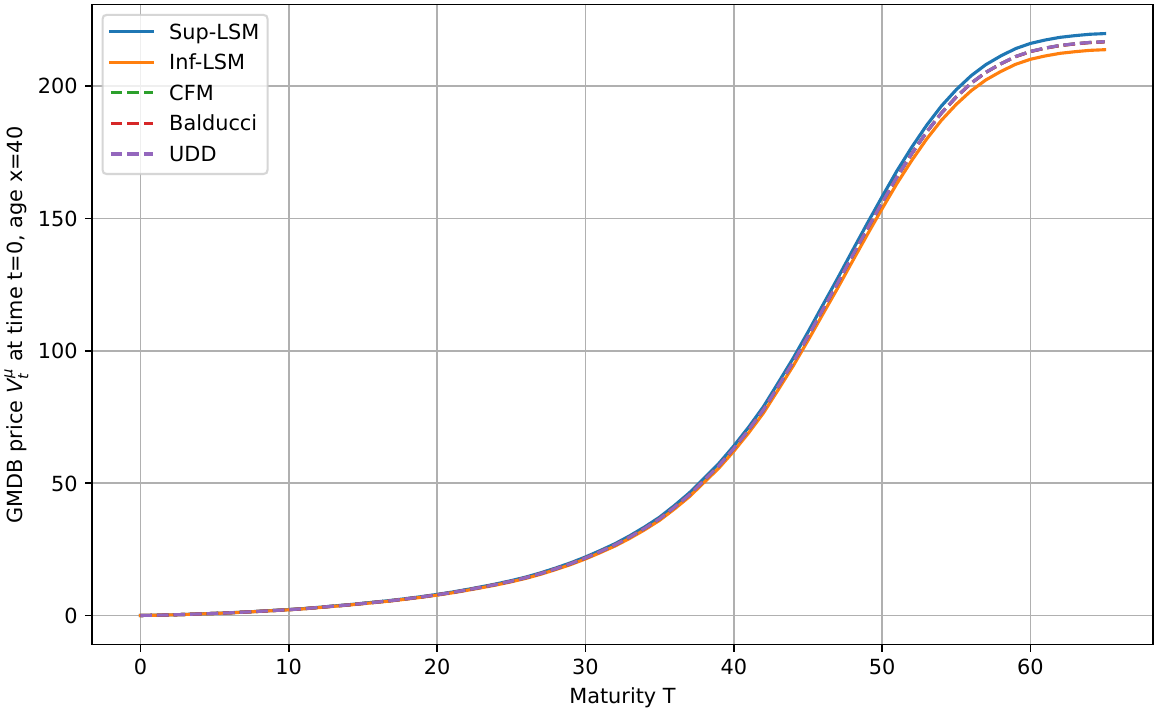} \hspace{-2mm}
  \hfill 
  \includegraphics[width=0.49\textwidth]{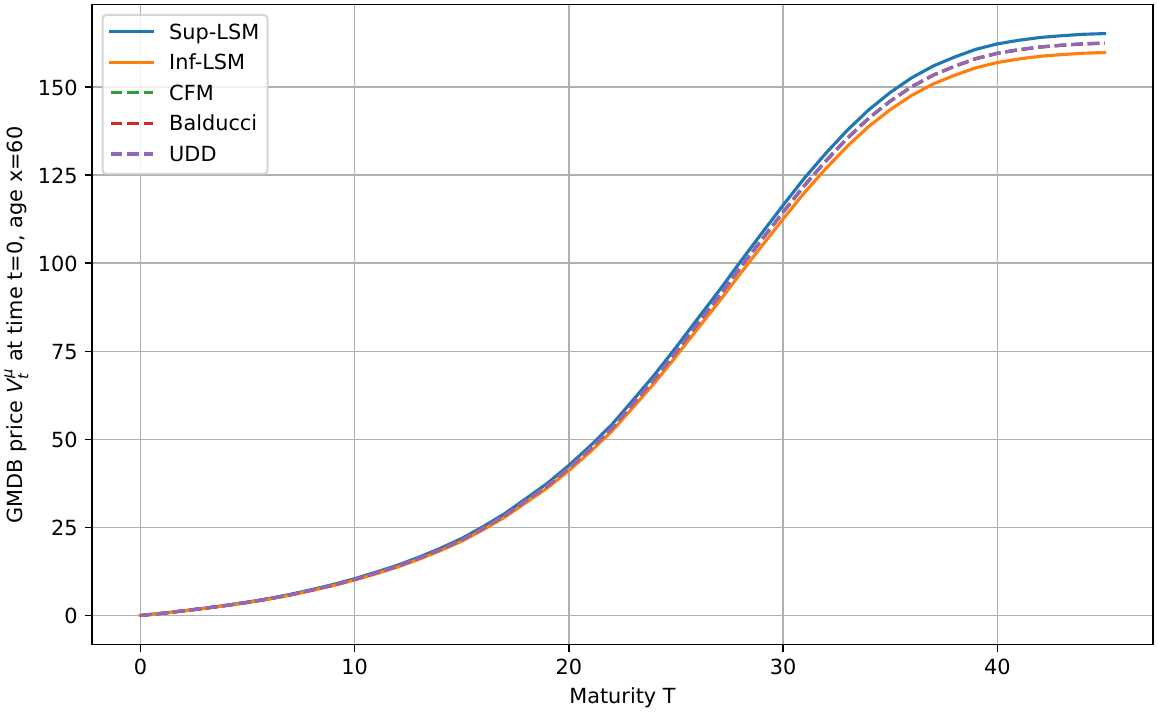}
  \\    \includegraphics[width=0.48\textwidth]{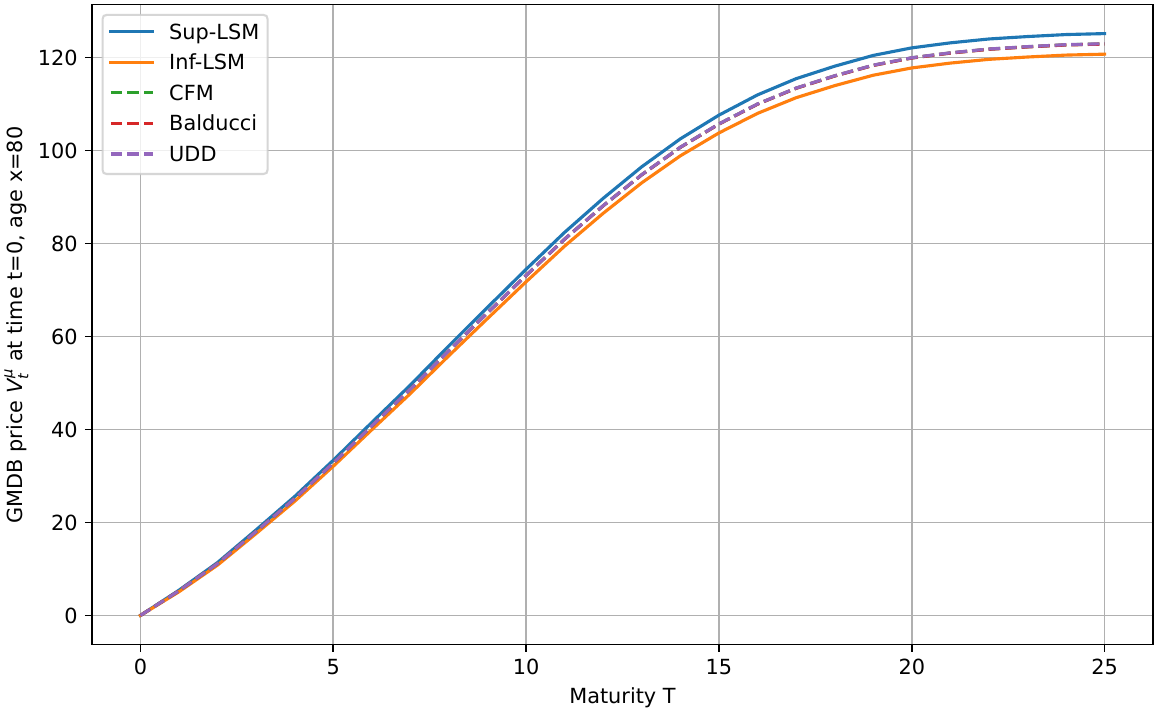}
  \caption{GMDB contract price $V^\mu_t$ at $t=0$ for different maturities $T$ with initial age $x=40$ (top left), $x=60$ (top right) and $x=80$ (bottom). Results are shown under the fractional-age assumptions \eqref{UDD}--\eqref{Bald}, together with the LSM-estimates of the GMDB upper bound $\bar{V}_0$ ('Sup-LSM') and lower bound $\bunderline{V}_0$ ('Inf-LSM')   from Eqs.\ \eqref{lbound2}--\eqref{ubound2}.\vspace{-2mm}}
  \label{fig:GMDB_strict}
\end{figure}

\begin{figure}[h!]
    \centering
    \includegraphics[width=0.88\textwidth]{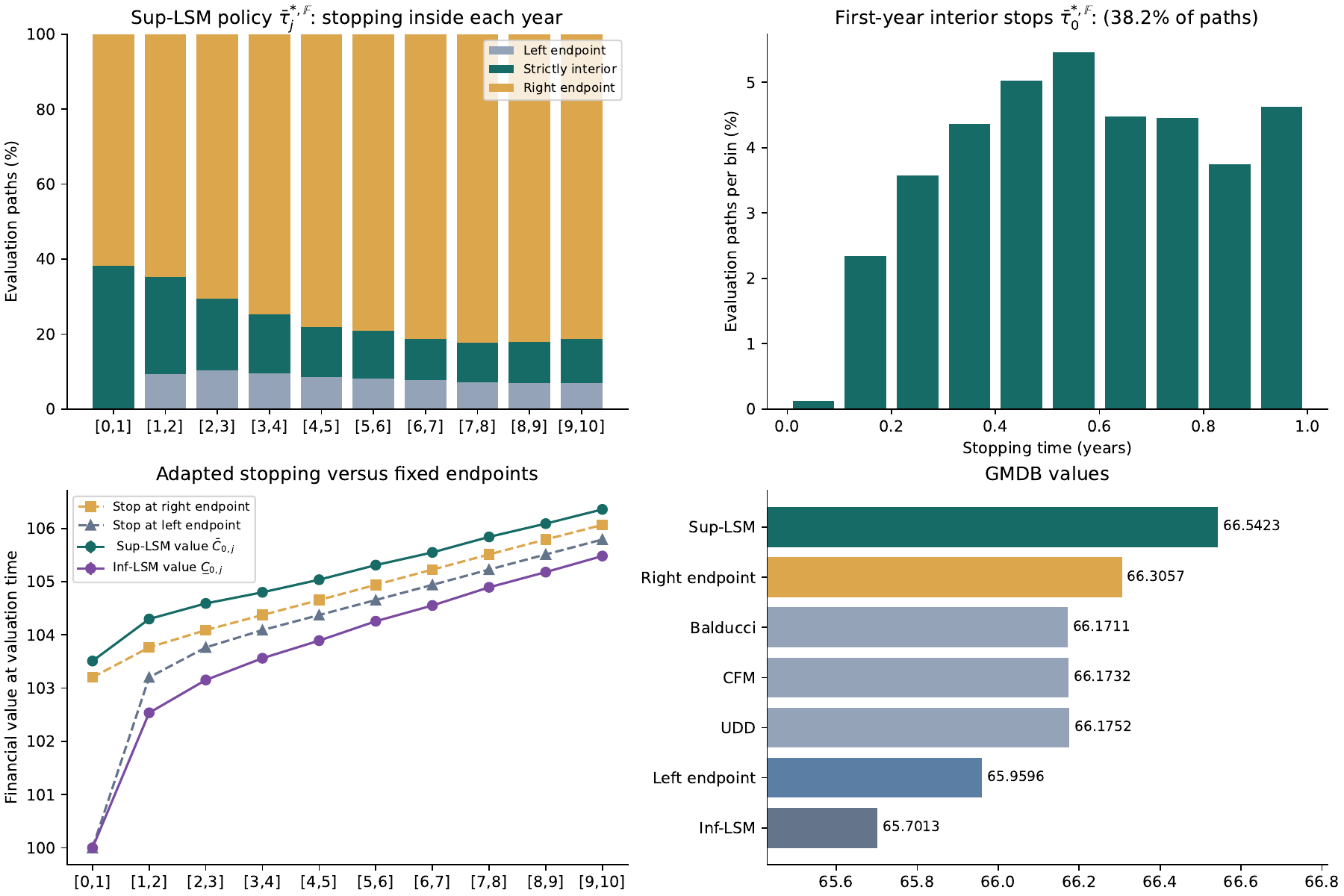}
    \caption{LSM approximation of the GMDB bounds in Proposition~\ref{propgmdb}. Upper left: endpoint and interior stopping frequencies under the sup-LSM policy $\bar\tau^{*,\mathbb{F}}_j$, $j=0,\ldots,9$. Upper right: distribution of first-year interior stops $\bar\tau^{*,\mathbb{F}}_0$. Lower left: annual financial values \eqref{eq:stopp_sup} under sup-LSM $(\bar{C}_{0,j})$, inf-LSM $(\bunderline{C}_{0,j})$, and fixed-endpoint policies. Lower right: corresponding estimates of the GMDB bounds $\bar{V}_0$ and $\bunderline{V}_0$ with UDD, CFM, and Balducci valuations.\vspace{-2mm}}
    \label{fig:gmdb_stopping_res}
\end{figure}

Finally, Figure \ref{fig:gmdb_stopping_res} focuses on the LSM approximation of the GMDB bounds, for a policyholder aged $x=80$ with contract maturity $T=10$. The upper-left panel reports the proportions of the sup-LSM policy stops $\bar\tau^{*,\mathbb{F}}_j$ in \eqref{eq:stopp_sup}  occurring at the left endpoint $\lceil t\rceil+j$ , strictly inside the interval, and at the right endpoint $\lceil t\rceil+j+1$, for $j=0,\ldots,9$; the upper-right panel shows the timing of first-year  interior stops $\bar\tau^{*,\mathbb{F}}_0$ $(j=0)$, with percentages expressed relative to all 400,000 evaluation paths. Positive frequencies of interior stopping illustrate the non-trivial, state-dependent nature of the learned policy. The lower-left panel compares the estimated annual financial values $\overline C_{0,j}$ (resp. $\bunderline C_{0,j}$), $j=0,\ldots 9$, obtained under the sup- and inf- LSM policies $\bar\tau^{*,\mathbb{F}}_j$ (resp.\ $\bunderline\tau^{*,\mathbb{F}}_j$) with those obtained by always stopping at either endpoint, thereby illustrating the effect of adapting the stopping decision to financial information. The lower-right panel compares the resulting estimated GMDB bounds $\bar{V}_0$ and $\bunderline{V}_0$ in  \eqref{lbound2}--\eqref{ubound2} with valuations under UDD, CFM and Balducci assumptions. Their differences illustrate valuation sensitivity to within-year mortality timing, which is not specified by the annual life table.

\subsection{Relaxed bounds}
For clarity, in this section we consider a deterministic interest rate $r = 0.01$ and set the market price of risk to zero ($\mu_S = r$), which reduces the Hamilton–Jacobi–Bellman (HJB) equations from four state variables to only two. While this simplification ignores stochastic interest rate effects, the qualitative behavior of bounds is expected to remain similar in the full four-state formulation. The other parameters used in the numerical examples are the same as in Section \ref{Sec: numerics_strict} above with $\delta = 1$. Moreover, the a priori time-dependent mortality rate $\mu_{x,a}$ is given by the Balducci approximation. We compare different variable annuities prices:
\begin{itemize}
    \item Prices obtained by using the Balducci approximation,
    \item Worst- and best-case prices by imposing expected constraints only at maturity $T$,
    \item Worst- and best-case prices by imposing expected constraints at each intermediate maturities $j\leq T$,
    \item Worst- and best-case prices by imposing a.s.\ constraints for each intermediate maturities $j\leq T$.
\end{itemize}

Figures~\ref{fig:GMAB_relaxed}, \ref{fig:GMIB_relaxed} and \ref{fig:GMDB_relaxed} display the upper and lower bounds (worst- and best-cases prices) for GMAB, GMIB, and GMDB for policyholder ages $x = 40, 60, 80$. Comparing expected bounds with almost-sure bounds, it is clear that the expected bounds are less restrictive, as they only enforce average behavior over time rather than across all possible mortality trajectories. Consequently, in the numerical examples considered here, the resulting bounds are wider than the almost-sure bounds. Note that for GMAB, the almost-sure bounds coincide with the price under Balducci approximation, so they are not displayed; for GMIB and GMDB, the almost-sure bounds are tighter and lie closer to the baseline.\\

Moreover, we observe that adding additional constraints at intermediate times further narrows the bounds, as they restrict extreme mortality trajectories that would otherwise satisfy the final maturity condition but produce very high or very low liabilities along the path. The effect of these intermediate-time constraints is stronger for GMIB and GMDB than for GMAB because GMIB and GMDB payoffs depend on survival or death at intermediate times, whereas GMAB depends primarily on survival at maturity $T$. \\

We also observe asymmetry in the bounds, when comparing upper and lower bounds to the baseline. This asymmetry depends on the nature of the payoff. For survival-contingent products such as GMIB and GMAB, the lower bounds are typically far below the baseline, reflecting the strong impact of higher mortality on survival-dependent liabilities, while the upper bounds remain closer. For death-contingent products such as GMDB, the opposite pattern is observed: the upper bounds lie far above the baseline, as higher mortality increases the present value of death benefits, whereas the lower bounds remain near the baseline, since reductions in mortality have limited effect. Those asymmetrical effects are reduced for older ages, as survival probabilities decrease. 

\begin{figure}[h!]
  \centering
  \hspace{-4mm}\includegraphics[width=0.49\textwidth]{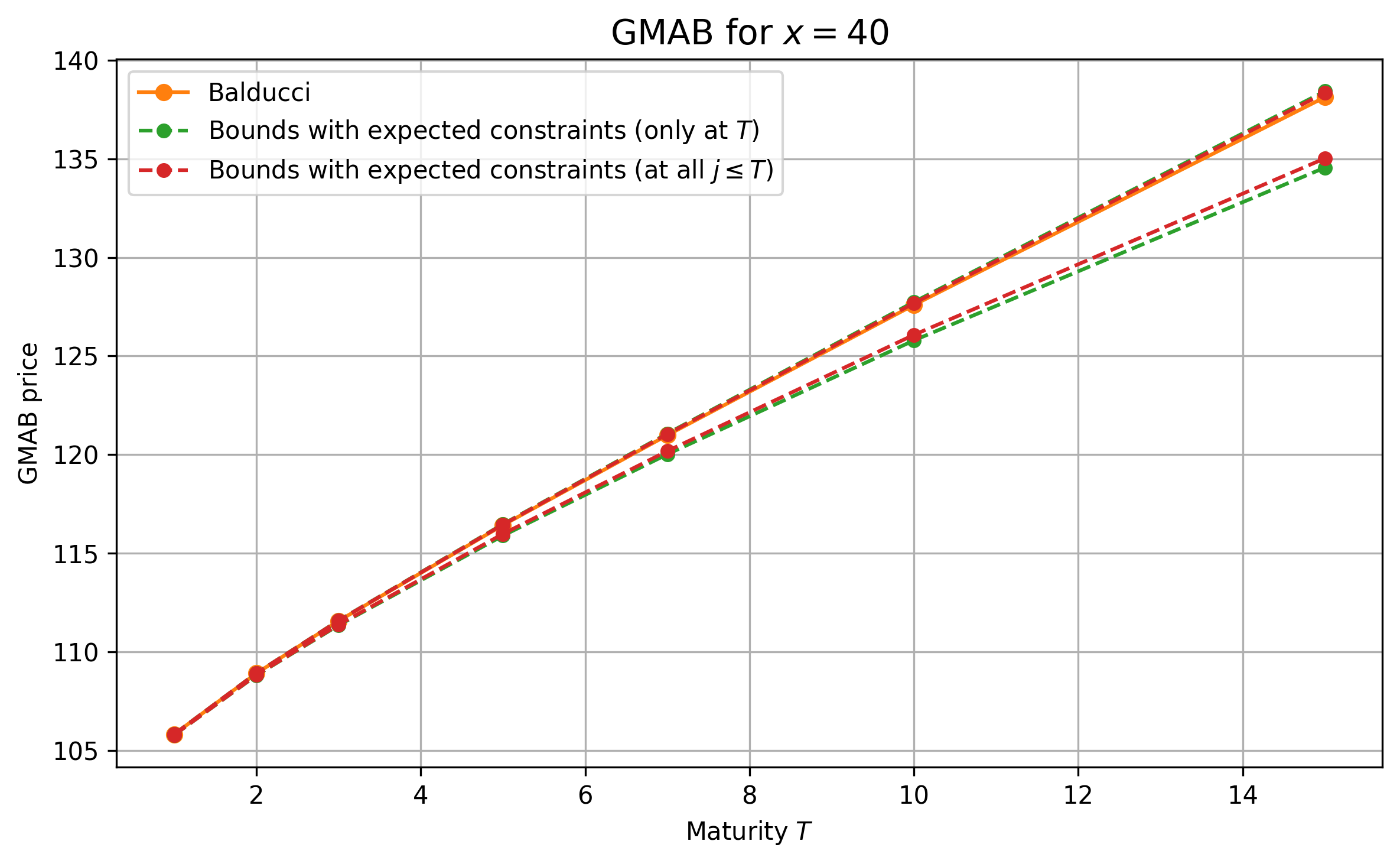} \hspace{-2mm}
  \hfill 
  \includegraphics[width=0.49\textwidth]{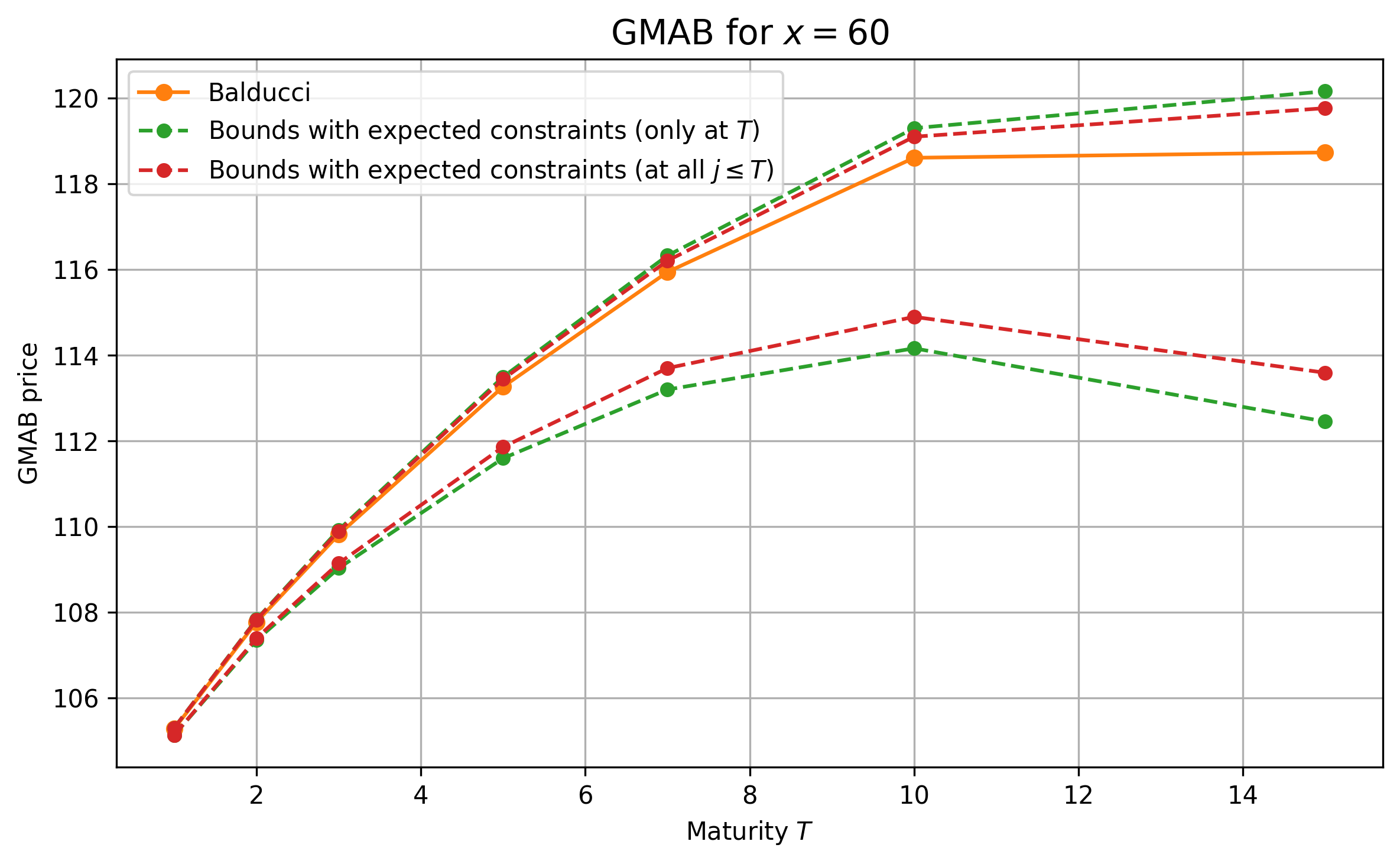}\\
  \includegraphics[width=0.49\textwidth]{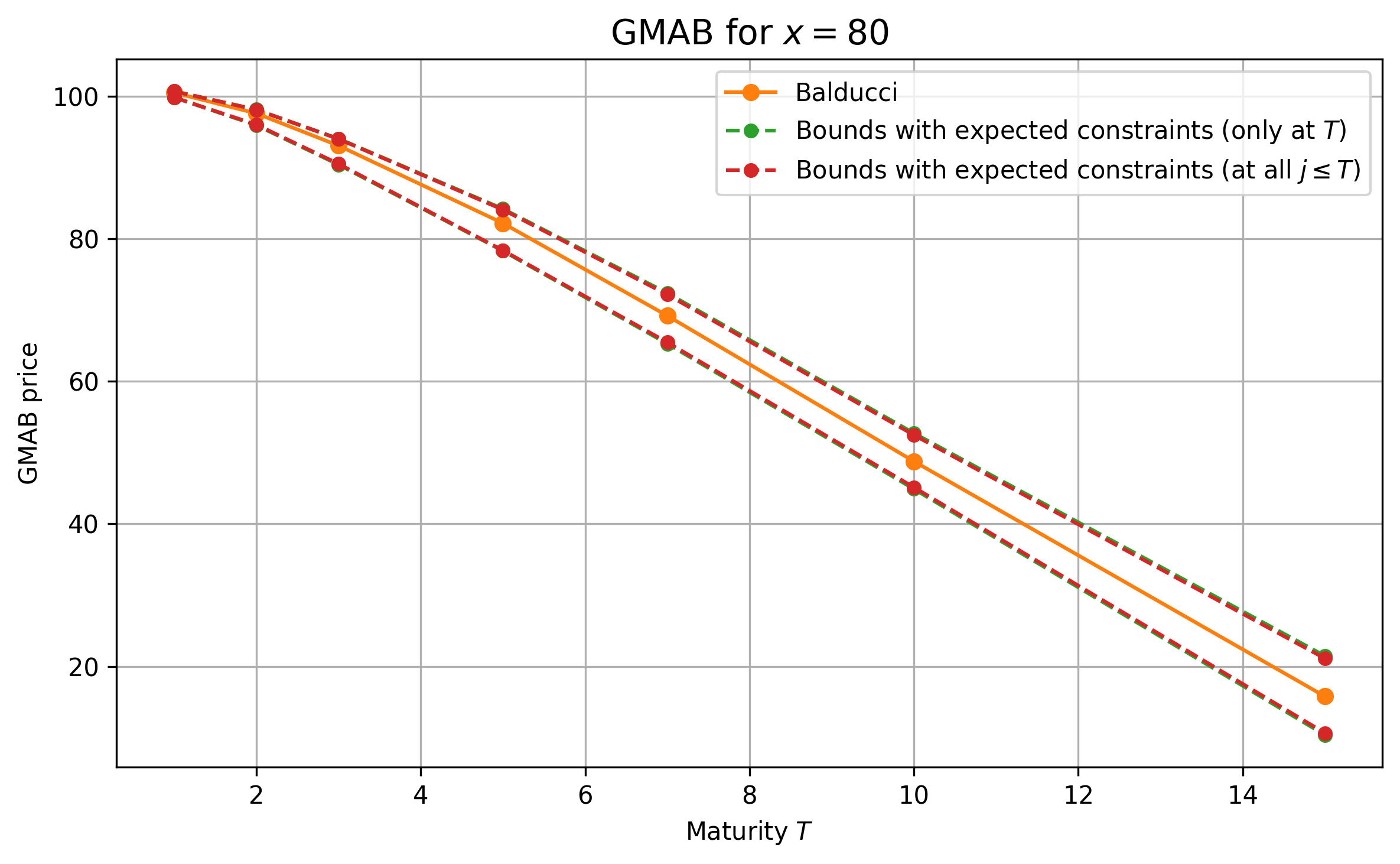}
  \caption{GMAB contract price for different maturities $T\in\{1,2,3,5,7,10,15\}$ with initial age $x=40$ (top left), $x=60$ (top right) and $x=80$ (bottom). Results are shown under the Balducci assumption \eqref{Bald}, together with the optimal relaxed bounds from Theorem \ref{prop: worst_case_price_regu} and Proposition \ref{prop: best_case_price}. }
  \label{fig:GMAB_relaxed}
\end{figure} 

\begin{figure}[h!]
  \centering
  \hspace{-4mm}\includegraphics[width=0.49\textwidth]{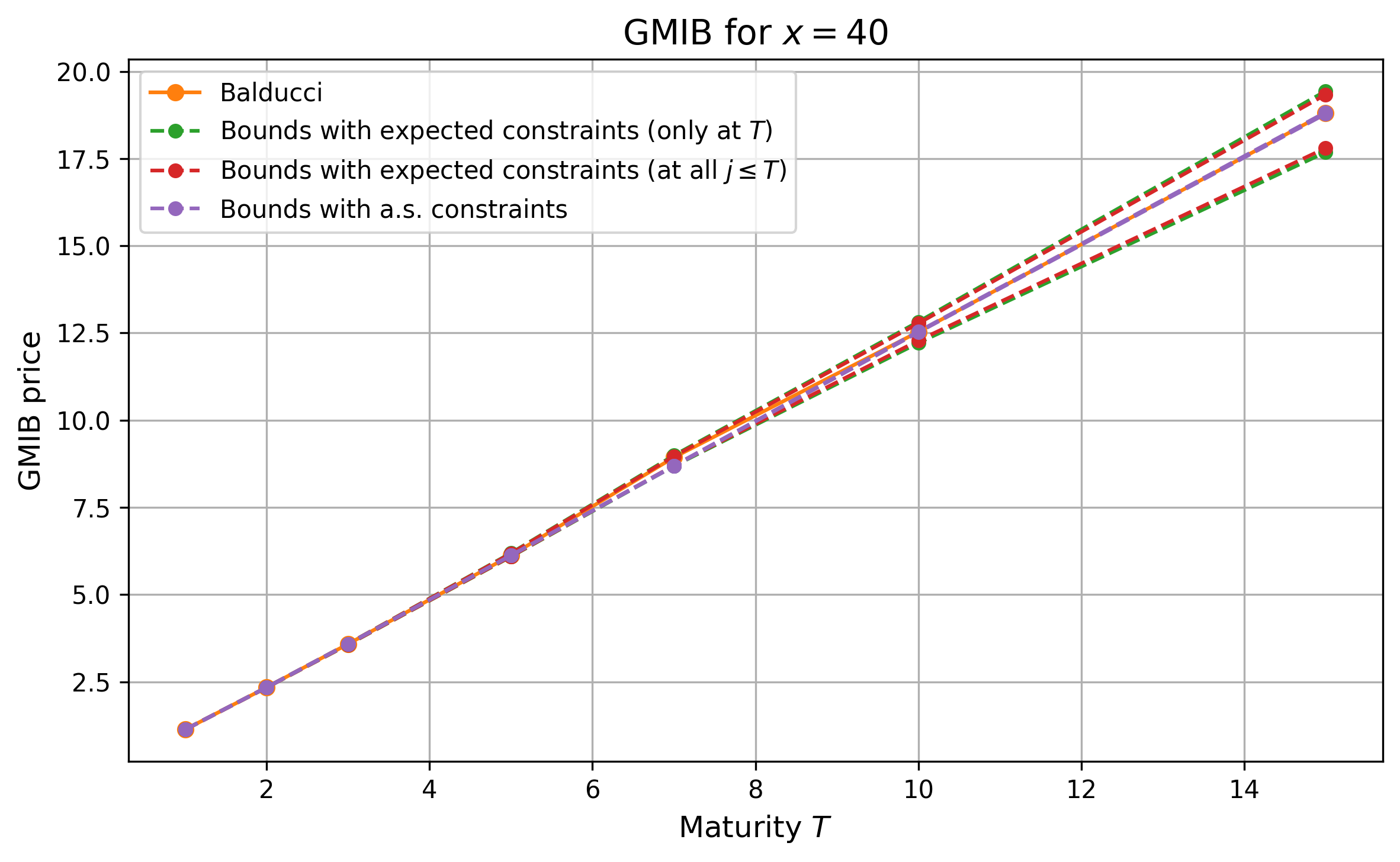} \hspace{-2mm}
  \hfill 
  \includegraphics[width=0.49\textwidth]{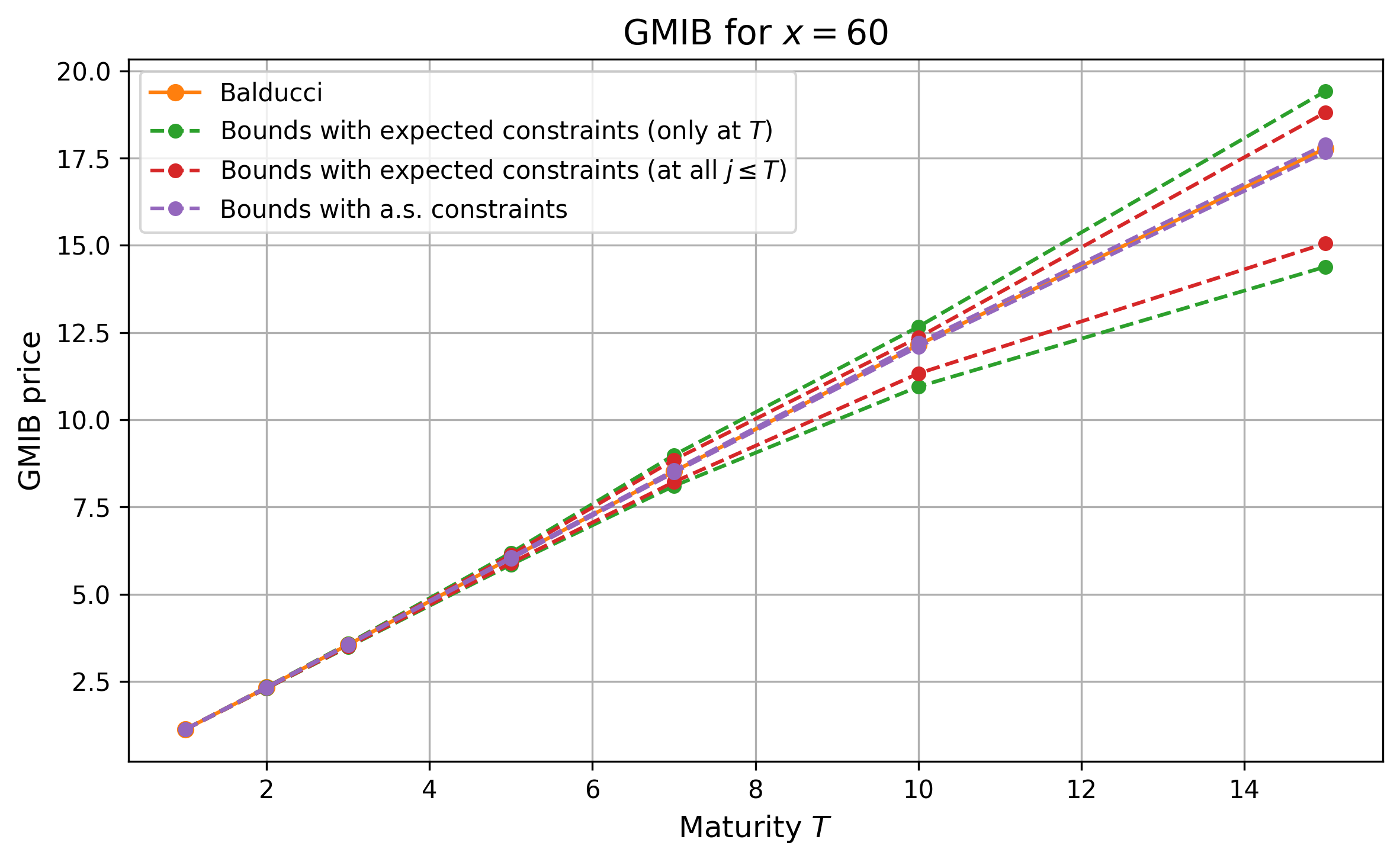}\\
  \includegraphics[width=0.49\textwidth]{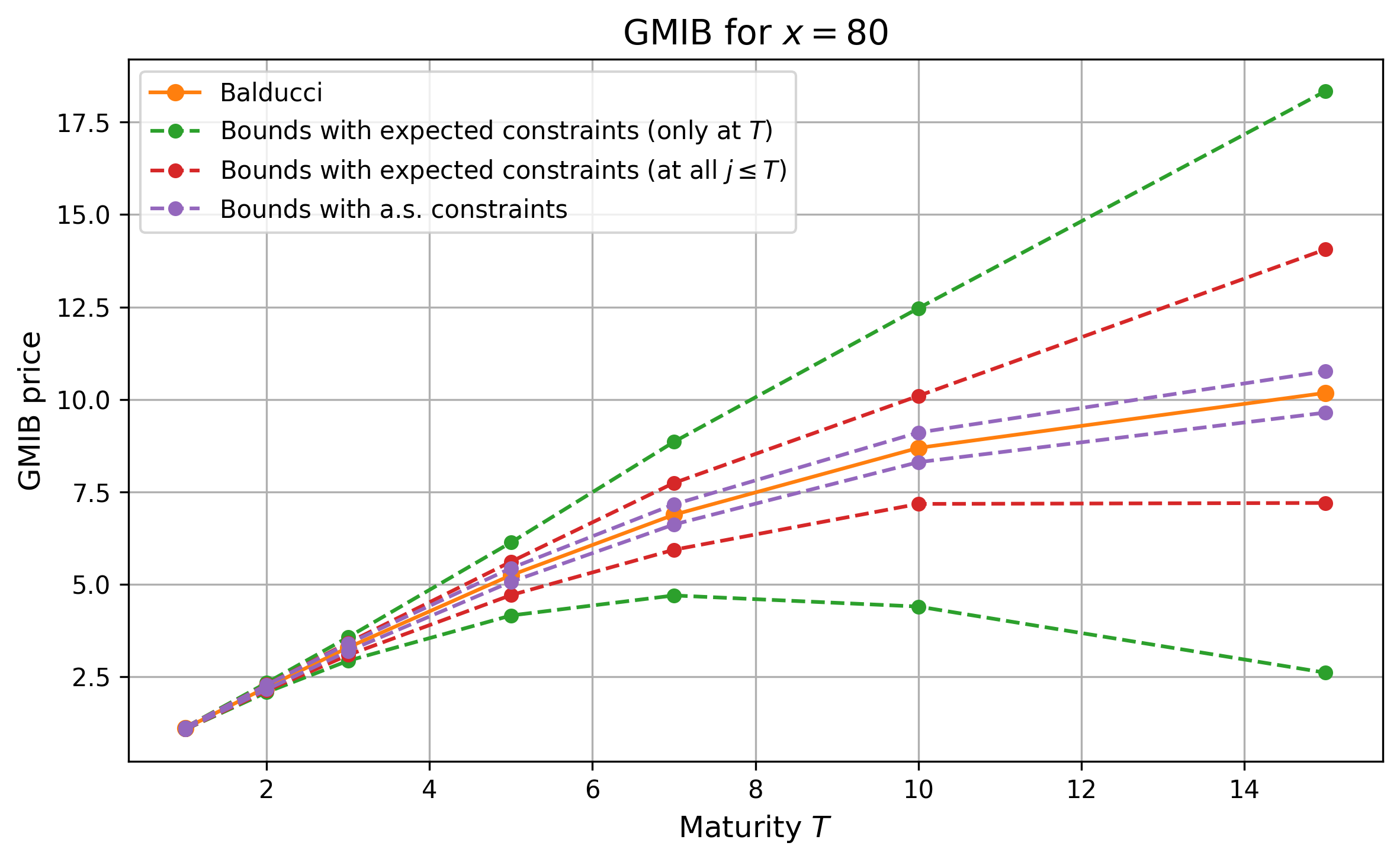}
  \caption{GMIB contract price for different maturities $T\in\{1,2,3,5,7,10,15\}$ with initial age $x=40$ (top left), $x=60$. (top right) and $x=80$ (bottom). Results are shown under the Balducci assumption \eqref{Bald}, together with the optimal relaxed bounds from Theorem \ref{prop: worst_case_price_regu}, Proposition \ref{prop: best_case_price} and a.s.\ bounds from Eqs.\ \eqref{upper_GMIB}--\eqref{lower_GMIB} }
  \label{fig:GMIB_relaxed}
\end{figure} 
\begin{figure}[h!]
  \centering
  \hspace{-4mm}\includegraphics[width=0.5\textwidth]{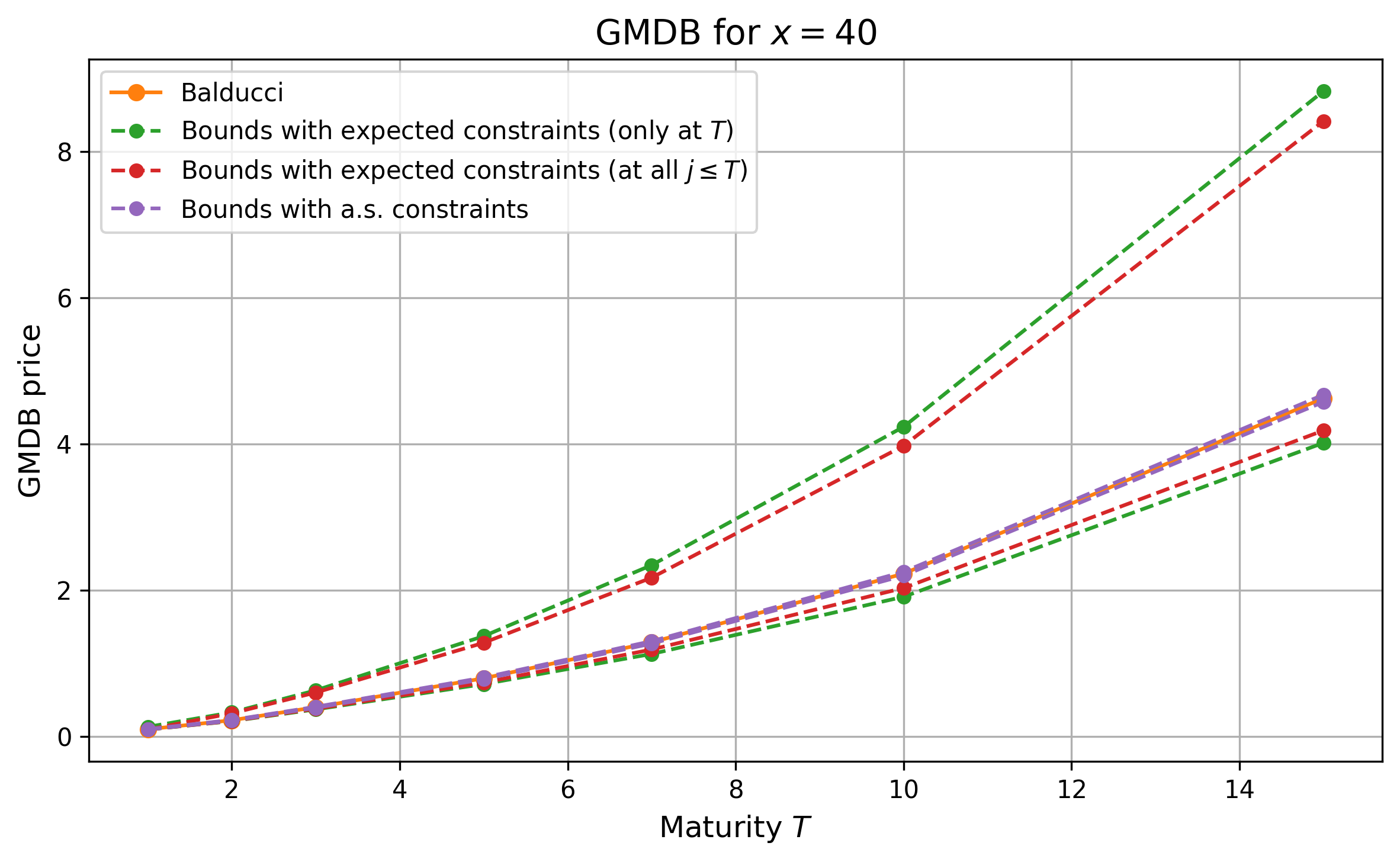} \hspace{-2mm}
  \hfill 
  \includegraphics[width=0.5\textwidth]{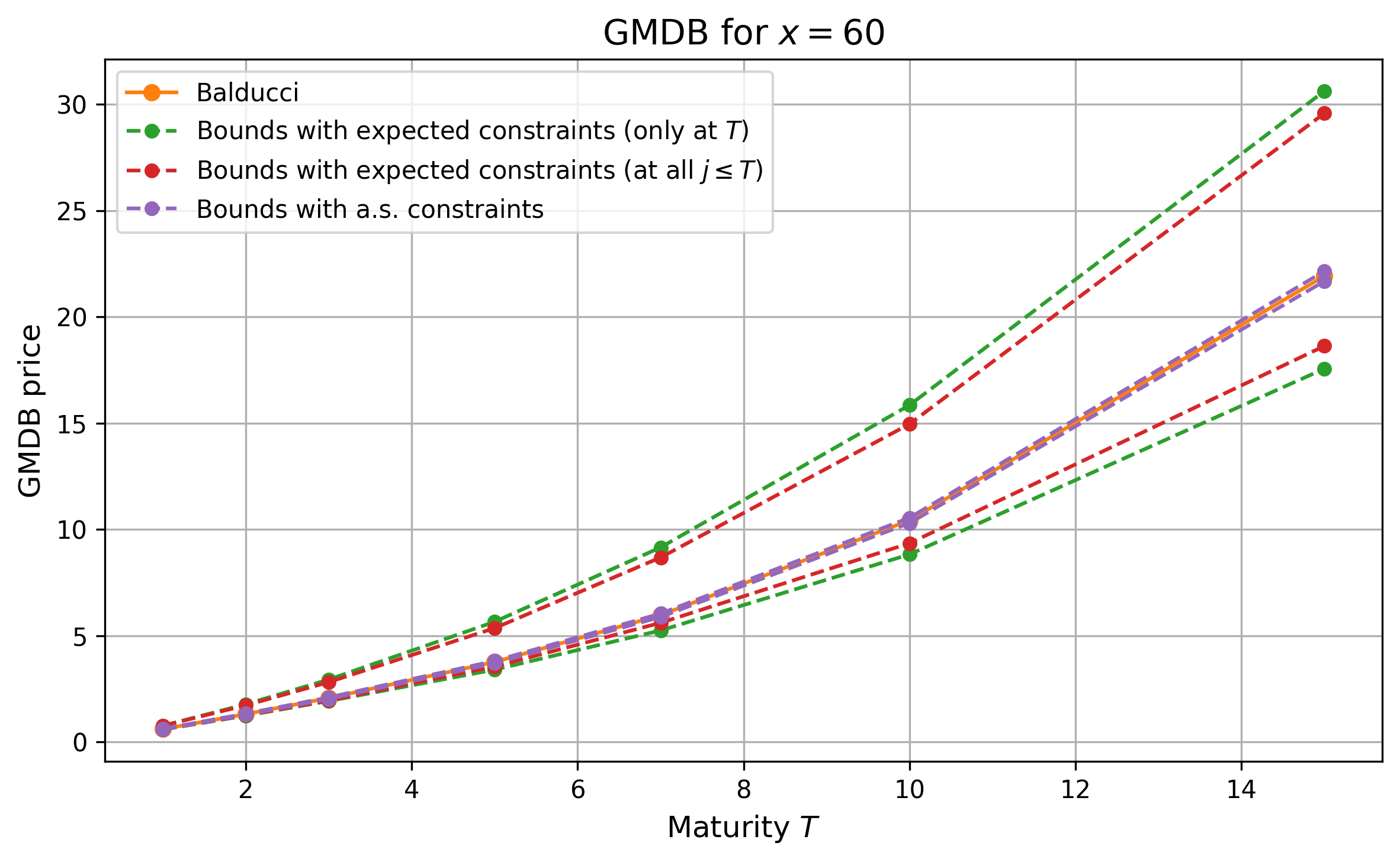}\\
  \includegraphics[width=0.5\textwidth]{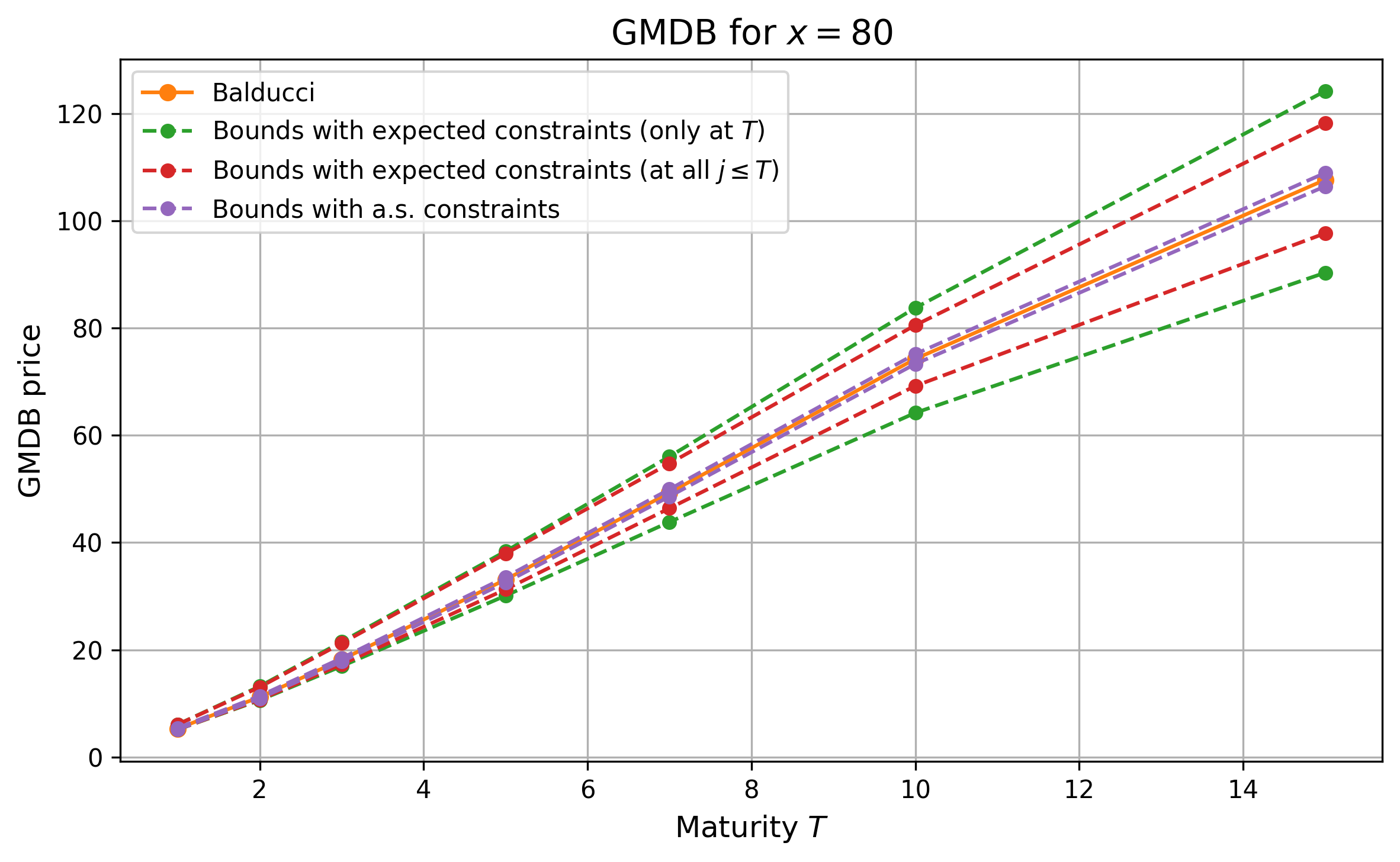}
  \caption{GMDB contract price for different maturities $T\in\{1,2,3,5,7,10,15\}$ with initial age $x=40$ (top left), $x=60$. (top right) and $x=80$ (bottom). Results are shown under the Balducci assumption \eqref{Bald}, together with the optimal relaxed bounds from Theorem \ref{prop: worst_case_price_regu}--\ref{prop: best_case_price} and a.s.\ bounds from Eqs.\ \eqref{lbound2}--\eqref{ubound2}.}
  \label{fig:GMDB_relaxed}
\end{figure}

We also check the convergence of the upper and lower bounds with expected constraints only at maturity, as $\delta\to \infty$. Figure \ref{fig:conv_delta_60_5} displays the upper and lower contract bounds for $x=60$ and $T=5$, plotted as functions of $\delta>0$. As expected, we observe a relatively fast convergence as $\delta$ increases, confirming the theoretical result of Proposition \ref{prop: convergence_delta}. \\

\begin{figure}[h!]
  \centering
  \hspace{-4mm}\includegraphics[width=0.5\textwidth]{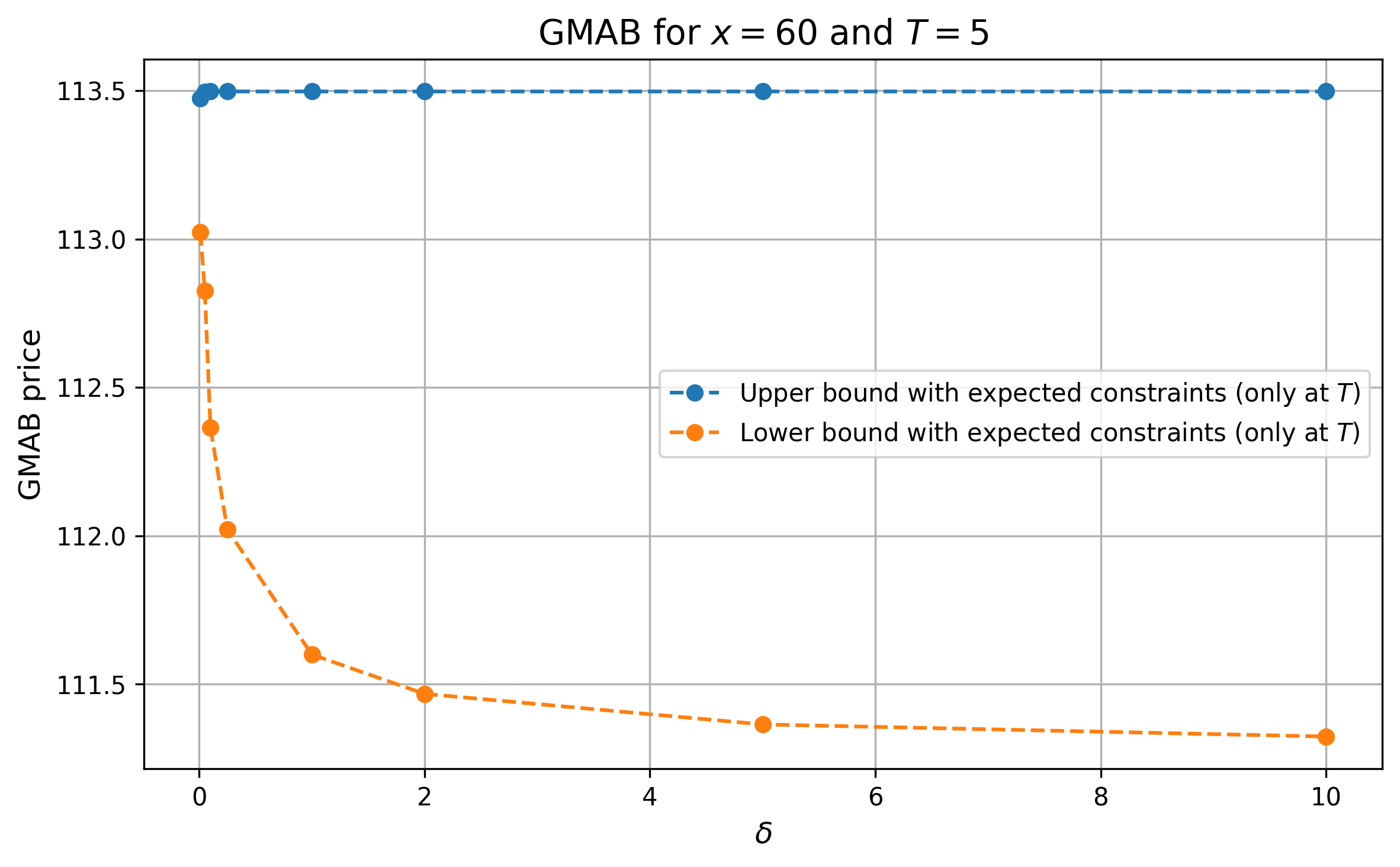} \hspace{-2mm}
  \hfill 
  \includegraphics[width=0.5\textwidth]{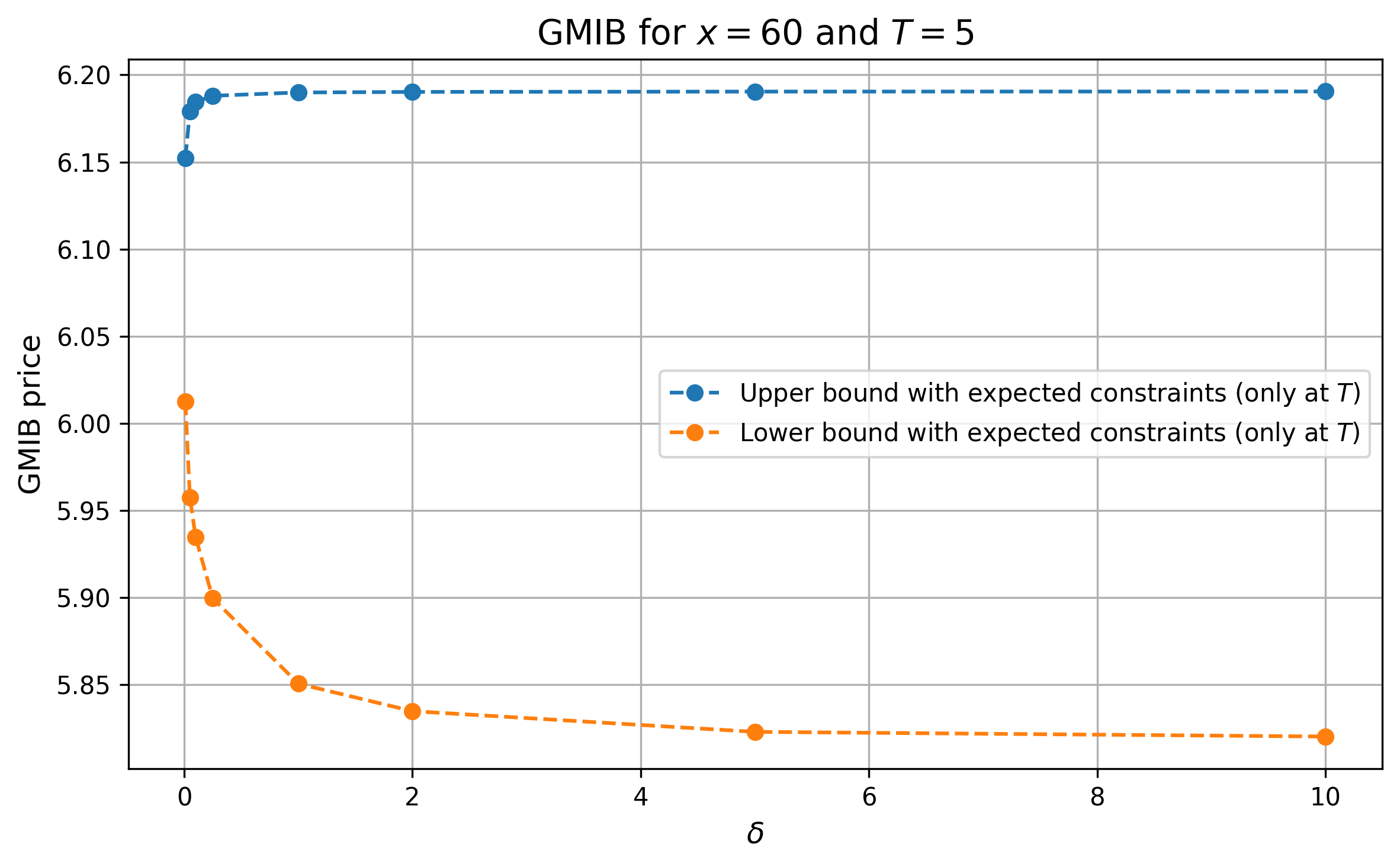}\\
  \includegraphics[width=0.5\textwidth]{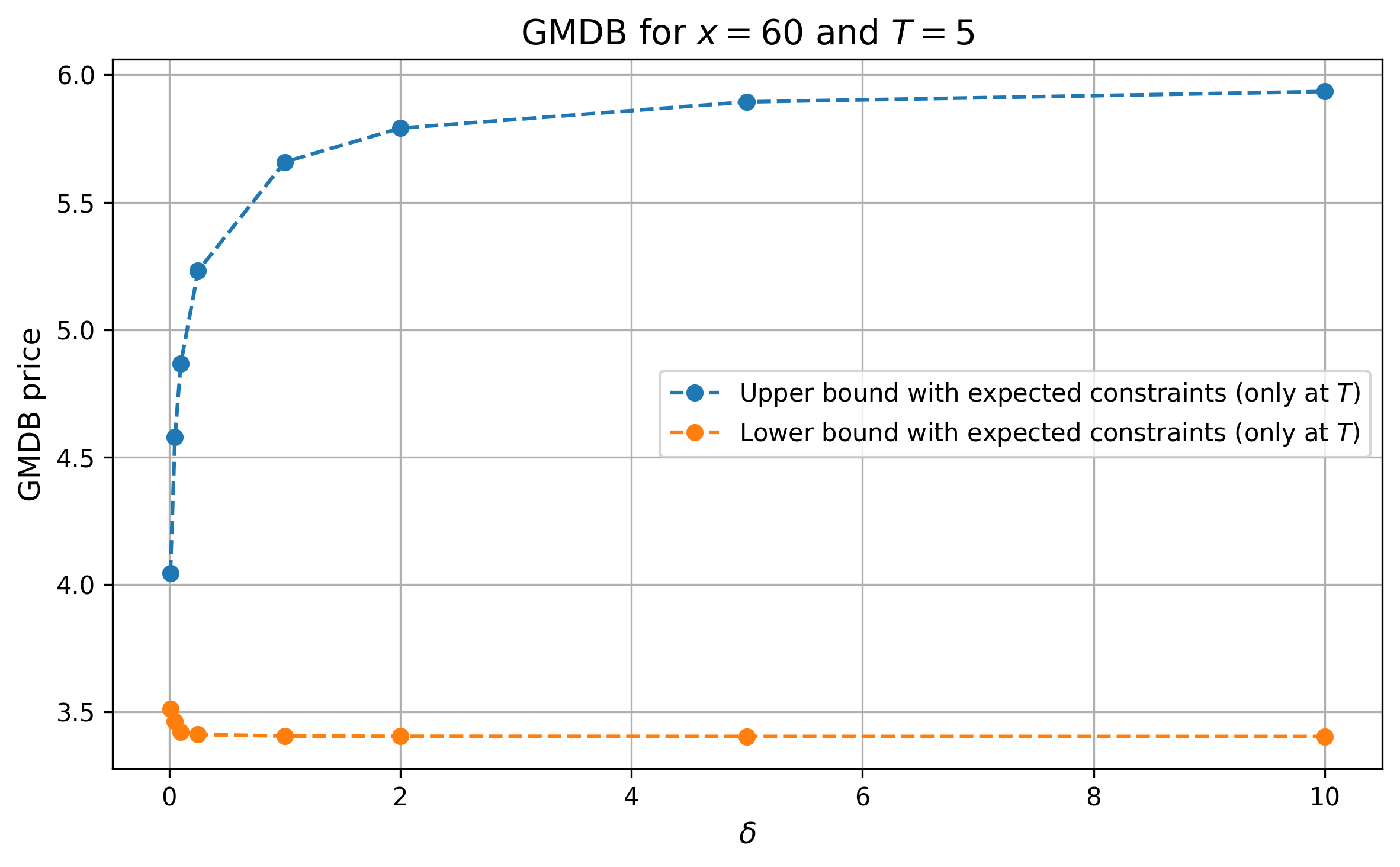}
  \caption{ Optimal relaxed bounds from Proposition \ref{prop: worst_case_price_regu}--\ref{prop: best_case_price} with expected constraint only at maturity $T$ as functions of $\delta>0$, for GMAB, GMIB and GMDB contract prices with initial age $x=60$ and maturity $T=5$. }
  \label{fig:conv_delta_60_5}
\end{figure}
We conclude this numerical section with some remarks on the computational aspects of the problem and the methodological choices adopted in its resolution. The original control problems are formulated as a four-dimensional HJB equations with multiple time constraints on $[0,T]$. This leads to a sequence of recursive HJB equations, which becomes computationally very demanding, especially for large maturities. In addition, the numerical resolution requires solving a high-dimensional static optimization problem with respect to $\boldsymbol{\lambda} \in \mathbb{R}^n$, which further increases the computational burden and slows convergence. To ensure tractable computation and reliable numerical results, we decide to reduce the state space to two variables. The full four-dimensional formulation, with stochastic interest rates, could in principle be addressed using modern deep-learning-based methods for high-dimensional control problems instead of the finite-difference scheme used above, see for example \cite{cheridito2025deep, dupret2024deep}, which are specifically designed to handle such large systems of recursive HJB equations. Let us also note that the previous section provides alternative almost-sure bounds obtained under stronger and less realistic assumptions. Although more restrictive, these bounds are fast to compute and can be derived directly. They thus serve as useful a priori estimates and sanity checks for the numerical results.

\section{Conclusion}
This paper addresses a structural source of model risk in life insurance: life tables provide survival probabilities only at integer ages, leaving  the within-year timing of deaths unspecified, even though many modern life-contingent liabilities, such as variable annuities, depend on the full lifetime distribution.  Our approach then complements fractional-age assumptions and  mortality rate models by providing a model-free mortality framework that characterizes the  admissible mortality intensities compatible with the tabulated one-year survival probabilities and with arbitrage-free financial valuation, yielding sharp upper and lower bounds for hybrid actuarial-financial functionals over this set. We present two distinct yet complementary approaches. In the strict (pathwise) matching regime, admissible mortality trajectories must reproduce each observed one-year survival probability almost surely. Because strict matching can be overly restrictive in practice, we then study a relaxed regime where life-table constraints are matched in expectation. We define the corresponding admissible set as the class of essentially bounded mortality intensities satisfying these constraints. To make the corresponding control problem suitable for a dynamic programming approach, we introduce a restricted admissible set around a (plausible) baseline mortality intensity consistent with the table, and solve the resulting robust control problem via a dual formulation with Lagrange multipliers and dynamic programming. Numerical experiments confirm intuitive patterns: (i) expectation-based bounds are wider than almost-sure bounds; (ii) adding intermediate-time constraints tightens bounds substantially for GMIB/GMDB; and (iii) the bound asymmetry depends on whether the payoff is survival- or death-contingent, with the effect diminishing at very old ages as survival probabilities shrink.  Overall, the proposed bounds provide a practical tool for incorporating mortality model risk into valuation and for risk management within a chosen financial model.

\appendix
\section{Proofs}
\label{sec: appendix}
\renewcommand{\thesection}{A}
\subsection{Proof of Proposition \ref{propgmab}}
\begin{proof}
 Using \eqref{GMAB}, \eqref{const} and \eqref{termproba}  with Assumption \ref{ass: extremal}, we have that for every $\mu \in \mathcal{W}$,
\begin{align*}
    V^\mu_t &= \mathbb{E}^{\mathbb{Q}}  \left[ \left(\hspace{-0.1mm}\prod_{j=0}^{\lfloor T \rfloor-\lceil t\rceil -1}\hspace{-1mm} \hat{p}_j \right) e^{-\int_t^T r_s \mathrm{d}s}  \, H(A_T,G^A_T) \ \big| \  \mathcal{H}_{t,t} \right] 
    \\
    &= {}_{T-t}p_{x+t} \, \mathbb{E}^{\mathbb{Q}}  \bigg[ \,e^{-\int_t^T r_s \mathrm{d}s}  \, H(A_T,G^A_T) \ \big| \  \mathcal{H}_{t,t} \, \bigg] \, .
\end{align*} 
Under Assumption \ref{Assumption: change of measure}, the minimal martingale measure is understood with respect
to the enlarged filtration $\widehat{\mathbb H}$. By the martingale representation property of the complete financial
market and since discounted traded prices remain
$(\widehat{\mathbb H},\mathbb Q)$-local martingales, we have $\q$-a.s.\ the \textit{immersion} property
\begin{equation} \label{eq:immersion}
\mathbb E^{\mathbb Q}[X\mid \widehat{\mathcal{H}}_t]
=
\mathbb E^{\mathbb Q}[X\mid\mathcal F_t], \qquad  0\le t\le s \ \  \text{and} \ \
X\in L^1(\mathcal F_s,\mathbb Q) \, ,
\end{equation} which implies $\mathbb E^{\mathbb Q}[X\mid\mathcal H_{t,t}]
=
\mathbb E^{\mathbb Q}[X\mid\mathcal F_t]$  using  $\mathcal F_t\subseteq\mathcal H_{t,t}
\subseteq\widehat{\mathcal H}_t$ and the tower property. Then, $$V^\mu_t = {}_{T-t}p_{x+t} \, C_t(T,G^A_T) \, .$$ As the survival probability ${}_{T-t}p_{x+t}$  is known and given by \eqref{termproba}, the objective functional $V^\mu_t$ does not depend on the underlying mortality control $\mu \in \mathcal{W}$. Therefore, any admissible $\mu$ is a maximizer and minimizer of $V^\mu_t$ and we find $\bar{V}_t = \bunderline{V}_t = V^\mu_t$.     
\end{proof} 
\subsection{Proof of Proposition \ref{propgmib}}
\begin{proof}
\smallskip
\noindent\textit{(i) Upper bound.} By Assumption~\ref{ass: extremal}, it suffices to optimize the mortality allocation
within each complete one-year interval.
We hence fix the interval $(\lceil t\rceil+j,\lceil t\rceil+j+1]$ for some $j = 0,\ldots, \lfloor T \rfloor - \lceil t \rceil -1$. The life-table constraint \eqref{const} forces the total survival factor over $(\lceil t\rceil+j,\lceil t\rceil+j+1]$ to be $\hat p_j$.
Since $\widetilde H\ge 0$, the integrand
\[
s\mapsto e^{-\int_t^s r_u\,\mathrm{d}u}\,\widetilde H(A_s,G_s^I)
\]
is nonnegative. Let $\mathcal{T}_{(a,b]}$ denote the set of $\mathbb{H}$-stopping times on $(a,b]$. Therefore,  for any $\tau_1,\tau_2\in\mathcal{T}_{(\lceil t\rceil +j , \lceil t\rceil + j+1]}$ with $\tau_1\le\tau_2$ a.s.,
\[
\int_{\lceil t\rceil +j}^{\tau_1} e^{-\int_t^s r_u\,\mathrm{d}u}\,\widetilde H(A_s,G_s^I)\,\mathrm{d}s
\le
\int_{\lceil t\rceil +j}^{\tau_2} e^{-\int_t^s r_u\,\mathrm{d}u}\,\widetilde H(A_s,G_s^I)\,\mathrm{d}s \, , 
\quad\text{a.s.}
\]
Taking conditional expectations given $\mathcal{H}_{t,t}$ leads to 
\begin{equation}\mathbb{E}^{\mathbb{Q}} \left[ \int_{\lceil t\rceil +j }^{\tau_1} e^{-\int_{t }^{s} r_u \mathrm{d}u } \, \widetilde{H}(A_s,G^I_s) \, \mathrm{d}s \ \big| \ \mathcal{H}_{t,t} \right] \leq \mathbb{E}^{\mathbb{Q}} \left[ \int_{\lceil t\rceil +j }^{\tau_2} e^{-\int_{t }^{s} r_u \mathrm{d}u } \, \widetilde{H}(A_s,G^I_s) \, \mathrm{d}s \ \big| \ \mathcal{H}_{t,t} \right] . \label{monotonicity}\end{equation} Hence the conditional expected accumulated discounted payoff \eqref{monotonicity} over $(\lceil t\rceil +j,\tau]$ is maximized by taking $\tau=\lceil t\rceil + j+1$.
To justify optimality over the entire admissible class $\mathcal{W}$,
write the survival factor appearing in~\eqref{GMIB} as
\[
S_s^\mu
:=
e^{-\int_t^s\mu_{x,c}(u)\,\mathrm{d}u}
\prod_{i:\,t<\tau_i\le s}(1-\mu_{x,i}).
\]
By Assumption~\ref{ass: extremal} and~\eqref{const},
\[
S_{\lceil t\rceil +j}^\mu=\prod_{i=0}^{j-1}\hat p_i,
\qquad
S_{\lceil t\rceil +j+1}^\mu=\prod_{i=0}^{j}\hat p_i,
\quad\text{almost surely}.
\]
Since $S^\mu$ is nonincreasing, its largest possible value
on $[\lceil t\rceil+j,\lceil t\rceil+j+1)$ is $\prod_{i=0}^{j-1}\hat p_i$.
This value is maintained throughout the interval by assigning
all the available mortality mass to its right endpoint.
Consequently, the upper bound is attained by setting
$\bar\mu_{x,c}^*\equiv0$ and choosing a single intervention
at $\bar\tau_j^*=\lceil t\rceil+j+1$ on each interval.
The life-table constraint~\eqref{const} then gives
$\bar\mu_{x,j}^*=1-\hat p_j$.
The resulting survival probability is deterministic and equals
\[
{}_{s-t}\bar p^*_{x+t}
=
\prod_{\substack{j=0,\ldots,\lfloor T \rfloor - \lceil t \rceil -1:\\  \lceil t\rceil+j+1\le s}}\hat p_j,
\qquad s\in[t,T].
\]
The preceding comparison, together with
Assumption~\ref{ass: extremal}, shows that
$S_s^\mu\le{}_{s-t}\bar p^*_{x+t}$ for every admissible
control and every $s\in[t,T]$, almost surely.

Using~\eqref{GMIB}, conditional Fubini, and the immersion
identity~\eqref{eq:immersion}, we obtain
\begin{align*}
\bar V_t
&=
\mathbb E^{\mathbb Q}\!\left[
\int_t^T
e^{-\int_t^s r_u\,\mathrm{d}u}
{}_{s-t}\bar p^*_{x+t}\,
\widetilde H(A_s,G_s^I)\,\mathrm{d}s
\,\middle|\,\mathcal H_{t,t}\right]
\\
&=
\int_t^T {}_{s-t}\bar p^*_{x+t}\,
\mathbb E^{\mathbb Q}\!\left[
e^{-\int_t^s r_u\,\mathrm{d}u}
\widetilde H(A_s,G_s^I)
\,\middle|\,\mathcal F_t\right]\mathrm{d}s
\\
&=
\int_t^T {}_{s-t}\bar p^*_{x+t}\,
C_t(s,G_s^I)\,\mathrm{d}s.
\end{align*}
where $C_t(s,G^I_s)$ is given by \eqref{call} using again Assumption \ref{Assumption: change of measure}. This proves~\eqref{upper_GMIB}--\eqref{sup_proba}.

\smallskip
\noindent\textit{(ii) Lower bound.}
The same monotonicity argument shows that the accumulated
discounted payoff is minimized by allocating mortality as early
as possible within each interval $(\lceil t\rceil+j,\lceil t\rceil+j+1])$.
For every admissible control, the nonincreasing survival factor
satisfies
\[
S_s^\mu\ge S_{\lceil t\rceil+j+1}^\mu
=\prod_{i=0}^{j}\hat p_i,
\qquad \lceil t\rceil+j<s\le \lceil t\rceil+j+1.
\]
Together with the prescribed survival factors at the endpoints
and Assumption~\ref{ass: extremal}, this yields
\[
S_s^\mu\ge
{}_{s-t}\bunderline p^*_{x+t}
:=
\prod_{\substack{j=0,\ldots,\lfloor T \rfloor - \lceil t \rceil -1:\\ \lceil t\rceil+j <s}}\hat p_j,
\qquad s\in[t,T].
\]
Since the discounted payoff is nonnegative, we consequently have
\[
V_t^\mu\ge
\int_t^T {}_{s-t}\bunderline p^*_{x+t}\,
C_t(s,G_s^I)\,\mathrm{d}s
\]
for every $\mu\in\mathcal W$.

To show that this lower bound is sharp, define, for each $n\ge1$,
an admissible control $\mu^{(n)}$ by setting
$\mu_{x,c}^{(n)}\equiv0$ and placing a single intervention
on each $(\lceil t\rceil+j,\lceil t\rceil+j+1]$ at
\[
\tau_j^{(n)}=\lceil t\rceil+j+\frac1n,
\qquad
\mu_{x,j}^{(n)}=1-\hat p_j.
\]
These controls satisfy~\eqref{const}, and their intervention
times decrease to $\lceil t\rceil+j$ as $n\to\infty$.
Their deterministic survival probabilities satisfy
\[
S_s^{\mu^{(n)}}
=
\prod_{\substack{j=0,\ldots,\lfloor T \rfloor - \lceil t \rceil -1:\\ \lceil t\rceil+j +1/n\le s}}\hat p_j \, 
\downarrow \, 
\prod_{\substack{j=0,\ldots,\lfloor T \rfloor - \lceil t \rceil -1:\\ \lceil t\rceil+j<s}}\hat p_j
=
{}_{s-t}\bunderline p^*_{x+t}
\]
for every $s\in[t,T]$.

By conditional dominated convergence, using
$0\le S_s^{\mu^{(n)}}\le1$ and the finite-horizon integrability
of the discounted payoff,
\[
\mathbb E^{\mathbb Q}\!\left[
\int_t^T e^{-\int_t^s r_u\,\mathrm{d}u}
\widetilde H(A_s,G_s^I)\,\mathrm{d}s
\right]<\infty,
\]
we obtain
\begin{align*}
V_t^{\mu^{(n)}}
&\downarrow
\mathbb E^{\mathbb Q}\!\left[
\int_t^T
e^{-\int_t^s r_u\,\mathrm{d}u}
{}_{s-t}\bunderline p^*_{x+t}\,
\widetilde H(A_s,G_s^I)\,\mathrm{d}s
\,\middle|\,\mathcal H_{t,t}\right]
\\
&=
\int_t^T {}_{s-t}\bunderline p^*_{x+t}\,
C_t(s,G_s^I)\,\mathrm{d}s
=
\bunderline V_t.
\end{align*}
The second equality follows as in part~\textit{(i)}.
Together with the preceding lower-bound comparison, this
proves~\eqref{lower_GMIB}.

Suppose now that $\prod_{j=0}^{\lfloor T \rfloor - \lceil t \rceil -1}\hat p_j<1$, and let $j_0$
be the first index such that $\hat p_{j_0}<1$.
For any admissible control, we then have
$S_{\lceil t\rceil+j_0}^\mu=1$.
By right-continuity, $S_s^\mu>\hat p_{j_0}$ on a
possibly random interval of positive length immediately after
$\lceil t\rceil+j_0$, whereas
${}_{s-t}\bunderline p^*_{x+t}=\hat p_{j_0}$ there.
Since the GMIB payoff satisfies
$\widetilde H(A_s,G_s^I)\ge A_0r_I>0$, this implies
\[
V_t^\mu>\bunderline V_t
\qquad\text{almost surely}
\]
for every $\mu\in\mathcal W$.
Hence the lower bound is not attained.

Finally, if $\prod_{j=0}^{\lfloor T \rfloor - \lceil t \rceil -1}\hat p_j=1$, including the case
$\lfloor T \rfloor - \lceil t \rceil  \in \{-1, 0\}$, the life-table constraints and
Assumption~\ref{ass: extremal} imply $S_s^\mu=1$ throughout
$[t,T]$ for every admissible control.
Thus
\[
\bar V_t=\bunderline V_t
=\int_t^T C_t(s,G_s^I)\,\mathrm{d}s,
\]
and every admissible control is optimal.
\end{proof}

\subsection{Proof of Proposition \ref{propgmdb}}
\label{Proof: GMDB}
\begin{proof}
We work with the in-force valuation problem on $\{T_x>t\}$. Let us denote 
$$ a_j:=\lceil t\rceil+j,
\qquad
b_j:= \lceil t\rceil+j+1, \qquad w_j :=  (1-\hat p_j)\prod_{i<j}\hat p_i,
\qquad
X_s :=  e^{-\int_t^s r_u\,\mathrm du} \widehat H(A_s,G_s^D) \, . $$
Under the financial model of Section~\ref{sec: fin}, 
$ X_s $ is
continuous and
\begin{equation}\label{eq:gmdb_integrability}
\mathbb E^{\mathbb Q}\!\left[
\sup_{t\le s\le T}X_s\right]<\infty.
\end{equation}
Indeed, $e^{-\int_t^s r_u\,\mathrm du}A_s$ is a square-integrable martingale with
bounded deterministic volatility, so its running supremum is
integrable by Doob's inequality. Moreover, $G_s^D=A_0e^{r_gs}$
is bounded on $[t,T]$, and
\[
\sup_{t\le s\le T}e^{-\int_t^s r_u\,\mathrm du}G_s^D
\le \left(\sup_{t\le s\le T}G_s^D\right)
     e^{\int_t^T|r_u|\,\mathrm du}.
\]
The right-hand side is integrable by Jensen's inequality and the
Gaussian exponential moments of the Hull--White short rate.
Since $\widehat H(A_s,G_s^D)\le A_s+G_s^D$,
\eqref{eq:gmdb_integrability} follows.

\smallskip
\noindent\textit{(i) Upper bound.}
By Assumption~\ref{ass: extremal}, only complete one-year
intervals contribute to the GMDB value. For any
$\mu\in\mathcal W$, define the survival factor
\[
S_s^\mu:=e^{-\int_t^s\mu_{x,c}(u)\,\mathrm du}
\prod_{i:\,t<\tau_i\le s}(1-\mu_{x,i}).
\]
The life-table constraints~\eqref{const} imply
\[
S_{a_j}^\mu=\prod_{i<j}\hat p_i,
\qquad S_{b_j}^\mu=\hat p_j \prod_{i<j}\hat p_i \, ,
\qquad \int_{(a_j,b_j]}(-\mathrm dS_s^\mu)=(1-\hat p_j) \prod_{i<j}\hat p_i .
\]
Using~\eqref{GMDB}, we can therefore write
\begin{equation}\label{eq:gmdb_stieltjes}
V_t^\mu
=\sum_{j=0}^{\lfloor T \rfloor - \lceil t \rceil -1}\mathbb E^{\mathbb Q}\!\left[
\int_{(a_j,b_j]}X_s(-\mathrm dS_s^\mu)
\,\middle|\,\mathcal H_{t,t}\right].
\end{equation}
In particular, the mass at an impulse time $\tau_i$ is
$S_{\tau_i-}^\mu\mu_{x,i}$.

Fix a year with $w_j>0$. Define
\[
R_j^\mu(s):=\frac{\prod_{i<j}\hat p_i-S_s^\mu}{(1-\hat p_j) \prod_{i<j}\hat p_i},
\qquad
\sigma_j^\mu(v):=
\inf\{s\in[a_j,b_j]:R_j^\mu(s)\ge v\},
\quad 0<v<1.
\]
Since $R_j^\mu$ is adapted, increasing and c\`adl\`ag, with
$R_j^\mu(a_j)=0$ and $R_j^\mu(b_j)=1$, each
$\sigma_j^\mu(v)$ is an $\mathbb H$-stopping time taking
values in $(a_j,b_j]$. The inverse-distribution identity gives
\begin{equation}\label{eq:gmdb_allocation_identity}
\int_{(a_j,b_j]}X_s(-\mathrm dS_s^\mu)
=(1-\hat p_j) \prod_{i<j}\hat p_i\int_0^1X_{\sigma_j^\mu(v)}\,\mathrm dv.
\end{equation}
Thus every admissible annual mortality allocation is an
average of stopping payoffs.

Let $U^j$ be the financial Snell envelope of $X$ on
$[a_j,b_j]$. Continuity and~\eqref{eq:gmdb_integrability}
ensure the existence of an optimal stopping time
$\bar\tau_j^{*,\mathbb F}$ satisfying
\[
U^j_{a_j}
=\mathbb E^{\mathbb Q}
 [X_{\bar\tau_j^{*,\mathbb F}}\mid\mathcal F_{a_j}],
\qquad
\overline C_{t,j}
=\mathbb E^{\mathbb Q}[U^j_{a_j}\mid\mathcal F_t].
\]
By the immersion property \eqref{eq:immersion} under $\mathbb Q$,
$U^j$ remains a class-$D$ supermartingale in
$\widehat{\mathbb H}$.
It is also the Snell envelope in this larger filtration:
it dominates $X$, and financial optimal stopping times attain
its value. Hence optional sampling gives, for every
$\mathbb H$-stopping time $\tau\in[a_j,b_j]$,
\begin{align}
\mathbb E^{\mathbb Q}[X_\tau\mid\mathcal H_{t,t}]
&\le \mathbb E^{\mathbb Q}[U^j_{a_j}\mid\mathcal H_{t,t}]
\nonumber\\
&=\mathbb E^{\mathbb Q}[U^j_{a_j}\mid\mathcal F_t]
=\overline C_{t,j}.
\label{eq:gmdb_stopping_upper}
\end{align}
Here the equality follows from equation \eqref{eq:immersion} and
$\mathcal F_t\subseteq\mathcal H_{t,t}
\subseteq\widehat{\mathcal H}_t$.
Combining~\eqref{eq:gmdb_stieltjes}--\eqref{eq:gmdb_stopping_upper}
and applying conditional Fubini yields
\[
V_t^\mu\le\sum_{j=0}^{\lfloor T \rfloor - \lceil t \rceil -1}\overline C_{t,j} \, (1-\hat p_j) \prod_{i<j}\hat p_i \, ,
\qquad \mu\in\mathcal W.
\]

If the financial optimal times can be chosen strictly after
$a_j$ in every year with positive weight, place a single
impulse of size $1-\hat p_j$ at each such time and set
$\bar\mu_{x,c}^*\equiv0$. This control satisfies~\eqref{const},
and equation \eqref{eq:immersion} gives
\[
V_t^{\bar\mu^*}
=\sum_{j=0}^{\lfloor T \rfloor - \lceil t \rceil -1} \mathbb E^{\mathbb Q}
 [X_{\bar\tau_j^{*,\mathbb F}}\mid\mathcal F_t] \, (1-\hat p_j) \prod_{i<j}\hat p_i
=\sum_{j=0}^{\lfloor T \rfloor - \lceil t \rceil -1}  \overline C_{t,j} (1-\hat p_j) \prod_{i<j}\hat p_i \, .
\]
Otherwise, use the admissible times
$\bar\tau_j^{(n)}=\bar\tau_j^{*,\mathbb F}\vee(a_j+1/n)$,
$n\ge1$, with the same jump sizes and zero continuous mortality.
These times decrease to $\bar\tau_j^{*,\mathbb F}$.
Continuity and conditional dominated convergence give
\[
V_t^{\bar\mu^{(n)}}
=\sum_{j=0}^{\lfloor T \rfloor - \lceil t \rceil -1} \mathbb E^{\mathbb Q}
 [X_{\bar\tau_j^{(n)}}\mid\mathcal F_t] \, (1-\hat p_j) \prod_{i<j} \hat p_i\
\longrightarrow\sum_{j=0}^{\lfloor T \rfloor - \lceil t \rceil -1} \overline C_{t,j} \, (1-\hat p_j) \prod_{i<j}\hat p_i \, .
\]
This proves~\eqref{lbound2} and the approximation assertion.

To justify nonattainment, write the Doob--Meyer decomposition
on $[a_j,b_j]$ as $U^j=M^j-K^j$, where $K^j_{a_j}=0$ and
$K^j$ is continuous and nondecreasing. The standard
characterization of the largest optimal stopping time
(see Theorem~2.15 and Remark~B.13 of \cite{kobylanski2012optimal})
gives
\[
\lambda_j:=\inf\{s\in[a_j,b_j]:K^j_s>0\}\wedge b_j,
\]
with $\inf\varnothing=+\infty$.
The immersion property \eqref{eq:immersion} preserves this decomposition, so $\lambda_j$ is
also the largest optimal stopping time in
$\widehat{\mathbb H}$. If an admissible mortality control attained the upper bound,
equality in the preceding annual inequalities and
\eqref{eq:gmdb_allocation_identity} would imply that, for each
$j$ with $w_j>0$, $\sigma_j^\mu(v)$ is optimal for almost every
$v\in(0,1)$. More explicitly, the nonnegative conditional
optimality gaps have integral zero, so each vanishes for
almost every $v$; the tower property then gives optimality
conditional on $\widehat{\mathcal H}_{a_j}$. Consequently,
\[
a_j<\sigma_j^\mu(v)\le\lambda_j , 
\quad \q\text{-a.s.,}
\]
for almost every such $v$. Thus $\lambda_j>a_j$ almost surely,
providing a financial optimal stopping time strictly after
$a_j$ in every year with positive weight. This proves that
no admissible optimizer exists when such a collection of
financial optimal times cannot be chosen.

\smallskip
\noindent\textit{(ii) Lower bound.}
Apply the same argument to $-X$. Its financial Snell envelope
exists, and its negative, denoted by $L^j$, is a submartingale
satisfying $L^j\le X$ on $[a_j,b_j]$. There exists an optimal
financial minimizing time $\underline\tau_j^{*,\mathbb F}$,
and immersion and optional sampling give
\[
\mathbb E^{\mathbb Q}[X_\tau\mid\mathcal H_{t,t}]
\ge\mathbb E^{\mathbb Q}[L^j_{a_j}\mid\mathcal F_t]
=\underline C_{t,j}
\]
for every $\mathbb H$-stopping time $\tau\in[a_j,b_j]$.
Equations~\eqref{eq:gmdb_stieltjes}
and~\eqref{eq:gmdb_allocation_identity} therefore imply
\[
V_t^\mu\ge\sum_{j=0}^{\lfloor T \rfloor - \lceil t \rceil -1} \underline C_{t,j} \, (1-\hat p_j) \prod_{i<j}\hat p_i \, ,
\qquad \mu\in\mathcal W.
\]
The controls with zero continuous mortality, jump sizes
$1-\hat p_j$, and times
$\underline\tau_j^{(n)}
=\underline\tau_j^{*,\mathbb F}\vee(a_j+1/n)$ satisfy
\[
V_t^{\underline\mu^{(n)}}
\longrightarrow\sum_{j=0}^{\lfloor T \rfloor - \lceil t \rceil -1} \underline C_{t,j}\, (1-\hat p_j) \prod_{i<j}\hat p_i 
\]
by continuity and conditional dominated convergence.
This proves~\eqref{ubound2}.

If the minimizing times can be chosen strictly after $a_j$,
the corresponding impulse control attains the bound.
Conversely, applying the largest-optimal-time argument above
to $-X$ shows that attainment implies the existence of such
financial minimizing times. This proves the stated
nonattainment alternative.

Years with $w_j=0$ contribute nothing to the value; their
intervention times may be set to $b_j$. If $\lfloor T\rfloor-\lceil t\rceil \in \{-1, 0\}$, both bounds
are zero and every admissible control is optimal.
\end{proof}

\subsection{Proof of Theorem \ref{prop: worst_case_price_regu}}\label{Proof: worst_case_price_regu}
\begin{proof} For $\boldsymbol{\lambda}\in \mathbb{R}^n$, let us define the inner functional $V(\boldsymbol{\lambda})$ as 
  {\small  \begin{equation}
    \begin{aligned}   \hspace{-1mm}V(\boldsymbol{\lambda}) \hspace{-0.5mm}:= \hspace{-1.9mm}\sup_{\mu_x \in \tilde{\mathcal{A}}^\delta_{bounded}} \hspace{-2,2mm}&\mathbb{E}^{\mathbb{Q}} \Bigg[  e^{-\int_0^T (r_u+\mu_x(u))\mathrm{d}u}H(A_T,G_T^A) \hspace{-0.5mm}+\hspace{-0.5mm}\int_0^T \hspace{-1.5mm}e^{-\int_0^s (r_u+\mu_x(u)) \mathrm{d}u}\hspace{-0.5mm}\left(  \widetilde{H}(A_s,G_s^I) +\mu_x(s) \widehat{H}(A_s,G_s^D) \right)\hspace{-0.5mm}\mathrm{d}s \Bigg] \\
        &-\sum_{j=1}^n\lambda_j\E^{\mathbb{P}}\left(e^{-\int_0^{j}\mu_x(s)\mathrm{d}s}\right) \label{eq: inner_problem_proof}
        \end{aligned}
    \end{equation}} Since we assume that there exists $\bar{v}(t,a,r,z; \boldsymbol{\lambda})\in C^{1,2,2,2}$ solution to \eqref{eq:HJB_sup_bis_n} that satisfies boundary conditions given by \eqref{eq:boundaries_HJB}, the verification theorem for stochastic control problems (see, for instance, \cite{touzi2012optimal}, or by adapting the arguments used in the proof of \cite[Theorem 1]{li2011uncertain} to the context of an uncertainty mortality control problem) implies that $\bar{v}(t,a,r,z; \boldsymbol{\lambda})$ is the value function of the stochastic control problem \eqref{eq:sup_stoch_problem_path_regu_n_constraint}. In particular,
    \begin{align*}
    V(\boldsymbol{\lambda})=&\bar{v}(0,A_0,r_0,Z_0;\boldsymbol{\lambda}).
    \end{align*}
    Moreover, the optimal feedback control given by \eqref{eq:max_mu_n} is admissible, so that there exists $\bar\mu^*(\cdot;\boldsymbol{\lambda})\in \tilde{\mathcal{A}}^\delta_{bounded}$ attaining the supremum in \eqref{eq: inner_problem_proof}, i.e.
    {\small \begin{align*}
        V(\boldsymbol{\lambda})=&\mathbb{E}^{\mathbb{Q}} \Bigg[ \, e^{-\int_0^T (r_u+\bar\mu_x^*(u;\boldsymbol{\lambda}))\mathrm{d}u}H(A_T,G_T^A) + \hspace{-0.5mm}\int_0^T \hspace{-1.5mm}e^{-\int_0^s (r_u+\bar\mu_x^*(u;\boldsymbol{\lambda})) \mathrm{d}u}\hspace{-0.35mm}\left(  \widetilde{H}(A_s,G_s^I) +\bar\mu_x^*(s;\boldsymbol{\lambda}) \widehat{H}(A_s,G_s^D) \right)\hspace{-0.25mm}\mathrm{d}s \Bigg]\\
        &-\sum_{j=1}^n\lambda_j\E^{\mathbb{P}}\left(e^{-\int_0^{j}\bar\mu_x^*(s;\boldsymbol{\lambda})\mathrm{d}s}\right).
    \end{align*}}
    Let
    $$
    \Phi(\boldsymbol{\lambda})
    :=
    \bar v(0,A_0,r_0,Z_0;\boldsymbol{\lambda})
    +\sum_{j=1}^n\lambda_j\,{}_j\hat p_x.
    $$
        By the definition of the inner value function,
    
    $$
    \Phi(\boldsymbol{\lambda})
    =
    \sup_{\mu_x\in\widetilde{\mathcal A}^{\delta}_{\mathrm{bounded}}}
    \left\{
    J(\mu_x)
    -\sum_{j=1}^n\lambda_j
    \left[
    Q_j(\mu_x)-{}_j\hat p_x
    \right]
    \right\},
    $$
    where
    $
    Q_j(\mu_x)
    :=
    \mathbb E^{\mathbb P}
    \left[
    e^{-\int_0^j\mu_x(s)\,\mathrm{d}s}
    \right]
    $
    and \(J(\mu_x)\) denotes the financial-actuarial part of the objective function, i.e. the expectation under the risk-neutral measure $\mathbb{Q}$. Hence, \(\Phi\) is convex as the supremum of affine functions of \(\boldsymbol{\lambda}\).We now show that the dual objective function \(\Phi\) admits a minimizer. Since \(\mu_{x,a}\in\mathcal A\), we have \(Q_j(\mu_{x,a})={}_j\hat p_x\) for \(j=1,\ldots,n\). Consider multiplicative perturbations of the form
    \[
    \mu^\theta(t)=(1+\theta_k)\mu_{x,a}(t),
    \qquad t\in[k-1,k),\quad k=1,\ldots,n.
    \]
    Since \(\mu_{x,a}\) is bounded, for sufficiently small \(\boldsymbol\theta\in\mathbb R^n\), these perturbations satisfy the constraints defining \(\widetilde{\mathcal A}^{\delta}_{\mathrm{bounded}}\). Moreover,
    \[
    Q_j(\mu^\theta)
    ={}_j\hat p_x
    \exp\left(
    -\sum_{k=1}^j\theta_k
    \int_{k-1}^k\mu_{x,a}(s)\,\mathrm{d}s
    \right),
    \qquad j=1,\ldots,n.
    \]
    Define \(F:\mathbb R^n\to\mathbb R^n\) by
    \[
    F(\boldsymbol\theta)
    :=\bigl(Q_1(\mu^\theta),\ldots,Q_n(\mu^\theta)\bigr)^\top.
    \]
    The Jacobian matrix \(DF(\boldsymbol 0)\) is lower triangular, with strictly negative diagonal entries and thus invertible. By the inverse function theorem, there exists \(\rho>0\) such that every vector \(\boldsymbol q\) satisfying $\|\boldsymbol q-\hat{\boldsymbol p}_x\|<\rho$ can be written as \(\boldsymbol q=F(\boldsymbol\theta)\) for some sufficiently small \(\boldsymbol\theta\) such that \(\mu^\theta\in\widetilde{\mathcal A}^{\delta}_{\mathrm{bounded}}\). Hence, for any \(0<\varepsilon<\rho\) and every \(\boldsymbol\lambda\neq\boldsymbol 0\), there exists an admissible control \(\mu\in \widetilde{\mathcal A}^{\delta}_{\mathrm{bounded}}\) such that
    \[
    Q_j(\mu)
    ={}_j\hat p_x-\varepsilon
    \frac{\lambda_j}{\|\boldsymbol\lambda\|},
    \qquad j=1,\ldots,n.
    \]
    Since \(J(\mu)\geq0\)\footnote{As $H,\tilde{H}$ and $\hat{H}$ are assumed to be positive.}, we obtain
    \[
    \begin{aligned}
    \Phi(\boldsymbol\lambda)
    &\geq J(\mu)+\sum_{j=1}^n\lambda_j
    \bigl({}_j\hat p_x-Q_j(\mu)\bigr) \geq \varepsilon\|\boldsymbol\lambda\|.
    \end{aligned}
    \]
    Therefore, \(\Phi\) is coercive. Moreover, since \(\Phi\) is finite-valued and convex on \(\mathbb R^n\), it is continuous on \(\mathbb R^n\). Hence, \(\Phi\) attains its minimum at some \(\boldsymbol\lambda^*\in\mathbb R^n\). Since \(\Phi\) is convex, admits a minimizer \(\boldsymbol\lambda^*\) and is differentiable at $\boldsymbol{\lambda}^*$, the first-order condition (FOC) yields
    \[
    \nabla\Phi(\boldsymbol\lambda^*)=\boldsymbol 0.
    \]
    Moreover, for every \(j=1,\ldots,n\),
    \[
    \frac{\partial\Phi}{\partial\lambda_j}(\boldsymbol\lambda^*)
    ={}_j\hat p_x-Q_j\bigl(\bar\mu^*(\cdot;\boldsymbol{\lambda})\bigr).
    \]
    Therefore, the FOC implies
    \[
    Q_j\bigl(\bar\mu^*(\cdot;\boldsymbol{\lambda})\bigr)
    ={}_j\hat p_x,\qquad j=1,\ldots,n,
    \]
    which shows that the optimal control \(\bar\mu^*(\cdot;\boldsymbol{\lambda})\in \mathcal{A}^\delta_{bounded}\). Finally, the strong duality follows by observing 
        {\small
        \begin{align*}
            &\sup_{\mu_x \in \mathcal{A}^\delta_{bounded}} \mathbb{E}^{\mathbb{Q}} \Bigg[ \, e^{-\int_0^T (r_u+\mu_x(u))\mathrm{d}u}H(A_T,G_T^A) +\int_0^T e^{-\int_0^s (r_u+\mu_x(u)) \mathrm{d}u}\left(  \widetilde{H}(A_s,G_s^I) +\mu_x(s) \widehat{H}(A_s,G_s^D)     \right)\mathrm{d}s \Bigg]\\
            &\leq \inf_{\boldsymbol{\lambda}\in \mathbb{R}^n} \hspace{-0.5mm}\sup_{\mu \in \tilde{\mathcal{A}}^\delta_{bounded}} \hspace{-1.8mm}\mathbb{E}^{\mathbb{Q}} \Bigg[  e^{-\int_0^T (r_u+\mu_x(u))\mathrm{d}u}H(A_T,G_T^A) +\hspace{-0.5mm}\int_0^T \hspace{-1.25mm}e^{-\int_0^s (r_u+\mu_x(u)) \mathrm{d}u}\hspace{-0.4mm}\left(  \widetilde{H}(A_s,G_s^I) +\mu_x(s) \widehat{H}(A_s,G_s^D) \right) \hspace{-0.4mm}\mathrm{d}s \Bigg]\\
            &~~~~~~~~~~~~~~~~~~~~~~~~~~~-\sum_{j=1}^n\lambda_j \left[\E^{\mathbb{P}}\left(e^{-\int_0^{j}\mu_x(s)\mathrm{d}s}\right)-{}_j\hat{p}_{x}\right] \\
            &=\min_{\boldsymbol{\lambda}\in \mathbb{R}^n} \left[\bar{v}(0,A_0,r_0,Z_0; \boldsymbol{\lambda})+\sum_{j=1}^n\lambda_j \, {}_j\hat{p}_{x}\right]\\
            &=  \mathbb{E}^{\mathbb{Q}} \Bigg[ \, e^{-\int_0^T (r_u+\bar{\mu}^*(u;\boldsymbol{\lambda}^*))\mathrm{d}u}H(A_T,G_T^A)+\int_0^T e^{-\int_0^s (r_u+\bar{\mu}^*(u;\boldsymbol{\lambda}^*)) \mathrm{d}u}\left(  \widetilde{H}(A_s,G_s^I) +\bar{\mu}^*(s;\boldsymbol{\lambda}^*) \widehat{H}(A_s,G_s^D)     \right)\mathrm{d}s \Bigg]
        \end{align*}}with $\mu_x^*(.;\boldsymbol{\lambda}^*)\in\mathcal{A}^\delta_{\mathrm{bounded}}$, which lead to \eqref{eq: strong_duality}.
\end{proof}

\subsection{Proof of Proposition \ref{prop: convergence_delta}}\label{proof: conv_delta}
\begin{proof}
    To prove the statement, we show that \begin{equation*} \mathcal{A} = \bigcup_{\delta>0} \mathcal{A}^\delta_{{bounded}}. \end{equation*} The inclusion \begin{equation*} \bigcup_{\delta>0} \mathcal{A}^\delta_{{bounded}} \subseteq \mathcal{A} \end{equation*} is immediate from the definition. Conversely, let \[ \mu_x\in\mathcal{A}. \]  Since \(\mu_{x,a}\in\mathcal{A}\), we have \[ C_x:= \operatorname*{ess\,sup}_{\omega\in\Omega} \sup_{t\in[0,T]}|\mu_x(t,\omega)| <\infty \] and \[ C_a:= \operatorname*{ess\,sup}_{\omega\in\Omega} \sup_{t\in[0,T]}|\mu_{x,a}(t,\omega)| <\infty. \] Set $\delta:=C_x+C_a<\infty$. Then, for almost every $\omega\in\Omega$ and for all $t\in[0,T]$, \[ \left|\mu_x(t,\omega)-\mu_{x,a}(t,\omega)\right| \leq |\mu_x(t,\omega)|+|\mu_{x,a}(t,\omega)| \leq C_x+C_a =\delta. \] Hence, \[ \mu_{x,a}(t,\omega)-\delta \leq \mu_x(t,\omega) \leq \mu_{x,a}(t,\omega)+\delta. \] Since $\mu_x(t,\omega)\geq0$, it follows that \[ (\mu_{x,a}(t,\omega)-\delta)_+ \leq \mu_x(t,\omega) \leq \mu_{x,a}(t,\omega)+\delta \] for all $t\in[0,T]$, almost surely. Therefore, $\mu_x\in\mathcal{A}^\delta_{{bounded}}$. Since $\mu_x\in\mathcal{A}$ was arbitrary, we obtain \[ \mathcal{A} \subseteq \bigcup_{\delta>0} \mathcal{A}^\delta_{{bounded}}, \] 
    and it concludes the proof. \vspace{-2mm}
\end{proof}

{\small \textbf{Disclosure statement} \
The authors declare that they have no conflict of interest. AI tools were used to correct typographical errors, improve the writing and assist in establishing Proposition \ref{propgmdb}. \\
\textbf{Funding details} \ This work was supported by the Swiss National Science Foundation (SNSF) under Grant no. 10003723.  Edouard Motte is a FRIA grantee of the Fonds de la Recherche Scientifique - FNRS, Belgium.\vspace{-2mm} 
{\footnotesize
		\bibliographystyle{apa}  
		\bibliography{bibbounds}   
		
	}

\end{document}